\documentclass[lettersize,journal]{IEEEtran}
\ifCLASSINFOpdf
\else
\usepackage[dvips]{graphicx}
\fi
\usepackage{url}
\usepackage[hidelinks]{hyperref}
\hypersetup{
	colorlinks   = true, %Colours links instead of ugly boxes
	urlcolor     = blue, %Colour for external hyperlinks
	linkcolor    = blue, %Colour of internal links
	filecolor	    = blue, % Color of file links
	citecolor   = {blue}, %Colour of citations
}

\usepackage{verbatim}
\usepackage{graphicx}
\usepackage{amsmath,amsfonts,amssymb,amsthm,array}
\usepackage{dsfont}
\usepackage{textcomp}
\usepackage{stfloats}
\usepackage{cite}
\usepackage[inline,shortlabels]{enumitem}
\usepackage[font=footnotesize,labelfont=bf]{caption}
\usepackage{subcaption}
\usepackage{mathtools}
\usepackage{multirow}
\newtheorem{lem}{Lemma}
\newtheorem{prop}{Proposition}
\newtheorem{cor}{Corollary}
\theoremstyle{definition}

\theoremstyle{plain}
\newtheorem{rem}{Remark}

\usepackage{float}
\usepackage{algorithmic}[2]
\usepackage[linesnumbered,boxed,ruled,vlined]{algorithm2e}
\usepackage{setspace}
\SetKwInOut{KwIn}{Input}
\SetKwInOut{KwOut}{Output}
\SetKwInOut{KwInt}{Initialize}

\SetCommentSty{mycommfont}
\usepackage{array}
\usepackage[table,xcdraw]{xcolor}
\usepackage{accents}
\usepackage{placeins}
\def\BibTeX{{\rm B\kern-.05em{\sc i\kern-.025em b}\kern-.08em
		T\kern-.1667em\lower.7ex\hbox{E}\kern-.125emX}}

\usepackage{balance}

\usepackage{xr}
\usepackage{cleveref}
\usepackage{flushend}

\allowdisplaybreaks

\begin{document}
\title{Half-Duplex APs with Dynamic TDD vs. Full-Duplex APs in Cell-Free Systems}
\author{Anubhab Chowdhury and Chandra R. Murthy, \IEEEmembership{Fellow, IEEE}	
\thanks{Financial support for this work from the Qualcomm Innovation Fellowship is gratefully acknowledged. A part of this work will be presented at the IEEE ICASSP in Apr. 2024. (\textit{Corresponding author: Chandra R. Murthy}.)}
\thanks{The authors are with the ECE Dept., Indian Institute of Science, Bangalore, 560012, India. (e-mails:\{anubhabc, cmurthy\}@iisc.ac.in.)}
}

\maketitle
\IEEEpeerreviewmaketitle
\begin{abstract}
In this paper, we present a comparative study of half-duplex (HD) access points (APs) with dynamic time-division duplex (DTDD) and full-duplex (FD) APs in cell-free (CF) systems. Although both DTDD and FD CF systems support concurrent downlink~(DL) transmission and uplink~(UL) reception capability, the sum spectral efficiency (SE) is limited by various cross-link interferences. We first present a novel pilot allocation scheme that minimizes the pilot length required to ensure no pilot contamination among the user equipments~(UEs) served by at least one common AP. Then, we derive the sum SE in closed form, considering zero-forcing combining and precoding along with the signal-to-interference plus noise ratio optimal weighting at the central processing unit. We also present a provably convergent algorithm for joint UL-DL power allocation and UL/DL mode scheduling of the APs (for DTDD) to maximize the sum SE. Further, the proposed algorithms are precoder and combiner agnostic and come with closed-form update equations for the UL and DL power control coefficients. Our numerical results illustrate the superiority of the proposed pilot allocation and power control algorithms over several benchmark schemes and show that the sum SE with DTDD can outperform an FD CF system with similar antenna density. Thus, DTDD combined with CF is a promising alternative to FD that attains the same performance using HD APs, while obviating the burden of intra-AP interference cancellation.
\end{abstract}

\section{Introduction}
Wireless systems for $5$G and beyond are required to serve an increasingly large number of user-equipments~(UEs) while supporting uniformly good quality-of-service~(QoS) and high spectral efficiency~(SE) requirements. With this in mind, two potential physical layer solutions have been envisioned: $(i)$ the capability to serve both uplink~(UL) and the downlink~(DL) UEs using the same time-frequency resources, thereby potentially doubling the sum UL-DL SE over conventional time/frequency division duplexing~(T/FDD) systems, and $(ii)$ distributed deployments of the remote radio units or access points~(APs) for ubiquitous connectivity and high macro-diversity as opposed to a centralized MIMO cellular system. The key enablers of the above are: $(i)$ the use of full-duplex~(FD) APs~\cite{Ashutosh_1} or the use of dynamic time-division duplexing~(DTDD) with half-duplex (HD) APs~\cite{DTDD_CF_Giuseppe,DTDD_TCoM,Hyejin_CST_2020}; and $(ii)$ cell-free multiple-input multiple-output~(CF-MIMO)~\cite{zheng2023mobile, J_Zheng_Asynch_CF, making_cellfree,Ubiquitous}; respectively. This paper presents a comparative performance analysis of FD and DTDD CF-MIMO\footnote{DTDD CF and FD CF are also referred to as network-assisted full-duplex~(NAFD) CF MIMO  in the literature~\cite{NAFD_Xia, Joint_UE_NAFD}.} systems. 

We note that, in CF systems, both FD and DTDD can enable simultaneous transmission (reception) to (from) HD UEs over the same time-frequency resources. However, cross-link interferences~(CLIs), i.e., the inter-AP interference~(InAI) and inter-UE interference~(InUI), limit the achievable SE in both systems. Additionally, in the FD system, the received signal at each AP is contaminated by its own transmitted signal, called \emph{self-interference~(SI)}~\cite{Ashutosh_1},  which we refer to as intra-AP interference~(IrAI) in the sequel.  The cancellation of  IrAI demands power-hungry and expensive hardware in addition to baseband signal processing overheads. However, if IrAI can be effectively canceled, in an FD-CF system, \emph{all} the APs in the vicinity can assist in a given UE's transmission/reception, while in the DTDD system, only the subset of APs operating in UL (DL) can assist in decoding (precoding) the UE's data signal.
On the other hand, as we will see, InAI is higher in an FD system, as \emph{all} the APs interfere with the received signal at any AP. However, with DTDD, only the subset of APs operating in DL mode cause InAI at the APs operating in UL mode. 
Thus, which duplexing scheme is better and under what conditions is not clear. Answering this question via a careful theoretical analysis is the main aim of this paper. 

\subsubsection{Literature review}
CF, DTDD, and FD are all technologies that have received intense research attention over the past few years. The focus of this literature review is on studies that consider either DTDD or FD systems in the CF architecture.
The analysis of the sum SE of a DTDD CF system with fully centralized processing was presented in~\cite{Pilot_Cont_DTDD}; however, with a pre-scheduled set of UL and DL APs. Further, in~\cite{NAFD_Xia, xia2021joint},
	the authors addressed the AP-mode selection for DTDD-based CF-systems to maximize the UL/DL SEs, with power constraints at the UEs and APs. In~\cite{D_Wang_System_J}, 
	the authors considered a hybrid-duplex~(both FD and HD APs) architecture, where antenna mode assignment at each multi-antenna AP is solved under the goal of secrecy SE maximization~\cite{D_Wang_System_J}. Here, we note that~\cite{D_Wang_System_J, xia2021joint,  NAFD_Xia}
	assume the availability of perfect channel state information~(CSI) \emph{at the central processing unit~(CPU),} rather than the statistical/estimated CSI. Therefore, all resource allocation needs to be executed in the time scale of fast-fading coefficients. Moreover, the authors assumed a fully centralized CPU-based CF system. Thus, the APs must send the full CSI to the CPU, incurring high front-haul overhead, especially when the numbers of APs and UEs are large, which is not scalable~\cite{Emil_CF_Book, making_cellfree}. In contrast, consideration of distributed processing, where precoding/combining is performed locally at the APs, is critical for scalable CF systems.

In~\cite{DTDD_TCoM}, we presented a distributed processing-based sum SE analysis with a low complexity algorithm for AP scheduling and fixed UL and DL power allocation. This work significantly extends our previous work \cite{DTDD_TCoM} by considering minimum mean squared error (MMSE)/zero-forcing (ZF) precoding and combining at the APs and signal-to-interference plus noise ratio (SINR)-optimal combining at the CPU, pilot allocation via graph coloring, joint UL-DL power allocation for sum SE maximization, and comparing the performance against an FD system.
Recently, the authors in~\cite{SPAWC_Virtual_Duplex} investigated the joint UL-DL power allocation problem, assuming perfect CSI, in a DTDD CF system. The power allocation policy of~\cite{SPAWC_Virtual_Duplex} is derived for a maximal ratio combiner and precoder~(MRC and MFP), with equal weighting-based combining of the APs' signals at the CPU in the UL, which is suboptimal for CF~\cite{making_cellfree}. Recently, the authors extended their work in~\cite{SPAWC_Virtual_Duplex} incorporating large-scale-fading-decoding~(LSFD) for the UL combining at the CPU~\cite{Mohammadi_JSAC}. However, the analyses in~\cite{Mohammadi_JSAC} assume the availability of orthogonal pilots for channel estimation and MRC/MFP for UL and DL data detection. Former can lead to inordinately high pilot overhead, and MRC/MFP are suboptimal choices for CF-system~\cite{making_cellfree}. Thus, it is essential to consider  ZF or MMSE-type combiners and precoders that are more robust to interference, which we address in this work.

On the other hand,  the performance improvement achieved by FD APs over conventional TDD CF systems has been investigated in~\cite{FD_CF_ICC, FD_CF_JSAC, FD_CF_TVT}. An expository study on the interplay of the CLIs and IrAI on the achievable SEs under DTDD and FD was conducted in~\cite{NCC_anubhab}. However, the results obtained were based on fixed power allocation and under the availability of perfect CSI. Now, finding a pilot allocation scheme in a CF system is challenging because multiple APs jointly serve the UEs in the area. Specifically, in contrast to a cellular system where only the serving AP requires CSI from a given UE, in CF, accurate CSI is required at all the APs in the vicinity of the UE. 
The works that account for pilot allocation in DTDD CF systems consider either complex iterative algorithm~\cite{DTDD_TCoM}, random allocation of pilots~\cite{Pilot_Cont_DTDD}, or assume orthogonal pilots across all UEs~\cite{SPAWC_Virtual_Duplex}. Recent works on FD CF systems have also considered orthogonal pilots~\cite{FD_CF_ICC, FD_imperfect_CSI, Joint_UE_NAFD} or have analyzed the effect of pilot contamination by abstracting it as an additive channel error term in the channel estimate~\cite{FD_CF_TVT}, which does not explicitly account for the pilot length, design or allocation across the UEs. Here, in contrast to pilot allocation with predetermined pilot length~\cite{DTDD_TCoM, NCC_anubhab, Lozano_Pilot, Heng_Liu_TVT, Pilot_power, Location_Pilot, Cluster, ICC_pilot, cell_free_small_cells}, we optimize the pilot length, and develop a low-complexity solution that incurs little signal processing overhead, and is applicable in both DTDD and FD settings.

\subsubsection{Contributions}
To the best of our knowledge, a comprehensive study of these two duplexing schemes accounting for practical issues such as pilot length optimization, optimal weighting at the CPU, interference canceling precoder and combiner design, AP scheduling, and UL/DL power allocation is not available in the literature. It is essential to account for these aspects because the critical bottleneck, namely, the CLIs of DTDD and FD, is heavily dependent on and can be controlled by these factors. In this regard, our analysis accounts for pilot allocation with limited orthogonal resources, and, consequently, includes the effects of coherent interference on the optimal weights and sum UL-DL SE, which was missing in previous work~\cite{Mohammadi_JSAC}. Further, we provide closed-form expressions for the SE with ZF combiner and precoder and, subsequently, closed-form updates for all the power allocation algorithms developed. Our algorithms need to be executed only in the time scale of large-scale fading, which remains constant for several channel coherence intervals, in contrast to instantaneous CSI-based approaches in~\cite{D_Wang_System_J, xia2021joint,  NAFD_Xia}.
Our key contributions are:
\begin{enumerate}[label=\arabic*.]	
\item We minimize the number of pilots required to ensure orthogonality among the UEs in close proximity via formulating an equivalent graph coloring problem, with the constraint that connected vertices~(i.e., the UEs which are connected to a common AP) are allotted distinct colors~(orthogonal pilots). Although the problem is NP-hard in general, it can be optimally solved for bipartite graphs (as is the case in our problem) 
via a low complexity greedy algorithm~(see Algorithm~\ref{algo:pilot_allocation}), thereby minimizing the number of colors~(i.e., the pilot length). We empirically show that Algorithm~\ref{algo:pilot_allocation} offers a substantial improvement in the normalized mean square error~(NMSE) of the estimated channels as well as the sum UL-DL SE compared to several existing greedy/iterative methods (see Fig.~\ref{fig:SNR_p_vs_NMSE} and Fig.~\ref{fig:fig_CDF_pilot_allocations}).
	
	\item We analyze the sum UL-DL SE considering MMSE combiners and regularized ZF~(RZF) precoders. We also derive closed-form expressions for the sum UL-DL SE with ZF combiners and precoders~(see Lemma~\ref{lemm:UL_SINR_DTDD} and Lemma~\ref{lem:DL_SE_DTDD}). These expressions uncover the effects of InAI, IrAI, and InUI on the UL-DL SEs, and how power control and  UL/DL scheduling of the APs~(for DTDD) dictate the strengths of these CLIs. Also, in the UL, we present an SINR optimal weighting scheme, which ensures that the received SINR at the CPU is maximized~(see Lemma~\ref{lem:optimal_weights}). 
	\item Next, we focus on the sum UL-DL SE maximization with set constraints on the UL/DL APs and transmit power constraints on the APs and UEs. This problem of joint AP scheduling and power control is non-convex and NP-hard. We decouple it into two sub-problems. 
	\begin{enumerate}[label=\roman*.,topsep=0pt,nosep]
		\item We optimize the UL and DL power control coefficients for a given AP schedule. We solve this non-convex problem using fractional programming~(FP)\footnote{FP convexifies the non-convex cost function such that the optimal solution of the surrogate cost function and the original cost function is the same~\cite{FP_I}. This is in contrast with the approach adopted in~\cite{SPAWC_Virtual_Duplex, FD_CF_JSAC}, where the convex cost function is typically a lower bound of the original cost function and the algorithms optimize the lower bound.}  that employs a series of equivalent convex reformulations. Further, closed-form solutions for the power allocation coefficients and associated auxiliary variables are derived using the alternating direction method of multipliers~(ADMM)\footnote{ADMM is an effective approach for reducing the computational cost in large dimensional problems compared to interior-point methods or general purpose solvers such as CVX or MOSEK~\cite{Boyd_ADMM}.} in the case of DL and using an augmented Lagrange multiplier in the case of UL.  Also, in the UL, we observe considerable improvement in the proposed SINR optimal combining coupled with UL power control compared to existing benchmarks and when either of these two schemes is applied individually~(see Fig.~\ref{fig:UL_power_control}). The resulting algorithms for each sub-problem are shown to converge to local optima~(see Proposition~\ref{prop:conv_UL} and Proposition~\ref{prop:DL_power_control}). Finally, our proposed FP-based algorithms are precoder/combiner scheme agnostic, unlike~\cite{SPAWC_Virtual_Duplex, Mohammadi_JSAC}, and require fewer auxiliary variables, which makes our solutions widely applicable and scalable for large distributed systems. 
		
		\item For AP scheduling, we develop a greedy AP mode~(UL/DL) selection algorithm, where, at each iteration, we select the AP and the corresponding mode such that the incremental gain in the sum UL-DL SE is maximum. This pragmatic low-complexity approach solves an otherwise exponentially complex scheduling algorithm in polynomial time.
	\end{enumerate}
\end{enumerate}

We perform extensive numerical experiments that reveal the superiority of the proposed pilot length optimization and pilot allocation scheme, the UL/DL power control algorithms, and AP-scheduling algorithm over several existing schemes~(see Fig.~\ref{fig:UL_power_control} and Fig.~\ref{fig:DL_power_control}). Surprisingly, our results show that for the same number of APs and antenna density, DTDD procures a better sum UL-DL SE compared to an FD-enabled CF system~(see Fig.~\ref{fig:figs_DTDD_FD}). Specifically, the $90\%$-likely sum UL-DL SE of the DTDD CF system is $21\%$ higher than that of the FD system under similar system parameters~(see Fig.~\ref{fig:CDF_DTDD_FD}). Further, we observe that even with double the antenna density, the performance of the FD system can be limited by InAI and IrAI, while DTDD is more resilient to InAI~(see Fig.~\ref{fig:InAI}). Thus, we can obtain the benefits of FD via DTDD itself, obviating the need for IrAI suppression at the APs.

\emph{Notation:} 
The transpose, hermitian, complex conjugation, and trace operations are denoted by $(\cdot)^T$, $(\cdot)^H$, $(\cdot)^*$, and ${\tt tr}(\cdot)$, respectively. For a matrix $\mathbf{A}\in\mathbb{C}^{M\times N}$, $[\mathbf{A}]_{:,j}\in\mathbb{C}^{M\times 1}$ denotes its $j$th column. For a vector $\mathbf{x}\in\mathbb{C}^{N}$, $[\mathbf{x}]_{m}$ denotes its $m$th entry. $\mathsf{diag}(\mathbf{x})\in\mathbb{C}^{N\times N}$ is a diagonal matrix with $m$th diagonal entry being $[\mathbf{x}]_{m}$.
$'| \cdot |'$, $'\backslash'$, $'\cup'$, and $'\mathsf{c}'$ denote the cardinality, set difference, union, and complement of sets, respectively. The symbol $\emptyset$ denotes an empty set. 
$\mathbb{E}[\cdot]$ and ${\tt var}\{\cdot\}$ denote the mean and variance of a random variable/vector, respectively. $\mathbf{x}\sim\mathcal{CN}(\mathbf{0}_{N},\mathbf{R}_{N})$ indicates that $\mathbf{x}\in\mathbb{C}^{N}$ is a zero mean~($\mathbf{0}_{N}$) circularly symmetric complex Gaussian (CSCG) random vector with covariance matrix $\mathbf{R}_N\in\mathbb{C}^{N\times N}$.
\section{System Model}
In the DTDD CF setup,  $M$ HD-APs, each equipped with $N$ antennas, jointly and coherently serve a total of $K$ single antenna UL and DL UEs using the same time-frequency resources. Let the sets $\mathcal{U}_{\mathsf{u}}$ and $\mathcal{U}_{\mathsf{d}}$ contain the indices of UL UEs and DL UEs, respectively, with $\mathcal{U}_{\mathsf{u}}\cap\mathcal{U}_{\mathsf{d}}=\emptyset$, $\mathcal{U}_{\mathsf{u}}\cup\mathcal{U}_{\mathsf{d}}=\mathcal{U}$, and $\lvert\mathcal{U}\rvert=K$.  
The UL channel from the $k$th UE to the $m$th AP is modeled as $\mathbf{f}_{mk} =\sqrt{\beta_{mk}}\mathbf{h}_{mk} \in \mathbb{C}^N$, where $\beta_{mk} > 0$ captures the effect of large scale fading and path-loss which remain unchanged over several channel coherence intervals and are known to the APs and the CPU~\cite{making_cellfree}. The fast fading component, $\mathbf{h}_{mk}  \sim \mathcal{CN}(\mathbf{0}_{N},\mathbf{I}_N) \in \mathbb{C}^N$, is independent and identically distributed~(i.i.d.) and is estimated at the APs using pilot signals at the beginning of each coherence block. 

\begin{figure*}
	\centering
\begin{subfigure}{0.35\linewidth}
	\centering
	\includegraphics[width=\textwidth]{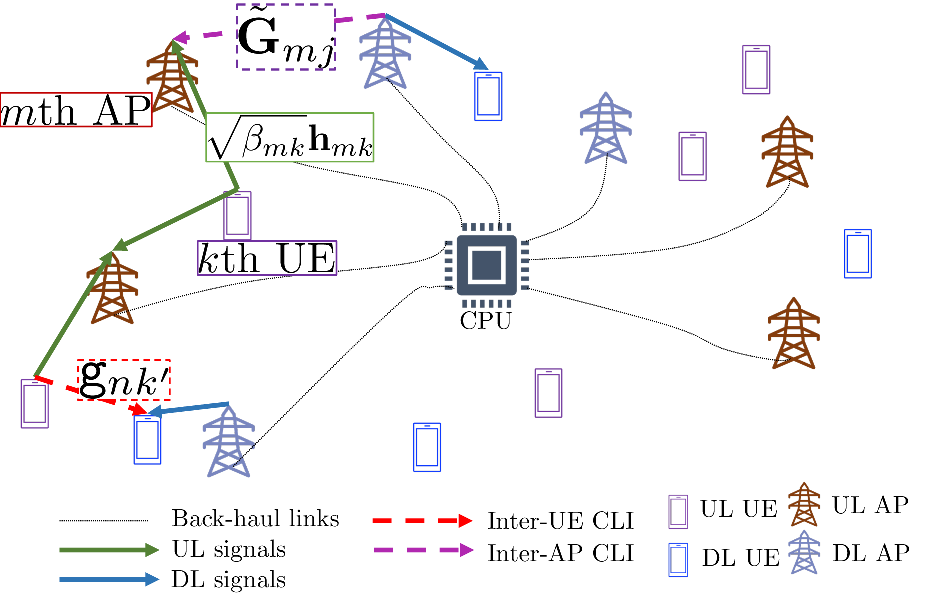}
	\caption{DTDD CF MIMO system: the overall system can serve UL and DL UEs simultaneously, forming a \emph{virtual} FD system.}\label{fig:system_model_DTDD}
	\end{subfigure}\hfill
\begin{subfigure}{0.35\linewidth}
	\centering
	\includegraphics[width=\textwidth]{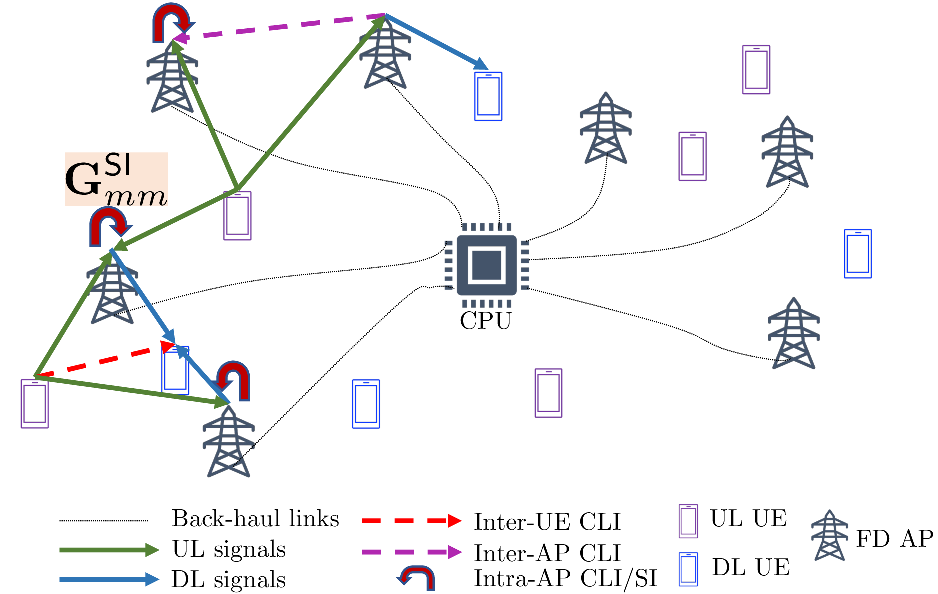}
	\caption{FD-enabled CF system: each AP can serve both UL and DL UEs; however, the APs suffer from IrAI. }\label{fig:system_model_FD}
\end{subfigure}\hfill
\begin{subfigure}{0.28\linewidth}
		\centering
		\includegraphics[width=\textwidth]{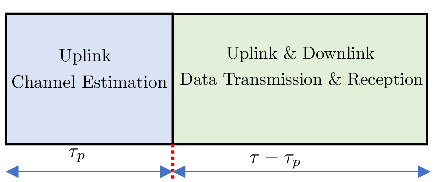}
		\caption{Frame structure for DTDD and FD. This allows us to use a unified pilot allocation strategy for both the duplexing schemes.}\label{fig:frame_structure}
	\end{subfigure}
\caption{The system models and frame structure for DTDD and FD CF systems, respectively.}
\end{figure*}

In the FD system, each AP is equipped with $N_{\mathsf{tx}}$ transmit and $N_{\mathsf{rx}}$ receive antennas. Let $\mathbf{f}_{\mathsf{u}, mk}=\sqrt{\beta_{\mathsf{u}, mk}}\mathbf{h}_{\mathsf{u}, mk}\in\mathbb{C}^{N_{\mathsf{rx}}}$ and $\mathbf{f}_{\mathsf{d}, mn}=\sqrt{\beta_{\mathsf{d}, mn}}\mathbf{h}_{\mathsf{d}, mn}\in\mathbb{C}^{N_{\mathsf{tx}}}$ be the UL channel between the receive antenna array of the $m$th AP to $k$th UL UE and the DL channel between the transmit antenna array of the $m$th AP to the $n$th DL UE, respectively. Here, $\beta_{\mathsf{i}, mk}$ and $\mathbf{h}_{\mathsf{i}, mk}$, $\mathsf{i}=\{\mathsf{u},\mathsf{d}\}$, corresponds to the fast and slow fading components of the UL and the DL channels, respectively, and follow similar statistical modeling as in the DTDD system.

The inter-AP channels remain constant for several coherence intervals and are mostly line-of-sight~(LoS), whose CSI can be made available to the CPU before the pilot and data transmission phase. Thus, the overhead of inter-AP channel estimation does not affect the available data transmission duration. However, the inter-AP CSI at the CPU may be erroneous, which we model using Gaussian distributed additive noise as per~\cite{Ashutosh_1,DTDD_TCoM,FD_CF_ICC}. Specifically, the inter-AP channel from the $j$th DL AP to the $m$th UL AP is denoted by $\tilde{\mathbf{{G}}}_{mj}\in \mathbb{C}^{N\times N}$ whose elements are i.i.d. $\mathcal{CN}(0,\zeta_{mj}^{\mathsf{InAP}})$. Here $\zeta_{mj}^{\mathsf{InAP}}$ captures the effects of both the large scale fading and the power of the residual InAI.

For FD CF,  IrAI at each AP can be suppressed via active SI cancelation and antenna isolation. It has been argued that the residual IrAI follows the Rayleigh distribution~\cite{Riihonen_loopback}. Following this, we model $m$th AP's residual IrAI channel between the transmit and the receive antenna links as $\mathbf{G}_{mm}^{\mathsf{SI}}\sim\mathcal{CN}(0,\zeta_{mm}^{\mathsf{SI}})$, where the residual IrAI $\zeta_{mm}^{\mathsf{SI}}$ depends on the interference suppression capability of the hardware\cite{FD_CF_TVT,FD_imperfect_CSI}.

Finally, let $\mathtt{g}_{nk}$ denote the channel between $k$th UL UE and the $n$th DL UE, modeled as $\mathcal{CN}(0,\epsilon_{nk})$, and is independent across all UEs~\cite{Ashutosh_1,DTDD_TCoM}. The channel modeling discussed above is illustrated in Fig.~\ref{fig:system_model_DTDD} and Fig.~\ref{fig:system_model_FD}.

\subsection{Problem Statement} Having described the system model, we now present the key problems considered in this paper. As illustrated in Fig.~\ref{fig:system_model_DTDD} and Fig.~\ref{fig:system_model_FD}, the performance of both DTDD and FD is affected by CLIs, viz. InAI, IrAI (for FD), and InUI. Now, the strengths of these CLIs, as we shall see in Sec.~\ref{sec:DTDD_SE},  depend on the estimated UL/DL channel statistics, choice of precoder and combiners, and, most critically, on the UL and DL power allocation strategies. Further, in DTDD, we get the additional flexibility to schedule the APs' UL/DL modes, which can reduce InAI. Keeping these in mind, our goal is to maximize the achievable sum UL-DL SE under these two duplexing schemes.\footnote{In~\cite{Ubiquitous}, the authors point out that since a distributed  CF system is inherently fair, additionally requiring fairness degrades the sum SE without an appreciable improvement in the fairness criterion. Thus, we consider the sum UL-DL SE as the metric to be maximized.} In this regard, we pose and address the following problems: $(i)$ Allocation of pilot signals that can ensure no contamination in the APs in the vicinity of every UE while using a minimum number of orthogonal pilots; $(ii)$ design of combiners and precoders at the APs and the CPU to reduce interference; $(iii)$ AP scheduling in the UL/DL modes for DTDD; and $(iv)$ UL/DL power allocation policies for DTDD and FD to maximize the sum UL-DL SE. We begin with channel estimation and pilot allocation in the next section.
\begin{rem}
Typically, the CPU consists of multiple cores with a multi-threaded software architecture capable of processing UL and DL data simultaneously for DTDD and FD. Essentially, the DL threads are run on a subset of the cores at the CPU, and the UL threads are run using a different set of cores.
\end{rem}

\section{Channel Estimation and Pilot Allocation}
\subsubsection{DTDD CF}
	During the channel estimation phase, the UL and DL UEs synchronously send UL pilot sequences to the APs. 
	The APs use the received pilot signals to estimate the channels between the UEs and the APs. Now, allocating orthogonal pilot sequences to all the UEs can incur inordinately high channel estimation overhead. Therefore, we consider that pilot length is constrained to be $
\tau_{p}$, where $\tau_{p}\leq K$, and these pilot sequences are reused among the UEs. Let $\mathcal{P}=\{\boldsymbol{\varphi}_1,\boldsymbol{\varphi}_2,\ldots,\boldsymbol{\varphi}_{\tau_p}\}$ be the set of orthonormal pilot sequences, where $\boldsymbol{\varphi}_l\in\mathbb{C}^{\tau_p}$. Let $l(k)$ denote the index of the pilot used by the $k$th UE and  $\mathcal{P}_{l(k)}$ denote the set of UE indices, including the $k$th UE, that use $\boldsymbol{\varphi}_{l(k)}$. Therefore, $\langle \boldsymbol{\varphi}_{l(k)}, \boldsymbol{\varphi}_{l(k')}\rangle = 1$, if $k'\in\mathcal{P}_{l(k)}$, and equals to $0$ if $k'\notin\mathcal{P}_{l(k)}$. Let $\mathcal{E}_{p,k}$ be the power of the pilot signal of the $k$th UE. Then, the received signal matrix at the $m$th AP becomes  $\mathbf{Y}_{p,m}=\sum\nolimits_{k=1}^{K}\sqrt{\tau_{p}\mathcal{E}_{p,k}}\mathbf{f}_{mk}\boldsymbol{\varphi}_{l(k)}^{T}+\mathbf{W}_{p,m}\in\mathbb{C}^{N\times \tau_{p}},$
where $\mathbf{W}_{p,m}$ is the receiver noise matrix with i.i.d. entires distributed as $\mathcal{CN}(0, N_{0})$. Post-multiplying $\mathbf{Y}_{p,m}$ by $\boldsymbol{\varphi}_{l(k)}^{*}$, the processed signal at the $m$th AP becomes $\mathbf{y}_{p,m}=\sqrt{\mathcal{E}_{p,k}\tau_{p}}\mathbf{f}_{mk}+\sum\nolimits_{n\in\mathcal{P}_{l(k)}\backslash k}\sqrt{\mathcal{E}_{p,n}\tau_{p}}\mathbf{f}_{mn}+\mathbf{w}_{p,m},$ with $\mathbf{w}_{p,m}=\mathbf{W}_{p,m}\boldsymbol{\varphi}_{l(k)}^{*}\sim\mathcal{CN}(\mathbf{0}_{N},N_{0}\mathbf{I}_N)$. We can evaluate the MMSE estimate estimate of the channel $\mathbf{f}_{mk}$, denoted by $\hat{\mathbf{f}}_{mk}$, using $\mathbf{y}_{p,m}$, as $\hat{\mathbf{f}}_{mk}={\mathbb{E}[\mathbf{f}_{mk}\mathbf{y}_{p,m}^{H}]}({\mathbb{E}[\mathbf{y}_{p,m}\mathbf{y}_{p,m}^{H}]})^{-1}\mathbf{y}_{p,m}$~\cite[see Chapter $12$]{Kay_Estimation}, which simplifies to $\hat{\mathbf{f}}_{mk}=\sqrt{\tau_{p}\mathcal{E}_{p,k}}\beta_{mk}c_{mk}\mathbf{y}_{p,m}$, with $c_{mk}\triangleq({\tau_{p}\mathcal{E}_{p,k}\beta_{mk}+\tau_{p}\sum\nolimits_{n\in\mathcal{P}_{l(k)}\backslash k}\mathcal{E}_{p,n}\beta_{mn}+N_{0}})^{-1}$. Further,  $\hat{\mathbf{f}}_{mk}\sim\mathcal{CN}(\mathbf{0}_{N},\alpha_{mk}^2\mathbf{I}_{N})$, with $\alpha_{mk}^2=c_{mk}\tau_{p}\mathcal{E}_{p,k}\beta_{mk}^2$. The estimation error, denoted by $\tilde{\mathbf{f}}_{mk} \triangleq {\mathbf{f}}_{mk} - \hat{\mathbf{f}}_{mk}$, is distributed as $\mathcal{CN}(\mathbf{0}_{N},\bar{\alpha}_{mk}^{2}\mathbf{I}_{N})$, with $\bar{\alpha}_{mk}\triangleq\sqrt{\beta_{mk}-\alpha_{mk}^2}$, and $\tilde{\mathbf{f}}_{mk}$ is uncorrelated with $\hat{\mathbf{f}}_{mk}$ due to orthogonality principle.
	\subsubsection{FD CF}
For the FD system, all UEs transmit pilots in the UL direction, and we estimate the channels between the UEs and every AP's \emph{transmit} and \emph{receive} antennas using these UL pilots. The DL precoders are later designed using channel reciprocity, which obviates the need for separate DL training~\cite{Ashutosh_1, FD_imperfect_CSI, FD_CF_ICC}. Similar to the DTDD case,  the estimated UL channel $\hat{\mathbf{f}}_{\mathsf{u}, mk}$ follows $\mathcal{CN}(\mathbf{0},\alpha_{\mathsf{u}, mk}^2\mathbf{I}_{N})$, with $\alpha_{\mathsf{u}, mk}^2=\tau_{p}\mathcal{E}_{p,k}\beta_{\mathsf{u}, mk}^2c_{\mathsf{u}, mk}$ and $c_{\mathsf{u}, mk}^{-1}={\tau_{p}\mathcal{E}_{p,k}\beta_{\mathsf{u}, mk}+\tau_{p}\sum\nolimits_{{n\in\mathcal{P}_{l(k)}\backslash k}}\mathcal{E}_{p,n}\beta_{\mathsf{u}, mn}+N_{0}}$. The estimated DL channel $\hat{\mathbf{f}}_{\mathsf{d}, mn}$ follows $\mathcal{CN}(\mathbf{0}_{N},\alpha_{\mathsf{d}, mn}^2\mathbf{I}_{N})$, with $\alpha_{\mathsf{d}, mn}^2=\tau_{p}\mathcal{E}_{p,n}\beta_{\mathsf{d}, mn}^2c_{\mathsf{d}, mn}$ and $c_{\mathsf{d}, mn}^{-1}={\tau_{p}\mathcal{E}_{p,n}\beta_{\mathsf{d}, mn}+\tau_{p}\sum\nolimits_{{n'\in\mathcal{P}_{l(n)}\backslash n}}\mathcal{E}_{p,n'}\beta_{\mathsf{d}, mn'}+N_{0}}$. The UL/DL channel estimation error $\tilde{\mathbf{f}}_{\mathsf{i}, mn}$ follows $\mathcal{CN}(\mathbf{0}_{N},\bar{\alpha}_{\mathsf{i}, mn}^{2}\mathbf{I}_{N})$, with $\bar{\alpha}_{\mathsf{i}, mn}\triangleq\sqrt{\beta_{\mathsf{i}, mn}-\alpha_{\mathsf{i}, mn}^2}$, $\mathsf{i}\in\{\mathsf{u},\mathsf{d}\}$.
	
We illustrate the frame structure for channel estimation in FD and DTDD systems in Fig.~\ref{fig:frame_structure}.
Next, we discuss the pilot allocation algorithm. Our goals are two-fold: to minimize the number of orthonormal pilots, $\tau_p$, for channel estimation and mitigate the effect of pilot contamination.

\subsubsection{Pilot Assignment Algorithm}
In CF systems, APs jointly serve the UEs. Thus, physically proximal UEs should not reuse the same pilot sequences even if their nearest APs are different. On the other hand, assigning pilot sequences so that the received signals at all the APs are contamination free requires a pilot length at least equal to the number of UEs; which in turn reduces the duration available for data transmission. However, we note that although all the APs can serve all the UEs, only a subset of APs within the vicinity of a UE receive a signal with sufficient strength for decodability. In other words, at a given AP, the pilot contamination caused by a UE located far away is minimal due to path-loss and shadowing. Hence, we consider a UE-centric clustering and ensure that given any UE, the received signals at the APs within its cluster are contamination-free. We emphasize that all the APs can still participate in data processing to/from all the UEs; we enforce orthogonality within the clusters only for the purpose of pilot allocation. This is illustrated in Fig.~\ref{fig:uncolored_graph}, where the $k$th UE is connected to the $m$th and $m'$th AP. Thus, the pilot assigned to the $k$th UE should be orthonormal to all the UEs served by the $m$th and $m'$th APs. Next, we discuss how to form such clusters.

We define the following sets:
\begin{subequations}\label{eq:sets_color}
	\begin{align}
		&\mathcal{U}_{k}\triangleq\left\{m~\mathrm{s.t.}~ \|\mathbf{u}_{k}-\mathbf{a}_{m}\|\leq r_{\mathsf{o}}, \forall m\in\mathcal{A}\right\},~\forall k\in\mathcal{U}\label{eq:U_k}\\
		&\mathcal{A}_{m}\triangleq\left\{k~\mathrm{s.t.}~ \|\mathbf{u}_{k}-\mathbf{a}_{m}\|\leq r_{\mathsf{o}}, \forall k\in\mathcal{U} \right\},~\forall m\in\mathcal{A}\label{eq:A_m},
	\end{align}
\end{subequations}
where $\mathcal{A}$ is the set of AP indices, $\mathbf{u}_{k}$ and $\mathbf{a}_{m}$ are the locations of the $k$th and $m$th AP, $r_{\mathsf{o}}\triangleq\max\left\{\max\nolimits_{k\in\mathcal{U}} d_{m_{k}k}, d_{\mathsf{SNR}_{\mathsf{o}}}\right\}$,
where $m_{k}$ is the AP index closest to the $k$th UE, i.e., if  $d_{mk}$ is the distance between the $m$th AP and the $k$th UE, then $d_{m_{k}k}=\min\left\{d_{mk}, \forall m\in\mathcal{A}\right\}$. Also, $d_{\mathsf{SNR}_{\mathsf{o}}}$ is the distance from any UE where the received SNR is at least $\gamma_{\mathrm{min}}$, i.e., $d_{\mathsf{SNR}_{\mathsf{o}}}=\max_{d}\left\{\frac{N\mathcal{E}_{p}\beta(d)}{N_{0}}\geq\gamma_{\mathrm{min}} \right\}$, with $\beta(d)=\left(d/d_{0}\right)^{-\mathrm{PL}}$, $d_{0}$ is the reference distance, and $\mathrm{PL}$ is the path-loss exponent. This choice of $r_{\mathsf{o}}$ ensures that:
\begin{enumerate}[label=(\arabic*)]
	\item There is no UE that is not connected to any AP. In particular, every UE is connected to at least one AP even if the received signal strength to its nearest AP is below $\gamma_{\mathrm{min}}$, i.e., if $\max_{k\in\mathcal{U}} d_{m_{k}k}>d_{\mathsf{SNR}_{\mathsf{o}}}$.
	\item Every UE is connected to all APs where the received signal strength is at least $\gamma_{\mathrm{min}}$. However, in a dense deployment where $\max_{k\in\mathcal{U}} d_{m_{k}k}<d_{\mathsf{SNR}_{\mathsf{o}}}$, unnecessary connections to UEs to APs where the received signal strength is below $\gamma_{\mathrm{min}}$ are avoided.
\end{enumerate}

With the clusters as defined above, our pilot length minimization problem can be written as
\begin{align}\label{eq:pilot_length_min}
	\min~&\tau_{p}\notag\\
	\mathrm{subject~to}~&\langle\boldsymbol{\varphi}_{l(k)}, \boldsymbol{\varphi}_{l(k')}\rangle=0, \forall k, k'\in\mathcal{A}_{m}, \forall m\in\mathcal{U}_{k}.
\end{align}
The constraint above ensures that any two UEs that are connected to a common AP are assigned orthogonal pilot sequences. Note that the above pilot assignment is only based on the UEs' and APs' locations and not on the channel state instantiations.
We recast~\eqref{eq:pilot_length_min}
 as a graph coloring problem. We define a graph $\mathcal{G}=\left\{\mathcal{V},\mathcal{E}\right\}$, where the vertex set $\mathcal{V}$ represents the UEs, i.e., $\mathcal{U}$, and edge set is $\mathcal{E}\triangleq\left\{e_{kk'}~\mathrm{s.t.}~ \mathcal{U}_{k}\cap\mathcal{U}_{k'}\neq\emptyset; \forall k, k'\in\mathcal{U}\right\}$. 
Thus, if two UEs, indexed by $k$ and $k'$, are connected to at least one common AP, there is an edge between them. A color assigned to a vertex represents the pilot sequence assigned to it. Then, to satisfy the constraint in~\eqref{eq:pilot_length_min}, we must ensure that any two connected vertices have distinct colors. On the other hand, if two vertices are not connected by an edge, they can potentially reuse the same color~(pilot sequence). Now, let $\mathcal{C}$ be the set of distinct colors, and $\mathcal{C}(k)$ indicate the color assigned to the $k$th UE.
Then, the equivalent coloring problem becomes
\begin{align}\label{eq:vertex_color}
	\min\quad& \lvert\mathcal{C}\rvert\notag\\
	\mathrm{subject~to}\quad &\mathcal{C}(k)\neq \mathcal{C}(k'), ~\mathrm{if}~e_{kk'}\in\mathcal{E}, \forall k, k'\in\mathcal{U}.
\end{align}
From the above arguments, we have the following proposition.
\begin{prop}
The pilot length minimization problem in~\eqref{eq:pilot_length_min} and the coloring problem in~\eqref{eq:vertex_color} are equivalent.
\end{prop}
	
The coloring problem in~\eqref{eq:vertex_color} is NP-complete~[Chapter~$3$, \cite{NP_complete}]. 
Therefore, we first recast the problem in~\eqref{eq:vertex_color} as a bipartite graph coloring problem, where there are two sets of nodes representing APs and UEs, and edges between them represent the connections defined by~\eqref{eq:sets_color}. This bipartite graph coloring problem can be solved efficiently using low complexity greedy techniques, such as the DSATUR algorithm~\cite{DSATUR}. We summarize the solution in Algorithm~\ref{algo:pilot_allocation}. 

The DSATUR algorithm procures \emph{optimal} coloring for all \emph{bipartite graphs}~\cite{DSATUR} in terms of minimizing the number of distinct colors. In each iteration, Algorithm~\ref{algo:pilot_allocation} performs two stages. The first stage selects the vertex among the uncolored vertices based on the number of connected colored or uncolored vertices. Once a vertex is chosen, it is assigned a color that has not been assigned to the vertices connected to it and has been reused the least number of times. When no distinct color is available, a new color is assigned, and the set $\mathcal{C}$ is updated. Finally, we generate orthonormal pilots based on the UE to color mapping.

Next, in Fig.~\ref{fig:uncolored_graph} and Fig.~\ref{fig:Illustrative_pilot}, we illustrate the AP-UE connectivity and the pilot assignment via an example.  We take $8$ APs over a $500$ square meter area. Fig.~\ref{fig:graph} illustrates the formation of the uncolored graph where two vertices (UEs) are connected if they share at least one common AP as per Fig.~\ref{fig:uncolored_graph}. Fig.~\ref{fig:colored_graph} shows the assignment of the pilot sequences via the corresponding colors of the UEs. Next, Fig.~\ref{fig:AP_UE_connections} demonstrates the outcome of the algorithm via a bipartite graph, where the nodes on the left-hand-side represent APs, and the right-hand-side nodes are for UEs. There is an edge if an AP and UE are connected via the set $\mathcal{A}_{m}$. We observe that all the edges emanating from an AP have distinct colors, which implies that all UEs connected to that AP are assigned orthonormal pilots.

\begin{algorithm}[!t]
	\DontPrintSemicolon
	\SetNoFillComment
	\SetNlSty{textbf}{}{:}% style for the line number notations
	\setstretch{0.8}
	\KwIn{$\mathcal{\bar{U}}=\left\{1,2,\ldots, K\right\}$, $\mathcal{C}=\emptyset$}
	\While{$\mathcal{\bar{U}}\neq\emptyset$}{
		\tcc{\textbf{Stage $1$}: Select the uncolored vertex}
		Select the $k\in\mathcal{\bar{U}}$ that has the maximum number of distinct colored vertices connected to it. \;
	If there is more than one such vertex, choose the $k\in\mathcal{\bar{U}}$ within the subset of vertices with the maximum number of distinct colors with the maximum number of vertices connected to it. \;
    Choose any $k\in\mathcal{\bar{U}}$  at random if there is more than one such vertex.\;
    \tcc{\textbf{Stage $2$}: Assign the least used color from the set of available colors}
    Assign color $c(k)$ to $k$th vertex such that
    		\begin{align}\label{eq:color_assign}
    			  &c(k)=\min\quad c(p)\in\mathcal{C}\notag\\
    			\quad&\mathrm{subject~to}~c(k')\neq c(p),\notag \\&\quad\quad\quad\quad\quad \forall k'\in \left\{l~\mathrm{s.t.}~e_{kl}\in\mathcal{E}\right\}. 
    			\end{align}\;
    		\If{$c(k)=\emptyset$}
           {Assign a new color $c(k)$ to vertex $k$\;
    	\textbf{Update}: $\mathcal{C}\leftarrow\mathcal{C} \cup c(k) $\;
    	}
}
	\caption{Pilot Allocation via Graph Coloring}\label{algo:pilot_allocation}
	\end{algorithm}

\begin{figure}[t!]
	\centering
	\includegraphics[width=0.34\textwidth]{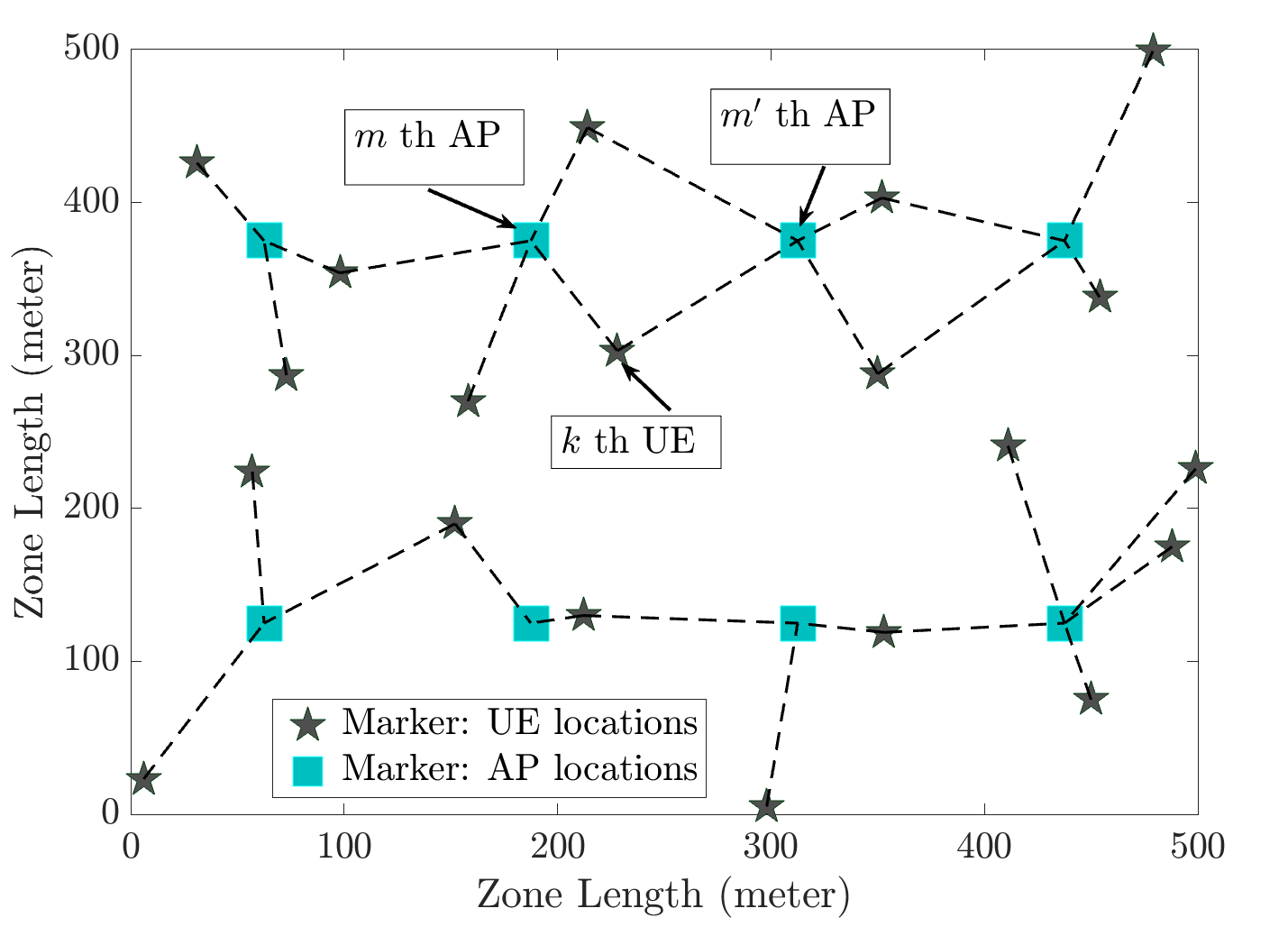}
	\caption{Uncolored AP-UE connections: the lines between the APs and UEs denote the clusters formed by $r_{\mathsf{o}}$.}\label{fig:uncolored_graph}
\end{figure}

\begin{figure*}
	\centering
\begin{subfigure}{0.328\linewidth}
	\centering
	\includegraphics[width=\textwidth]{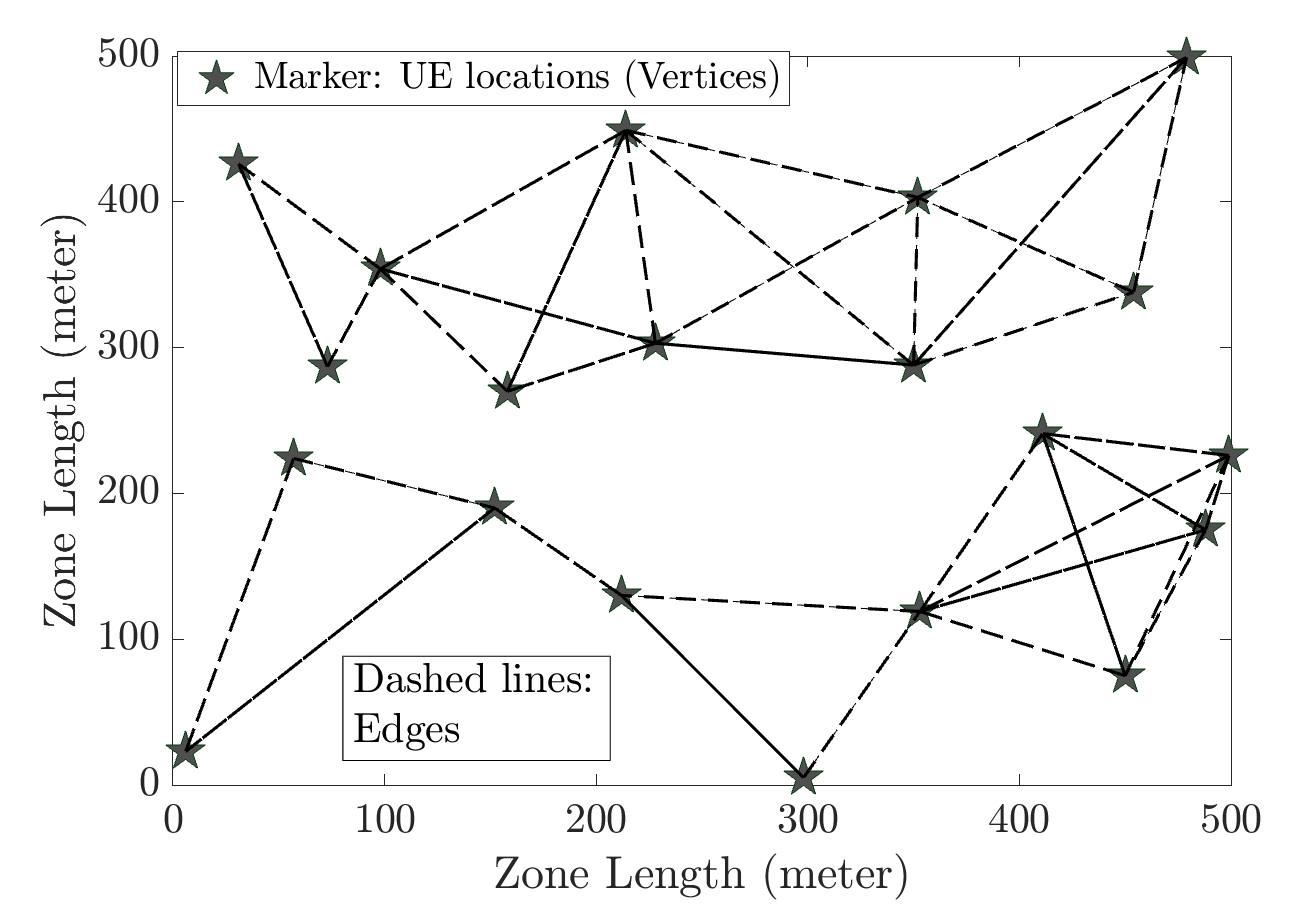}
	\caption{Graph formed by connecting UEs (vertices) that share common AP(s).}\label{fig:graph}
\end{subfigure}\hfil
\begin{subfigure}{0.328\linewidth}
	\centering
	\includegraphics[width=\textwidth]{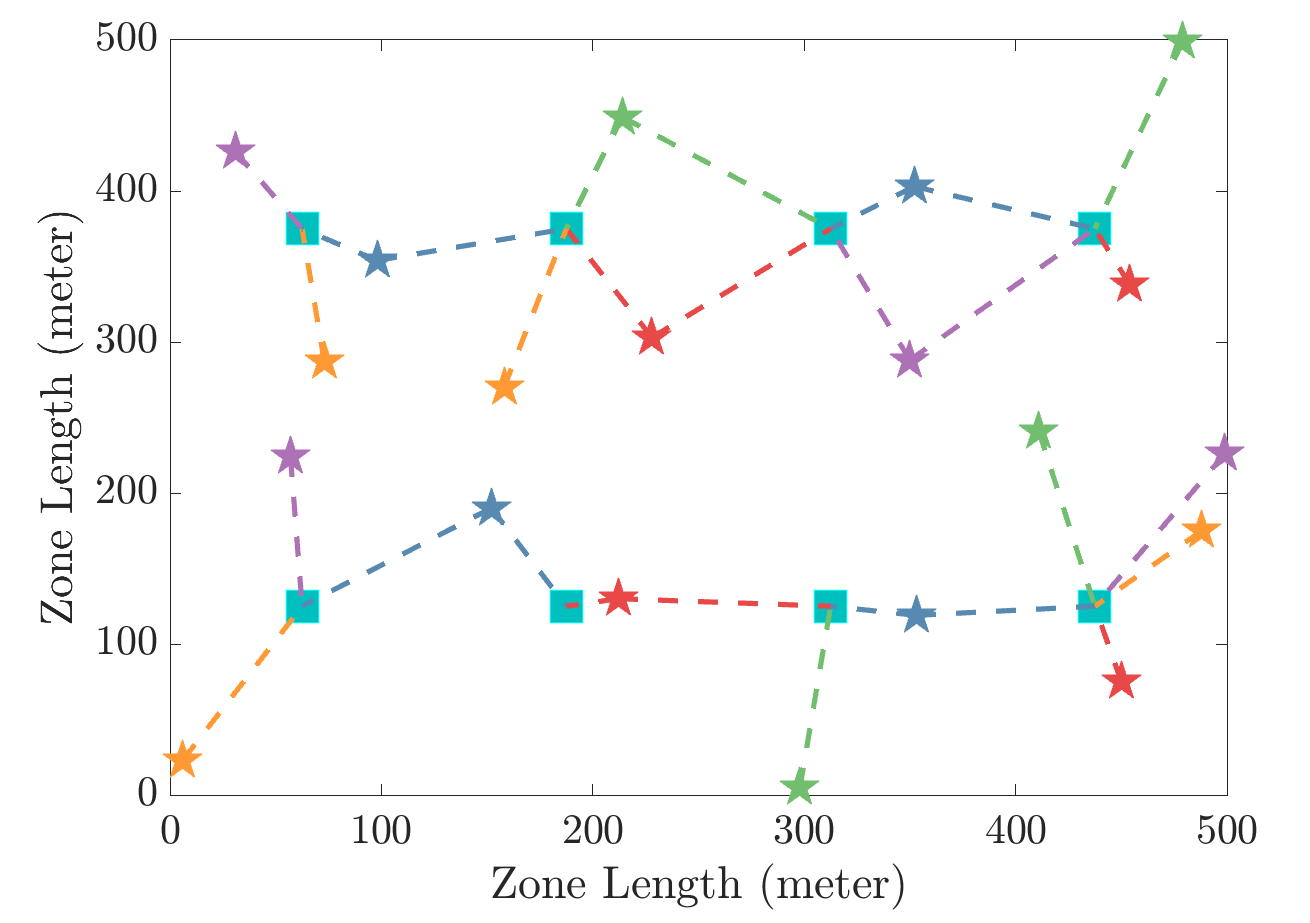}
	\caption{Colored AP-UE connections: Distinct colors correspond to  distinct orthonormal pilot sequences.}\label{fig:colored_graph}
\end{subfigure}\hfil
\begin{subfigure}{0.328\linewidth}
	\centering
	\includegraphics[width=\textwidth]{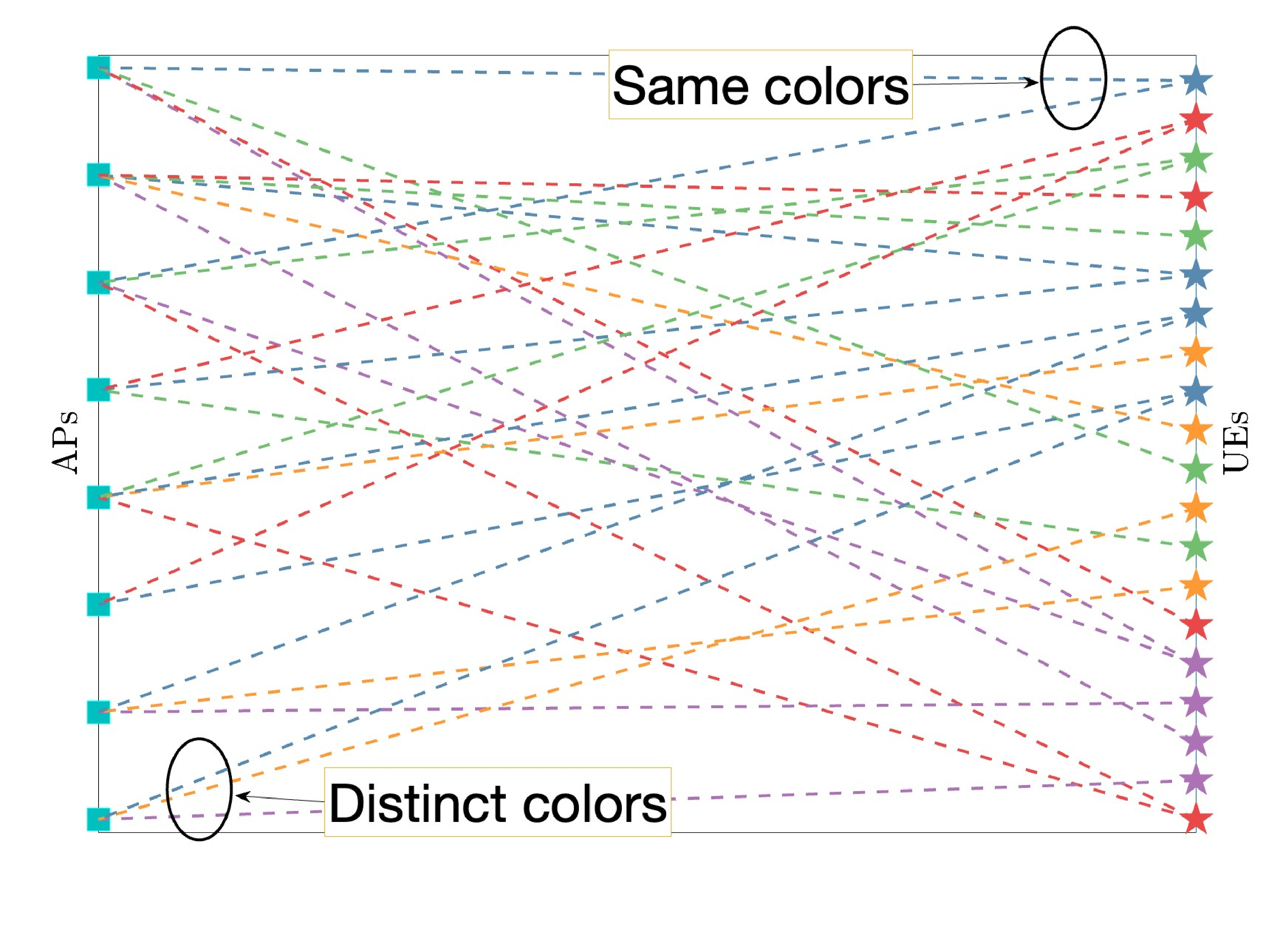}
	\caption{All lines emerging from a UE have the same color, and all the lines merging into an AP have distinct colors. }\label{fig:AP_UE_connections}
\end{subfigure}
\caption{An example of the pilot assignment algorithm with $M=8, N=10, K = 20, d_{0} =20$ m, and $\mathrm{PL}=3.76$. The pilot SNR is $10$ dB and $\gamma_{\mathrm{min}}$ is $0$ dB. We observe that to solve~\eqref{eq:vertex_color} we need $5$ colors, and hence $\tau_{p} = 5$.}\label{fig:Illustrative_pilot}
\end{figure*}

The attractive features of Algorithm~\ref{algo:pilot_allocation} are: the algorithm does not require an exchange of instantaneous  UL/DL SINRs for pilot allocation, unlike existing algorithms~\cite{ICC_pilot, Cluster,Tabu}. Also, algorithms in the literature that require only UE locations for pilot allocation~\cite{Location_Pilot, Pilot_power} do not consider pilot length optimization. The main advantage of our pilot allocation technique is that, in a DTDD and FD-based system, we can isolate the issue of pilot allocation from AP scheduling and power control, which makes our solution easy to implement. Numerical experiments illustrating the superiority of Algorithm~\ref{algo:pilot_allocation} compared to existing works are delegated to Sec.~\ref{sec:numerical}.

In the following sections, we present the SE analysis, AP scheduling algorithm, and UL-DL power allocation policies for the DTDD system. As we will see in Sec.~\ref{sec:FD_CF}, our analysis easily generalizes to an FD-enabled CF system also.

\section{Spectral Efficiency Analysis: CF DTDD}\label{sec:DTDD_SE}
 This section presents the UL and DL signaling model and derives closed-form expressions for the sum UL-DL SE, which we further use for AP scheduling and power allocation. Let the sets  $\mathcal{A}_{\mathsf{u}}$ and $\mathcal{A}_{\mathsf{d}}$ contain the indices of the APs scheduled in the UL and DL, respectively. In a DTDD system, the APs are HD. Thus, $\mathcal{A}_{\mathsf{u}}\cap\mathcal{A}_{\mathsf{d}}=\emptyset$. Also, let $\mathcal{A}_{\mathsf{s}}\triangleq\mathcal{A}_{\mathsf{u}}\cup\mathcal{A}_{\mathsf{d}}\subseteq\mathcal{A}$. 

\subsection{Analysis with MMSE combiner \& RZF precoder} 
\emph{In the UL,} the $k$th UE~($k\in\mathcal{U}_{\mathsf{u}}$) sends the symbol $s_{\mathsf{u},k}$ with power $\mathcal{E}_{\mathsf{u},k}$. The data symbol of each UE is modeled as zero mean, unit variance, and  $\mathbb{E}[s_{\mathsf{u},k}s_{\mathsf{u},k'}^*]=0, k'\neq k, \forall k,k'\in\mathcal{U}_{\mathsf{u}}$.
The signal received at the $m$th UL AP can be expressed as 
   $\mathbf{y}_{\mathsf{u},m}=\sum\nolimits_{\substack{n\in\mathcal{U}_{\mathsf{u}}}}{\textstyle\sqrt{\mathcal{E}_{\mathsf{u},n}}}\mathbf{f}_{mn}s_{\mathsf{u},n}+\sum\nolimits_{j\in \mathcal{A}_{\mathsf{d}}} \tilde{\mathbf{G}}_{mj}\mathbf{x}_{\mathsf{d},j}+\mathbf{w}_{\mathsf{u},m}\in \mathbb{C}^{N},$ where  $\mathbf{x}_{\mathsf{d},j}=\sqrt{\mathcal{E}_{\mathsf{d}}}\mathbf{P}_{j}\mathsf{diag}(\boldsymbol{\kappa}_{j})\mathbf{s}_{\mathsf{d}}=\sqrt{\mathcal{E}_{\mathsf{d}}}\left[\mathbf{p}_{j1},\ldots,\mathbf{p}_{j|\mathcal{U}_{\mathsf{d}}|}\right]\mathsf{diag}(\boldsymbol{\kappa}_{j})\mathbf{s}_{\mathsf{d}}$
is the transmitted DL data vector with  $\mathcal{E}_{\mathsf{d}}$ being the total radiated power, $\mathbf{p}_{jn}=\left[\mathbf{P}_{j}\right]_{:,n}\in\mathbb{C}^{N}$ being the precoding matrix for the $n$th DL UE, and $\boldsymbol{\kappa}_{j}\triangleq[\kappa_{j1},\kappa_{j2},\ldots,\kappa_{j|\mathcal{U}_{\mathsf{d}}|}]^{T}$ being the vector of power control coefficients, all at the $j$th DL AP. Here,  $\kappa_{jn}$, i.e., the $n$th element of $\boldsymbol{\kappa}_{j}$, indicates the fraction of power dedicated by the $j$th AP to  the $n$th DL UE ($n\in\mathcal{U}_{\mathsf{d}}$). The DL signal vector $\mathbf{s}_{\mathsf{d}}=\left[s_{\mathsf{d},1},\ldots,s_{\mathsf{d},|\mathcal{U}_{\mathsf{d}}|}\right]^T$ follows $\mathbb{E}\left[\mathbf{s}_{\mathsf{d}}\mathbf{s}_{\mathsf{d}}^{H}\right]=\mathbf{I}_{|\mathcal{U}_{\mathsf{d}}|}$. Finally, $\mathbf{w}_{\mathsf{u},m} \sim \mathcal{CN}(\mathbf{0}_{N},N_{0}\mathbf{I}_N)$ is the additive noise.
Each AP pre-processes the received signals using local combiners and sends them to the CPU for joint decoding, which is important for the scalability of the overall system~\cite{making_cellfree}. Let $\mathbf{v}_{mk}\in\mathbb{C}^{N}$ be the local combining vector at the $m$th AP for $k$th UE's UL data stream. Then, the local estimate of the $k$th UE's signal at the $m$th AP becomes $
\hat{s}_{\mathsf{u},mk}=\sqrt{\mathcal{E}_{\mathsf{u},k}}\mathbf{v}_{mk}^{H}\mathbf{f}_{mk}s_{\mathsf{u},k}+\sum\nolimits_{n\in\mathcal{U}_{\mathsf{u}}\backslash k}\sqrt{\mathcal{E}_{\mathsf{u},n}}\mathbf{v}_{mk}^{H}\mathbf{f}_{mn}s_{\mathsf{u},n}+\sqrt{\mathcal{E}_{\mathsf{d}}}\sum\nolimits_{j\in\mathcal{A}_{\mathsf{d}}}\sum\nolimits_{n\in\mathcal{U}_{\mathsf{d}}}\kappa_{jn}\mathbf{v}_{mk}^{H} \tilde{\mathbf{G}}_{mj}\mathbf{p}_{jn}s_{\mathsf{d},n}+\mathbf{v}_{mk}^{H}\mathbf{w}_{\mathsf{u},mk}$.
To design $\mathbf{v}_{mk}$, the AP utilizes the knowledge of $\hat{\mathbf{f}}_{mk}, \forall k\in\mathcal{U}_{\mathsf{u}}$, so that the MSE of the locally estimated signal is minimized, i.e., $\mathbf{v}_{mk}^{\mathsf{opt.}}=\min\limits_{\mathbf{v}_{mk}\in\mathbb{C}^{N}}\mathbb{E}\left[\lvert s_{\mathsf{u},k}-\mathbf{v}_{mk}^{H} \mathbf{y}_{\mathsf{u},m}\rvert^{2} \vert \hat{\mathbf{f}}_{mk}\right]$. 
Thus, the optimal combiner is $\mathbf{v}_{mk}^{\mathsf{opt.}}=\mathbf{R}_{mk}^{-1}\hat{\mathbf{f}}_{mk}$, where $ \mathbf{R}_{mk}=\Big(\sum\nolimits_{k\in\mathcal{U}_{\mathsf{u}}}\mathcal{E}_{\mathsf{u},k} \hat{\mathbf{f}}_{mk}\hat{\mathbf{f}}_{mk}^{H}+ \sum\nolimits_{k\in\mathcal{U}_{\mathsf{u}}}\mathcal{E}_{\mathsf{u},k}\mathbb{E}[\tilde{\mathbf{f}}_{mk}\tilde{\mathbf{f}}_{mk}^{H}]+\hspace*{-1mm}\mathcal{E}_{\mathsf{d}}\hspace*{-1mm}\sum\nolimits_{j\in\mathcal{A}_{\mathsf{d}}}\sum\nolimits_{n\in\mathcal{U}_{\mathsf{d}}}\kappa_{jn}^2\mathbb{E}[\tilde{\mathbf{G}}_{mj}\mathbf{p}_{jn}\mathbf{p}_{jn}^{H}\tilde{\mathbf{G}}_{mj}^{H}]+N_{0}\mathbf{I}_{N}\Big)$.
Next, these locally estimated signals relayed from the APs are linearly combined at the CPU with combining weights being $\omega_{mk}, \forall m\in\mathcal{A}_{\mathsf{u}}, \forall k\in\mathcal{U}_{\mathsf{u}}$ so that the received SINR at the CPU for each UE is maximized.  Thus, the $k$th stream of the received signal at the CPU is $\hat{s}_{\mathsf{u},k}=
\sqrt{\mathcal{E}_{\mathsf{u},k}}\sum\nolimits_{m\in\mathcal{A}_{\mathsf{u}}}\omega_{mk}^{*}\mathbf{v}_{mk}^{H}\mathbf{f}_{mk}s_{\mathsf{u},k}+\sum\nolimits_{n\in\mathcal{U}_{\mathsf{u}}\backslash k}\sqrt{\mathcal{E}_{\mathsf{u},n}}\sum\nolimits_{m\in\mathcal{A}_{\mathsf{u}}}\omega_{mk}^{*}\mathbf{v}_{mk}^{H}\mathbf{f}_{mn}s_{\mathsf{u},n}+
\sum\nolimits_{m\in\mathcal{A}_{\mathsf{u}}}\omega_{mk}^{*}\sum\nolimits_{j\in\mathcal{A}_{\mathsf{d}}}\sqrt{\mathcal{E}_{\mathsf{d}}}\sum\nolimits_{n\in\mathcal{U}_{\mathsf{d}}}\kappa_{jn}\mathbf{v}_{mk}^{H} \tilde{\mathbf{G}}_{mj}\mathbf{p}_{jn}s_{\mathsf{d},n}+\sum\nolimits_{m\in\mathcal{A}_{\mathsf{u}}}\omega_{mk}^{*}\mathbf{v}_{mk}^{H}\mathbf{w}_{\mathsf{u},mk}$.
Then the UL SINR of the $k$th UE, denoted by $\eta_{\mathsf{u},k}$,  at the CPU is equal to~\cite{making_cellfree}
\begin{align}\label{eq:UL_SINR}
	\hspace{-.3cm}\frac{\mathcal{E}_{\mathsf{u},k}\lvert\boldsymbol{\omega_{k}}^{H}\mathbb{E}[\mathbf{u}_{kk}]\rvert^{2}}{\boldsymbol{\omega_{k}}^{H}
		\begin{pmatrix}
\sum\nolimits_{i\in\mathcal{U}_{\mathsf{u}}}\mathcal{E}_{\mathsf{u},i}\mathbb{E}[\mathbf{u}_{ki}\mathbf{u}_{ki}^{H}]-\mathcal{E}_{\mathsf{u},k}\mathbb{E}\left[\mathbf{u}_{kk}\right]\mathbb{E}[\mathbf{u}_{kk}^{H}]&\\ +\sum\nolimits_{n\in\mathcal{U}_{\mathsf{d}}}\mathbb{E}[\mathbf{a}_{kn}\mathbf{a}_{kn}^{H}]+\mathbf{N}_{\mathsf{eff.}}
		\end{pmatrix}
		\boldsymbol{\omega_{k}}},\hspace{-.3cm}
\end{align}
where $\boldsymbol{\omega_{k}}\triangleq[\omega_{1k},\omega_{2k},\ldots,\omega_{|\mathcal{A}_{\mathsf{u}}|k}]^T\in\mathbb{C}^{|\mathcal{A}_{\mathsf{u}}|}$, $\mathbf{u}_{ki}\triangleq[\mathbf{v}_{1k}^{H}\mathbf{f}_{1i},\mathbf{v}_{2k}^{H}\mathbf{f}_{2i},\ldots,\mathbf{v}_{|\mathcal{A}_{\mathsf{u}}|k}^{H}\mathbf{f}_{|\mathcal{A}_{\mathsf{u}}|i}]^{T}\in\mathbb{C}^{|\mathcal{A}_{\mathsf{u}}|}$,  $[\mathbf{a}_{kn}]_{m}$ $=\sum\nolimits_{j\in\mathcal{A}_{\mathsf{d}}}\sqrt{\mathcal{E}_{\mathsf{d}}}\kappa_{jn}\mathbf{v}_{mk}^{H} \tilde{\mathbf{G}}_{mj}\mathbf{p}_{jn}$, and 
$\mathbf{N}_{\mathsf{eff.}}=N_{0}{\tt diag}( \mathbb{E}[\|\mathbf{v}_{1k}\|^2], \ldots,\mathbb{E}[\|\mathbf{v}_{|\mathcal{A}_{\mathsf{u}}|k}\|^2])\in\mathbb{C}^{|\mathcal{A}_{\mathsf{u}}|\times |\mathcal{A}_{\mathsf{u}}|}$. We evaluate the optimal weights using the following lemma.

\begin{lem}\label{lem:weighted_UL_SINR}
	The SINR of the $k$th~($k\in\mathcal{U}_{\mathsf{u}}$) UE is maximized by $\boldsymbol{\omega}_{k}^{\mathsf{opt.}}=\mathsf{c}_{k}\sqrt{\mathcal{E}_{\mathsf{u},k}}\mathbf{R}_{\omega,k}^{-1}\mathbb{E}\left[\mathbf{u}_{kk}\right]$, where\footnote{We use the fact that for any given vector $\mathbf{x}\in\mathbb{C}^{N}$ and a positive definite matrix $\mathbf{A}$, $\max_{\mathbf{y}\in\mathbb{C}^{N}}[\lvert \mathbf{y}^{H}\mathbf{x}\rvert^{2}/(\mathbf{y}^{H}\mathbf{A}\mathbf{y})]=\mathbf{x}^{H}\mathbf{A}\mathbf{x}$, and $\mathbf{y}^{\mathsf{opt.}}=\mathbf{A}^{-1}\mathbf{x}$.} 
	\begin{align}
	\mathbf{R}_{\omega,k}\triangleq&(\textstyle{\sum\nolimits_{i\in\mathcal{U}_{\mathsf{u}}}\mathcal{E}_{\mathsf{u},i}\mathbb{E}[\mathbf{u}_{ki}\mathbf{u}_{ki}^{H} ]-\mathcal{E}_{\mathsf{u},k}\mathbb{E}[\mathbf{u}_{kk}]\mathbb{E}[\mathbf{u}_{kk}^{H}]}\notag\\&\textstyle{+\sum\nolimits_{n\in\mathcal{U}_{\mathsf{d}}}\mathbb{E}[\mathbf{a}_{kn}\mathbf{a}_{kn}^{H}]+\mathbf{N}_{\mathsf{eff.}})},
	\end{align}
	 and $\mathsf{c}_{k}$ is a scaling constant and can be taken as $\sqrt{\mathcal{E}_{\mathsf{u},k}}$ to make the weights dimensionless. Also, the maximum SINR is $ \mathcal{E}_{\mathsf{u},k}\mathbb{E}[{\mathbf{u}}_{kk}^{H}]\mathbf{R}_{\omega,k}^{-1}\mathbb{E}[{\mathbf{u}}_{kk}]$.
\end{lem}

\emph{In the DL,} the counterpart of the MMSE combiner is the RZF precoder, which is $\mathbf{p}_{jn}=\mathcal{E}_{\mathsf{d},n}\big(\sum\nolimits_{n'\in\mathcal{U}_{\mathsf{d}}}\mathcal{E}_{\mathsf{d},n'} \hat{\mathbf{f}}_{jn'}\hat{\mathbf{f}}_{jn'}^{H}+ \sum\nolimits_{n'\in\mathcal{U}_{\mathsf{d}}}\mathcal{E}_{\mathsf{d},n'} \mathbb{E}\left[\tilde{\mathbf{f}}_{jn'}\tilde{\mathbf{f}}_{jn'}^{H}\right]+N_{0}\mathbf{I}_{N}\big)^{-1}\hat{\mathbf{f}}_{jn}$. Assuming channel reciprocity, the signal received at the $n$th~($n\in\mathcal{U}_{\mathsf{d}}$) DL UE can be written as $r_{\mathsf{d},n}=\sum\nolimits_{j\in\mathcal{A}_{\mathsf{d}}}\kappa_{jn}\sqrt{{\mathcal{E}_{\mathsf{d}}}}\mathbf{f}_{jn}^{T}\mathbf{p}_{jn}s_{\mathsf{d},n}+\sum\nolimits_{k\in\mathcal{U}_{\mathsf{u}}}\sqrt{\mathcal{E}_{\mathsf{u},k}}\mathtt{g}_{nk}s_{\mathsf{u},k}+\sum\nolimits_{j\in\mathcal{A}_{\mathsf{d}}}\sum\nolimits_{\substack{ q\in\mathcal{U}_{\mathsf{d}}\backslash n}}\kappa_{jq}\sqrt{{\mathcal{E}_{\mathsf{d}}}}\mathbf{f}_{jn}^{T}\mathbf{p}_{jq}s_{\mathsf{d},q}+w_{\mathsf{d},n}$. The AWGN $w_{\mathsf{d},n}$ follows $\mathcal{CN}(0,N_{0})$. Now, the corresponding DL SINR, denoted by $\eta_{\mathsf{d},n}$, becomes
     \begin{align}\label{eq:DL_SINR}
     	\frac{\mathcal{E}_{\mathsf{d}}|\boldsymbol{\kappa}_{n}^{T}\mathbb{E}[\mathbf{d}_{nn}]|^2}{\begin{pmatrix}\textstyle{\mathcal{E}_{\mathsf{d}}{\tt var}\big\lbrace\boldsymbol{\kappa}_{n}^{T}\mathbf{d}_{nn}\rbrace+\mathcal{E}_{\mathsf{d}}\sum\nolimits_{q\in \mathcal{U}_{\mathsf{d}}\backslash n}\mathbb{E}[\lvert\boldsymbol{\kappa}_{q}^{T}\mathbf{d}_{nq}\rvert^{2}]}&\\\textstyle{+\sum\nolimits_{k\in\mathcal{U}_{\mathsf{u}}}\mathcal{E}_{\mathsf{u},k}\mathbb{E}|\mathtt{g}_{nk}|^{2}+N_{0}}\end{pmatrix}},
     \end{align}
  where $\boldsymbol{\kappa}_{n}\triangleq[\kappa_{1n},\kappa_{2n},\ldots,\kappa_{|\mathcal{A}_{\mathsf{d}}|n}]^{T}$ and $\mathbf{d}_{nq}\triangleq[\mathbf{f}_{1n}^T\mathbf{p}_{1q}, \mathbf{f}_{2n}^T\mathbf{p}_{2q}, \ldots,\mathbf{f}_{\lvert\mathcal{A}_{\mathsf{d}}\rvert n}^T\mathbf{p}_{\lvert\mathcal{A}_{\mathsf{d}}\rvert q}]^{T}$. 
 The DL power control coefficients, $\boldsymbol{\kappa}_{n}$, are the weights assigned to the signal transmitted to each UE from every AP, which can be controlled to maximize the sum SE.

We use results from random matrix theory to derive closed-form expressions for the SINRs.  We present the expressions for the case of ZF precoders and combiners. Our analysis includes not only the effects of CLIs but also coherent interference, unlike~\cite{Mohammadi_JSAC}, where MRC and MFP are considered without accounting for the effects of pilot contamination.
 
 \subsection{Analysis with  ZF combiner \& precoder}
We first construct the ZF combining and precoding matrices.
As $\tau_{p}<K$,  the concatenated estimated channel matrix at $m$th UL AP, denoted by $\mathbf{\hat{F}}_{\mathsf{u},m}=[\mathbf{\hat{f}}_{m1},\ldots,\mathbf{\hat{f}}_{m|\mathcal{U}_{\mathsf{u}}|}]\in\mathbb{C}^{N\times |\mathcal{U}_{\mathsf{u}}|}$,  may not have full column rank, i.e., $\mathbf{\hat{F}}_{\mathsf{u},m}^{H}\mathbf{\hat{F}}_{\mathsf{u},m}$ may not be invertible. Thus, we construct a full rank matrix as $\mathbf{Z}_{\mathsf{u},m}\triangleq\mathbf{Y}_{p,m}\boldsymbol{\Phi}\in\mathbb{C}^{N\times\tau_{p}}$, where $\boldsymbol{\Phi}$ is the pilot matrix with its $l$th column being $\boldsymbol{\varphi}_{l}$ and $\mathbf{Y}_{p,m}$. We observe that we can compute $\hat{\mathbf{f}}_{mk}$ from $\mathbf{Z}_{\mathsf{u},m}$ using the relation $\hat{\mathbf{f}}_{mk}=c_{mk}\sqrt{\tau_{p}\mathcal{E}_{p,k}}\beta_{mk}\mathbf{Z}_{\mathsf{u},m}\mathbf{e}_{l(k)}$, where $\mathbf{e}_{l(k)}\in\mathbb{C}^{\tau_{p}}$ is the standard basis with $l(k)$th coordinate being $1$.\footnote{Here, we recapitulate that $k$th UE uses the pilot sequence $\boldsymbol{\varphi}_{l(k)}$, and the set $\mathcal{P}_{l(k)}$ contains all the UE indices, including the $k$th UE, which use the pilot sequence $\boldsymbol{\varphi}_{l(k)}$. Thus, multiplying $\mathbf{Z}_{\mathsf{u},m}$ by $\mathbf{e}_{l(k)}$, we obtain a common estimate for the channels of all UEs which use the pilot sequence $\boldsymbol{\varphi}_{l(k)}$. } Then, the  ZF combining vector intended for the $k$th UE at the $m$th AP becomes $\mathbf{v}_{mk}=\gamma_{mk}\mathbf{Z}_{\mathsf{u},m}(\mathbf{Z}_{\mathsf{u},m}^{H}\mathbf{Z}_{\mathsf{u},m})^{-1}\mathbf{e}_{l(k)}$,
where we set $\gamma_{mk}=\sqrt{\tau_{p}\mathcal{E}_{p,k}}\beta_{mk}$ to ensure that $\mathbb{E}[\mathbf{v}_{mk}^{H}\hat{\mathbf{f}}_{mk}]=\alpha_{mk}^2$.
In the DL, we again construct a full rank matrix $\mathbf{Z}_{\mathsf{d},j}=\mathbf{Y}_{p,j}\boldsymbol{\Phi}\in\mathbb{C}^{N\times\tau_{p}},~\forall j\in\mathcal{A}_{\mathsf{d}}$.
We can obtain $\hat{\mathbf{f}}_{jn}$ from $\mathbf{Z}_{\mathsf{d},j}$ as $\hat{\mathbf{f}}_{jn}=c_{jn}\sqrt{\tau_{p}\mathcal{E}_{p,j}}\beta_{jn}\mathbf{Z}_{\mathsf{d},j}\mathbf{e}_{l(n)}$.
Next, the ZF precoding vector~(unit normalized) intended for the $n$th DL UE at the $j$th AP is $\mathbf{p}_{jn}=\frac{\mathbf{Z}_{\mathsf{d},j}\left(\mathbf{Z}_{\mathsf{d},j}^{H}\mathbf{Z}_{\mathsf{d},j}\right)^{-1}\mathbf{e}_{l(n)}}{\sqrt{\mathbb{E}[\|\mathbf{Z}_{\mathsf{d},j}\left(\mathbf{Z}_{\mathsf{d},j}^{H}\mathbf{Z}_{\mathsf{d},j}\right)^{-1}\mathbf{e}_{l(n)}\|^{2}]}}$.
Using~\cite[Lemma~$6$]{WH_Verdu}, we can show that $\mathbb{E}[\|\mathbf{Z}_{\mathsf{d},j}(\mathbf{Z}_{\mathsf{d},j}^{H}\mathbf{Z}_{\mathsf{d},j})^{-1}\mathbf{e}_{l(n)}\|^{2}]=\frac{c_{jn}}{N-\tau_{p}}$.\footnote{Thus, the condition $N\geq(\tau_{p}+1)$ has to be satisfied. Each AP has $\tau_{p}$ combining and precoding vectors, and the same vector is utilized for all UEs sharing the same pilot sequence. }
Next, we present closed-form expressions for the optimal weights and UL-DL SINRs with ZF combiners and precoders.
\begin{lem}\label{lem:optimal_weights}
The optimal weights for the $k$th UE's data stream at the CPU for $\lvert\mathcal{A}_{\mathsf{u}}\rvert$ UL APs is $\boldsymbol{\omega}_{k}^{\mathsf{opt}}=\mathcal{E}_{\mathsf{u},k}\mathbf{R}_{\omega,k}^{-1}\bar{\mathbf{u}}_{k}$, with $\bar{\mathbf{u}}_{k}=[\alpha_{1k}^2,\alpha_{2k}^2,\ldots,\alpha_{|\mathcal{A}_{\mathsf{u}}|k}^2]^{T}$ and $\mathbf{R}_{\omega,k}=\sum\nolimits_{k'\in\mathcal{P}_{l(k)}\backslash k}\mathcal{E}_{\mathsf{u},k'}\bar{\mathbf{u}}_{k'} \bar{\mathbf{u}}_{k'}^{H}+\frac{1}{N-\tau_p}\mathbf{\dot{R}}_{\omega,k}+\frac{\mathcal{E}_{\mathsf{d}}}{N-\tau_{p}}\mathbf{\ddot{R}}_{\omega,k}$, where  $\mathbf{\dot{R}}_{\omega,k}$ and 
$\mathbf{\ddot{R}}_{\omega,k}\in\mathbb{C}^{\lvert\mathcal{A}_{\mathsf{u}}\rvert\times \lvert\mathcal{A}_{\mathsf{u}}\rvert}$ are two diagonal matrices with $m$th diagonal entries  $$\left[\mathbf{\dot{R}}_{\omega,k}\right]_{m}=\sum\nolimits_{k'\in\mathcal{U}_{\mathsf{u}}}\mathcal{E}_{\mathsf{u},k'}\alpha_{mk}^2(\beta_{mk'}-\alpha_{mk'}^2)+N_{0}\alpha_{mk}^2,$$ and $$\left[\mathbf{\ddot{R}}_{\omega,k}\right]_{m}=\sum\nolimits_{j\in\mathcal{A}_{\mathsf{d}}}N\kappa_{jn}^2\zeta_{mj}^{\mathsf{InAP}}\alpha_{mk}^2.$$ 
\end{lem}
\begin{proof}
With ZF combining, the $m$th term of $\mathbb{E}[\mathbf{u}_{kk}]$ in~\eqref{eq:UL_SINR}  evaluates to $\gamma_{mk}\mathbb{E}[\mathbf{e}_{l(k)}^{T}(\mathbf{Z}_{\mathsf{u},m}^{H}\mathbf{Z}_{\mathsf{u},m})^{-1}\mathbf{Z}_{\mathsf{u},m}^{H}\mathbf{f}_{mk}]=\alpha_{mk}^2$. Next, we evaluate $\mathbb{E}[\mathbf{u}_{ki}\mathbf{u}_{ki}^{H}]\in\mathbb{C}^{\lvert\mathcal{A}_{\mathsf{u}}\rvert\times \lvert\mathcal{A}_{\mathsf{u}}\rvert}$ in closed form. First, we consider the UEs that use the same pilot as the $k$th UE, i.e., the UEs indexed by $i\in \mathcal{P}_{l(k)}, i\neq k$. In this case, the $m$th diagonal entry of $\mathbb{E}[\mathbf{u}_{ki}\mathbf{u}_{ki}^{H}]$ can be expanded as	\begin{multline}\label{eq:UL_weights_1}
	\!\!\!\!\!\mathbb{E}[\mathbf{v}_{mk}^{H}\mathbf{f}_{mi}\mathbf{f}_{mi}^{H}\mathbf{v}_{mk}]=\mathbb{E}[\mathbf{v}_{mk}^{H}\hat{\mathbf{f}}_{mi}\hat{\mathbf{f}}_{mi}^{H}\mathbf{v}_{mk}]+\mathbb{E}[\mathbf{v}_{mk}^{H}\tilde{\mathbf{f}}_{mi}\tilde{\mathbf{f}}_{mi}^{H}\mathbf{v}_{mk}]\\{(a)\atop=}\alpha^{4}_{mk}+\mathbb{E}[\mathbf{v}_{mk}^{H}\mathbb{E}[\tilde{\mathbf{f}}_{mi}\tilde{\mathbf{f}}_{mi}^{H}]\mathbf{v}_{mk}].
	\end{multline}
	Here, in $(a)$ we use the fact that $\mathbb{E}[\mathbf{v}_{mk}^{H}\hat{\mathbf{f}}_{mi}\hat{\mathbf{f}}_{mi}^{H}\mathbf{v}_{mk}]=\alpha^{4}_{mk}$ if $i\in\mathcal{P}_{l(k)}$.
Next, since $\tilde{\mathbf{f}}_{mi}\sim\mathcal{CN}(\mathbf{0},(\beta_{mi}-\alpha_{mi}^2)\mathbf{I}_{N})$ and 
	$\mathbb{E}[\mathbf{v}_{mk}^{H}\mathbf{v}_{mk}]=\frac{\alpha_{mk}^2}{N-\tau_p}$, 
	$\mathbb{E}[\mathbf{v}_{mk}^{H}\mathbb{E}[\tilde{\mathbf{f}}_{mi}\tilde{\mathbf{f}}_{mi}^{H}]\mathbf{v}_{mk}] = \frac{1}{N-\tau_p}\alpha_{mk}^2(\beta_{mi}-\alpha_{mi}^2).$
	\textcolor{black}{
	The off-diagonal $(m,n)$th element of $\mathbb{E}[\mathbf{u}_{ki}\mathbf{u}_{ki}^H]$ can be calculated as $[\mathbb{E}[\mathbf{u}_{ki}\mathbf{u}_{ki}^H]]_{m,n}=\mathbb{E}[\mathbf{v}_{mk}^{H}\hat{\mathbf{f}}_{mi}\hat{\mathbf{f}}_{ni}^{H}\mathbf{v}_{nk}]+\mathbb{E}[\mathbf{v}_{mk}^{H}\tilde{\mathbf{f}}_{mi}\tilde{\mathbf{f}}_{ni}^{H}\mathbf{v}_{nk}]=\alpha_{mk}^2\alpha_{ni}^2.$}
	\textcolor{black}{Next, for the UEs that do not share the same pilot as the $k$th UE~($\forall i'\notin \mathcal{P}_{l(k)}$), it is easy to show that the off-diagonal entries of $\mathbb{E}[\mathbf{u}_{ki'}\mathbf{u}_{ki'}^{H}]$ are zero. The $m$th diagonal entry can be evaluated as $
		\mathbb{E}[\mathbf{v}_{mk}^{H}\mathbf{f}_{mi'}\mathbf{f}_{mi'}^{H}\mathbf{v}_{mk}]=\mathbb{E}[\mathbf{v}_{mk}^{H}\mathbb{E}[\tilde{\mathbf{f}}_{mi'}\tilde{\mathbf{f}}_{mi'}^{H}]\mathbf{v}_{mk}]={\alpha_{mk}^2(\beta_{mi'}-\alpha_{mi'}^2)}/{(N-\tau_{p})}.$}
		
	\textcolor{black}{
	Next, we can show that
	$\mathbb{E}[|\mathbf{v}_{mk}^{H} \tilde{\mathbf{G}}_{mj}\mathbf{p}_{jn}|^2]=\frac{N\zeta_{mj}^{\mathsf{InAP}}\alpha_{mk}^2}{(N-\tau_{p})}.$
	Finally,
	 $\mathbf{N}_{\mathsf{eff.}}=N_{0}{\tt diag}( \frac{\alpha_{1k}^2}{N-\tau_{p}}, \ldots,\frac{\alpha_{|\mathcal{A}_{\mathsf{u}}|k}^2}{N-\tau_{p}}).$ This completes the key steps in the derivation of $\boldsymbol{\omega}_{k}^{\mathsf{opt.}}$. With a little algebraic manipulation, we arrive at the expressions in the Lemma.}
\end{proof}
Later, in Fig.~\ref{fig:UL_power_control}, we illustrate the correctness of the above lemma by matching the UL SE obtained using the weights computed from the above lemma with the sum UL SE obtained from Lemma~\ref{lem:weighted_UL_SINR} by averaging the expectations over channel realizations. We next provide an explicit closed-form expression for the UL SE, applicable for any choice of weights.
\begin{lem}\label{lemm:UL_SINR_DTDD}
	The UL SINR of the $k$th UE with ZF-combining can be written as
	\begin{align*}
	\eta_{\mathsf{u},k}=\frac{(N-\tau_{p})\mathcal{E}_{\mathsf{u},k}(\sum\nolimits_{m\in\mathcal{A}_{\mathsf{u}}}\omega_{mk}^{*}\alpha_{mk}^{2})^2}{{\tt EST}_{\mathsf{u},k}+{\tt MUI}_{\mathsf{u},k}+{\tt IAP}_{k}+N_{0}\sum\limits_{m\in\mathcal{A}_{\mathsf{u}}}|\omega_{mk}^{*}|^{2}\alpha_{mk}^{2}},
	\end{align*}
where ${\tt EST}_{\mathsf{u},k}$, ${\tt MUI}_{\mathsf{u},k}$ and ${\tt IAP}_{k}$ capture the error due to channel estimation, multi-UE interference from the pilot sharing UEs, and inter-AP CLI. These are respectively evaluated as
\begin{subequations}
	\begin{align}
		&{\tt EST}_{\mathsf{u},k}=\sum\limits_{k'\in\mathcal{U}_{\mathsf{u}}}\mathcal{E}_{\mathsf{u},k'}\hspace*{-2mm}\sum\limits_{m\in\mathcal{A}_{\mathsf{u}}}|\omega_{mk}^*|^2\alpha_{mk}^2(\beta_{mk'}-\alpha_{mk'}^2),\\
&{\tt MUI}_{\mathsf{u},k}=(N-\tau_{p})\sum\limits_{i\in\mathcal{P}_{l(k)}\backslash k}\hspace*{-1mm}\mathcal{E}_{\mathsf{u},i}\left(\sum\nolimits_{m\in\mathcal{A}_{\mathsf{u}}}\omega_{mk}^{*}\alpha_{mi}^{2}\right)^2\hspace*{-2.5mm},
\\
&{\tt IAP}_{k}=N\sum\limits_{n\in\mathcal{U}_{\mathsf{d}}}\sum\limits_{m\in\mathcal{A}_{\mathsf{u}}}\sum\limits_{j\in\mathcal{A}_{\mathsf{d}}}\mathcal{E}_{\mathsf{d}}\kappa_{jn}^2\zeta_{mj}^{\mathsf{InAP}}|\omega_{mk}^{*}|^{2}\alpha_{mk}^2.
	\end{align}
\end{subequations}
\end{lem}
\begin{proof}
\textcolor{black}{Follows by using similar techniques as the proof of  Lemma~\ref{lem:optimal_weights}.}
\end{proof}

Note that the above result is also valid when equal weights are applied at the CPU, i.e., $\omega_{mk}=\frac{1}{\lvert\mathcal{A}_{\mathsf{u}}\rvert}, \forall k, \forall m$. However, as will be shown in Fig.~\ref{fig:UL_power_control}, optimal weighting yields significantly improved sum SE compared to equal weighting. Next, we present the DL SINR with ZF precoding.

	 \begin{lem}\label{lem:DL_SE_DTDD}
	 	The DL SINR of the $n$th UE with ZF precoding is 	 
		\begin{align}\label{eq:DL_SINR_DTDD}
	 		\eta_{\mathsf{d},n}=\frac{(N-\tau_{p})\mathcal{E}_{\mathsf{d}}\left(\sum\nolimits_{j\in\mathcal{A}_{\mathsf{d}}}\alpha_{jn}\kappa_{jn}\right)^2}{{\tt EST}_{\mathsf{d},n}+{\tt MUI}_{\mathsf{d},n}+{\tt IUE}_{n}+N_{0}},
	 	\end{align}
 	where  ${\tt EST}_{\mathsf{d}, n}$, ${\tt MUI}_{\mathsf{d}, n}$, and ${\tt IUE}_{n}$ capture the interference due to channel estimation error, the DL multi-UE interference from pilot sharing UEs, and the UL UE to DL UE CLI. These are respectively evaluated as
 	\begin{subequations}
 		\begin{align}
 			&{\tt EST}_{\mathsf{d},n}=\sum\nolimits_{q\in\mathcal{U}_{\mathsf{d}}}\sum\nolimits_{j\in\mathcal{A}_{\mathsf{d}}}\mathcal{E}_{\mathsf{d}}\kappa_{jq}^2(\beta_{jn}-\alpha_{jn}^2),\label{eq:dl_err}\\
 			&{\tt MUI}_{\mathsf{d}, n}=(N-\tau_{p})\mathcal{E}_{\mathsf{d}}\sum\nolimits_{ q\in\mathcal{P}_{l(n)}\backslash n}\hspace*{-.1cm}\left(\sum\nolimits_{j\in\mathcal{A}_{\mathsf{d}}}\hspace*{-.1cm}\kappa_{jq}\alpha_{jn}\right)^2,\hspace*{-.14cm}\label{eq:dl_ncoh}
 			\\&{\tt IUE}_{n}=\sum\nolimits_{{k\in\mathcal{U}_{\mathsf{u}}}}\mathcal{E}_{\mathsf{u},k}\epsilon_{nk}\label{eq:dl_UEUE}
 			.\end{align}
 	\end{subequations}	
	 \end{lem}
	 	\begin{proof}\textcolor{black}{
		Here, we provide the important steps involved in the proof. The beamforming gain, $|\textstyle{\sum\nolimits_{j\in\mathcal{A}_{\mathsf{d}}}}\kappa_{jn}\sqrt{{\mathcal{E}_{\mathsf{d}}}}\mathbf{f}_{jn}^{H}\mathbf{p}_{jn}|^{2}$, after substituting for $\mathbf{p}_{jn}$, with $\mathbb{E}[\|\mathbf{Z}_{\mathsf{d},j}(\mathbf{Z}_{\mathsf{d},j}^{H}\mathbf{Z}_{\mathsf{d},j})^{-1}\mathbf{e}_{l(n)}\|^{2}]={c_{jn}}/{(N-\tau_{p})}$ and  $\hat{\mathbf{f}}_{jn}=c_{jn}\sqrt{\tau_{p}\mathcal{E}_{p,j}}\beta_{jn}\mathbf{Z}_{\mathsf{d},j}\mathbf{e}_{l(n)}$, can be evaluated as follows: 
\begin{align}
	&\textstyle{\mathcal{E}_{\mathsf{d}}\textstyle{|\sum\nolimits_{j\in\mathcal{A}_{\mathsf{d}}}}\kappa_{jn}\sqrt{({N-\tau_{p})}{c_{jn}^{-1}}}\hat{\mathbf{f}}_{jn}^{H}\mathbf{Z}_{\mathsf{d},j}\left(\mathbf{Z}_{\mathsf{d},j}^{H}\mathbf{Z}_{\mathsf{d},j}\right)^{-1}\mathbf{e}_{l(n)}|^{2}}\notag\\&=\textstyle{(N-\tau_{p})\mathcal{E}_{\mathsf{d}}(\sum\nolimits_{j\in\mathcal{A}_{\mathsf{d}}}\kappa_{jn}\alpha_{jn})^{2}}.
\end{align}
Next, from the denominator of~\eqref{eq:DL_SINR}, we can rewrite $\mathcal{E}_{\mathsf{d}}({\tt var}\lbrace\boldsymbol{\kappa}_{n}^{T}\mathbf{d}_{nn}\rbrace+\sum\nolimits_{q\in \mathcal{U}_{\mathsf{d}}\backslash n}\mathbb{E}[\lvert\boldsymbol{\kappa}_{q}^{T}\mathbf{d}_{nq}\rvert^{2}])$ as $\mathcal{E}_{\mathsf{d}}(\sum\nolimits_{q\in \mathcal{U}_{\mathsf{d}}}\mathbb{E}[\lvert\boldsymbol{\kappa}_{q}^{T}\mathbf{d}_{nq}\rvert^{2}]
-|\mathbb{E}[\boldsymbol{\kappa}_{n}^{T}\mathbf{d}_{nn}]|^2)$. Substituting for $\mathbf{d}_{nq}$ and $\boldsymbol{\kappa}_{q}$, we get
 $\sum\nolimits_{q\in \mathcal{U}_{\mathsf{d}}}\mathbb{E}[\lvert\boldsymbol{\kappa}_{q}^{T}\mathbf{d}_{nq}\rvert^{2}]=\sum\nolimits_{q\in\mathcal{U}_{\mathsf{d}}}\mathbb{E}[|\sum\nolimits_{j\in\mathcal{A}_{\mathsf{d}}}\kappa_{jq}\mathbf{f}_{jn}^{T}\mathbf{p}_{jq}|^2]$, which can be evaluated as $(N-\tau_{p})\sum\nolimits_{q\in\mathcal{P}_{l(n)}}(\sum\nolimits_{j\in\mathcal{A}_{\mathsf{d}}}\kappa_{jq}\alpha_{jn})^2+\sum\nolimits_{q\in\mathcal{U}_{\mathsf{d}}}\sum\nolimits_{j\in\mathcal{A}_{\mathsf{d}}}\kappa_{jq}(\beta_{jn}-\alpha_{jn}^2).$ Then, with $|\mathbb{E}[\boldsymbol{\kappa}_{n}^{T}\mathbf{d}_{nn}]|^2=(N-\tau_{p})(\sum\nolimits_{j\in\mathcal{A}_{\mathsf{d}}}\kappa_{jn}\alpha_{jn})^2$, we get~\eqref{eq:dl_err} and~\eqref{eq:dl_ncoh}. 
  }
	 	\end{proof}
		Lemma~\ref{lemm:UL_SINR_DTDD} and Lemma~\ref{lem:DL_SE_DTDD} explicitly capture the dependence of the SINRs on the UL-DL power control coefficients and underlying UL-DL AP sets. Thus, the sum UL-DL SE, given by~\eqref{eq:sum_UL_DL_SE_DTDD} below, too, depends on the choice of these parameters.
	 	\begin{multline}\label{eq:sum_UL_DL_SE_DTDD}
	 		{\mathcal{R}_{\mathsf{s}}(\mathcal{A}_{\mathsf{s}},\boldsymbol{\kappa},\boldsymbol{\mathcal{E}}_{\mathsf{u}})={\frac{\tau-\tau_{p}}{\tau}\Big[\sum\nolimits_{k\in\mathcal{U}_{\mathsf{u}}}\log[1+\eta_{\mathsf{u},k}]}}\\{+{\sum\nolimits_{n\in\mathcal{U}_{\mathsf{d}}}\log[1+\eta_{\mathsf{d},n}]\Big]}}, 
	 	\end{multline}
	 	where $\boldsymbol{\kappa}\triangleq[\boldsymbol{\kappa}_{1}^{T},\ldots,\boldsymbol{\kappa}_{n}^{T},\ldots, \boldsymbol{\kappa}_{|\mathcal{U}|_{\mathsf{d}}}^{T}]^T\in\mathbb{C}^{|\mathcal{U}_{\mathsf{d}}||\mathcal{A}_{\mathsf{d}}|}$~\textcolor{black}{(recall $\boldsymbol{\kappa}_{n} = [\kappa_{1n},\kappa_{2n},\ldots,\kappa_{|\mathcal{A}_{\mathsf{d}}|n}]^{T}$)}, $\boldsymbol{\mathcal{E}}_{\mathsf{u}}\triangleq\left[\mathcal{E}_{\mathsf{u},1},\ldots,\mathcal{E}_{\mathsf{u},|\mathcal{U}_{\mathsf{u}}|}\right]\in\mathbb{C}^{|\mathcal{U}_{\mathsf{u}}|}$, and $\mathcal{A}_{\mathsf{s}}=\mathcal{A}_{\mathsf{u}}\cup\mathcal{A}_{\mathsf{d}}\subseteq \mathcal{A}$.  
Next, we optimize $\boldsymbol{\kappa}$, $\boldsymbol{\mathcal{E}}_{\mathsf{u}}$, and find a schedule for UL/DL APs in $\mathcal{A}_{\mathsf{s}}$ to maximize $\mathcal{R}_{\mathsf{s}}(.)$.

	 \section{Joint Scheduling and Power Control}\label{sec:Scheduling_Power_Allocation}
	 Here, we aim to solve the following problem:
	 	 	\begin{subequations}\label{eq:problem_AP_powercontrol}
	 \begin{align}
	 	\underset{\left\{\mathcal{A}_{\mathsf{s}},\kappa_{jn},\mathcal{E}_{\mathsf{u},k}\right\}}{\max}\quad& \mathcal{R}_{\mathsf{s}}\left(\mathcal{A}_{\mathsf{s}},\boldsymbol{\kappa},\boldsymbol{\mathcal{E}}_{\mathsf{u}}\right)\label{eq:sum_R_maximization}
	 	\\\mathrm{subject~to}\quad&\mathcal{A}_{\mathsf{u}}\cup\mathcal{A}_{\mathsf{d}}=\mathcal{A}_{\mathsf{s}}\subseteq\mathcal{A}; \mathcal{A}_{\mathsf{u}}\cap\mathcal{A}_{\mathsf{d}}=\emptyset;\label{eq:set_constraints} \\&
	 	0\leq\mathcal{E}_{\mathsf{u},k}\leq\mathcal{E}_{\mathsf{u}}, \forall k\in\mathcal{U}_{\mathsf{u}};\label{eq:UL_power_constraint} \\& \sum\nolimits_{\substack{ q\in\mathcal{U}_{\mathsf{d}}}}\kappa_{jq}^2\leq1, \forall j\in\mathcal{A}_{\mathsf{d}}\label{eq:per_AP_power}.
	 \end{align}
	 	 	\end{subequations}
With $\mathcal{R}_{\mathsf{s}}(\cdot)$ given by~\eqref{eq:sum_UL_DL_SE_DTDD}, the above problem is non-convex and NP-hard. We next present our solution.

\subsubsection{AP-scheduling}
Let us define $\overline{\mathcal{R}}_{\mathsf{s}}\left(\mathcal{A}_{\mathsf{s}}\right)=\max\limits_{\kappa_{jn},\mathcal{E}_{\mathsf{u},k}}\quad \mathcal{R}_{\mathsf{s}}\left(\mathcal{A}_{\mathsf{s}},\boldsymbol{\kappa},\boldsymbol{\mathcal{E}}_{\mathsf{u}}\right)$ as the sum UL-DL SE attained via optimizing the power allocation coefficients when the underlying set of APs is $\mathcal{A}_{\mathsf{s}}=\mathcal{A}_{\mathsf{u}}\cup\mathcal{A}_{\mathsf{d}}$. Now, the UL and DL power control coefficients for a sub-set of APs,  $\mathcal{A}_{\mathsf{s}}$, may not be optimal when we add one more AP in either UL or DL, say $\{j_{\mathsf{mode}}\}, \mathsf{mode}\in\{\mathsf{u},\mathsf{d}\}$, to $\mathcal{A}_{\mathsf{s}}$. Thus, we need to reoptimize the power control coefficients for $\mathcal{A}_{\mathsf{s}}\cup\{j_{\mathsf{mode}}\}$. 
The following proposition holds when we optimize the power control coefficients for both $\mathcal{A}_{\mathsf{s}}$ and $\mathcal{A}_{\mathsf{s}}\cup\{j_{\mathsf{mode}}\}$.
\begin{prop}
	For two sets of scheduled UL and DL APs, $\mathcal{A}_{\mathsf{s}}$ and $\mathcal{A}_{\mathsf{t}}$, where $\mathcal{A}_{\mathsf{s}}\subseteq\mathcal{A}_{\mathsf{t}}$, $\overline{\mathcal{R}}_{\mathsf{s}}\left(\mathcal{A}_{\mathsf{s}}\right)\leq\overline{\mathcal{R}}_{\mathsf{s}}\left(\mathcal{A}_{\mathsf{t}}\right)$, where, for both $\mathcal{A}_{\mathsf{s}}$ and $\mathcal{A}_{\mathsf{t}}$, the UL and DL power control coefficients are optimized to maximize the sum UL-DL SE.
\end{prop}

Thus, $\overline{\mathcal{R}}_{\mathsf{s}}\left(\mathcal{A}_{\mathsf{s}}\right)$ is a monotonically non-decreasing function of the underlying scheduled AP set.  
Motivated by this, we schedule APs one-by-one, where, at each iteration, we schedule an AP and its associated mode such that incremental gain in the (optimized) sum UL-DL SE is maximized. The recipe is sketched in Algorithm~\ref{algo:GA_submod}. Our algorithm procures the AP schedule in polynomial time~($\mathcal{O}(M)$), whereas exhaustive search-based AP scheduling requires optimization of the sum UL-DL SEs over all $2^M$ AP configurations, which is not a scalable approach. Such greedy approaches have been previously used in the antenna selection literature~\cite{Antenna_selection}.

\begin{algorithm}[t!]
	\DontPrintSemicolon
	\SetNlSty{textbf}{}{:}
	\setstretch{.3}
	%\SetAlgoLined
	\KwIn{$\mathcal{A}$: the set of all AP indices}
	\KwInt{$\mathcal{A}_\mathsf{u}=\mathcal{A}_\mathsf{d}=\emptyset$
	$\mathcal{A}_\mathsf{s}=\mathcal{A}_\mathsf{u}\cup\mathcal{A}_\mathsf{d}$}
	
	\While{ $\mathcal{A}_\mathsf{s}^{\mathsf{c}}\neq\emptyset$ }
	{
		Evaluate: \begin{align*}\begin{cases}i^{\star}_\mathsf{u}=\arg \max\limits_{i\in \mathcal{A}_s^{\mathsf{c}}} \overline{\mathcal{R}}_\mathsf{s} (\mathcal{A}_\mathsf{u}\cup\{i_\mathsf{u}\})\\
		i^{\star}_\mathsf{d}=\arg \max\limits_{i\in \mathcal{A}_\mathsf{s}^{\mathsf{c}}} \overline{\mathcal{R}}_\mathsf{s} (\mathcal{A}_\mathsf{d}\cup\{i_\mathsf{d}\})\end{cases}\end{align*}\;				
		\If{$\overline{\mathcal{R}}_\mathsf{s} (\mathcal{A}_\mathsf{u}\cup\{i_\mathsf{u}^{\star}\})\ge\overline{\mathcal{R}}_\mathsf{s} (\mathcal{A}_\mathsf{d}\cup\{i_\mathsf{d}^{\star}\})$}{Update: $\mathcal{A}_\mathsf{s}\leftarrow\mathcal{A}_\mathsf{u}=\mathcal{A}_\mathsf{u}\cup\{i_\mathsf{u}^{\star}\}$\; 
			\Else{Update: $\mathcal{A}_\mathsf{s}\leftarrow\mathcal{A}_\mathsf{d}=\mathcal{A}_\mathsf{d}\cup \{i_\mathsf{d}^{\star}\}\;$}
		}
				
	}
		\caption{UL and DL Mode Selection}\label{algo:GA_submod}
\end{algorithm}

We next address the power control problem given an AP schedule $\mathcal{A}_{\mathsf{s}}$. Now, as the UL power control coefficients affect the DL SE via inter-UE CLI and the DL power control coefficients affect the UL SE via inter-AP CLI, the joint optimization of UL-DL sum SE is still a non-convex and prohibitively complex task. In the next subsections, we present an alternating optimization approach, where we optimize the DL power control coefficients to maximize the DL sum SE given the UL power control coefficients, and vice-versa, until convergence.

\subsubsection{Uplink power control}\label{sec:UL_power_control}
First, we note that the UL power control coefficients influence the sum UL-DL SE predominantly via the sum UL SE. This is because, although the InUE term in the sum DL SE depends on UL transmit powers $\mathcal{E}_{u,k}, \forall k\in\mathcal{U}_{\mathsf{u}}$, it does not scale with the number of antennas (see~\eqref{eq:dl_UEUE}). Thus, the DL multi-UE and DL coherent interference terms dominate the InUE term in the sum DL SE. Hence, for mathematical tractability, we consider the sum UL SE maximization to optimize $\mathcal{E}_{u,k}, \forall k\in\mathcal{U}_{u}$. This simplifies the joint UL-DL SE maximization and comes with closed-form update equations as well as convergence guarantees for the sum UL SE maximization sub-problem. We use a similar approach for DL power control, by maximizing the DL SE.

Second, we note that, with SINR-optimal weighting at the CPU, since the weights maximize the received SINR at the CPU for each UE, they also maximize the sum UL SE. However, UEs' transmit powers can be optimized to improve the sum UL SE. Also, the recipe for power control presented here can be applied to any choice of weights at the CPU. 

For convenience, we rewrite the problem as follows:
\begin{subequations}
	\begin{align}
		\max_{\boldsymbol{\mathcal{E}}_{\mathsf{u}}} \quad & \sum\nolimits_{k\in\mathcal{U}_{\mathsf{u}}}\log\left(1+\frac{\mathsf{G}_{\mathsf{u},k}\left(\boldsymbol{\mathcal{E}}_{\mathsf{u}}\right)}{\mathsf{I}_{\mathsf{u},k}\left(\boldsymbol{\mathcal{E}}_{\mathsf{u}}\right)}\right),\label{eq:UL_R_maximization}\\\mathrm{subject~to}\quad& 0\leq \mathcal{E}_{\mathsf{u},k}\leq\mathcal{E}_{\mathsf{u}}, \forall k\in\mathcal{U}_{\mathsf{u}},
	\end{align}\label{eq:UL_R_maximization_full}
\end{subequations}
where, with the help of  Lemma~\ref{lemm:UL_SINR_DTDD}, we define  $\mathsf{G}_{\mathsf{u},k}(\boldsymbol{\mathcal{E}}_{\mathsf{u}})=(N-\tau_{p})\mathcal{E}_{\mathsf{u},k}\left(\sum\nolimits_{m\in\mathcal{A}_{\mathsf{u}}}\omega_{mk}^{*}\alpha_{mk}^{2}\right)^2$ and $\mathsf{I}_{\mathsf{u},k}(\boldsymbol{\mathcal{E}}_{\mathsf{u}})={\tt EST}_{\mathsf{u},k}+{\tt MUI}_{\mathsf{u},k}+\sigma_{\mathsf{eff.},\mathsf{u},k}^2$, with $\sigma_{\mathsf{eff.},\mathsf{u},k}^2\triangleq\left({\tt IAP}_{k}+N_{0}\sum\nolimits_{m\in\mathcal{A}_{\mathsf{u}}}|\omega_{mk}^{*}|^{2}\alpha_{mk}^{2}\right)$ being the effective noise that is independent of the UL transmit powers. Now, we recognize that~\eqref{eq:UL_R_maximization} is a sum of the logarithm of ratios, which can be converted into a convex problem via FP using a few auxiliary variables that can be iteratively solved in closed form. Further, FP guarantees that the optimal value of the objective and the variables that attain the optimum are the same for both the original objective function and the transformed surrogate objective function. This makes FP an excellent choice for the problem at hand. We next present the recipe in the context of our work.

To this end, we first introduce a set of auxiliary variables $\boldsymbol{\varpi}_{\mathsf{u}}\triangleq\left\{ \varpi_{\mathsf{u},1},\ldots,\varpi_{\mathsf{u},\lvert\mathcal{U}_{\mathsf{u}}\rvert} \right\}$, and formulate an \emph{equivalent} problem of $\max\nolimits_{\boldsymbol{\mathcal{E}}_{\mathsf{u}}, \boldsymbol{\varpi}_{\mathsf{u}}}\quad f\left(\boldsymbol{\mathcal{E}}_{\mathsf{u}}, \boldsymbol{\varpi}_{\mathsf{u}}\right)$, with  
\begin{align}\label{eq:equivalent_UL}
	f\left(\boldsymbol{\mathcal{E}}_{\mathsf{u}}, \boldsymbol{\varpi}_\mathsf{u}\right)\triangleq& \sum\nolimits_{k\in\mathcal{U}_{\mathsf{u}}}\ln\left(1+\varpi_{ \mathsf{u},k}\right)-\sum\nolimits_{k\in\mathcal{U}_{\mathsf{u}}}\varpi_{\mathsf{u},k}\notag\\&+\sum\nolimits_{k\in\mathcal{U}_{\mathsf{u}}}\frac{(1+\varpi_{\mathsf{u},k})\mathsf{G}_{\mathsf{u},k}\left(\boldsymbol{\mathcal{E}}_{\mathsf{u}}\right)}{\mathsf{G}_{\mathsf{u},k}\left(\boldsymbol{\mathcal{E}}_{\mathsf{u}}\right)+\mathsf{I}_{\mathsf{u},k}\left(\boldsymbol{\mathcal{E}}_{\mathsf{u}}\right)}.
\end{align}
and the same constraints are same as in~\eqref{eq:UL_R_maximization}.  Using results available in~\cite{FP_I}, we can show that the above two problems are 
{equivalent} in the sense that $\mathcal{E}_{u,k}, \forall k,$ are the solution to~\eqref{eq:UL_R_maximization_full} if and only if they are also the solution to~\eqref{eq:equivalent_UL}, and further, the maximum value of the objective in~\eqref{eq:equivalent_UL} and~\eqref{eq:UL_R_maximization} are the same. 
Now, to maximize $f\left(\boldsymbol{\mathcal{E}}_{\mathsf{u}}, \boldsymbol{\varpi}_{\mathsf{u}}\right)$, alternately optimize $\boldsymbol{\mathcal{E}}_{\mathsf{u}}$ and  
$\boldsymbol{\varpi}_{\mathsf{u}}$, while keeping the other variable constant. We observe that $f\left(\boldsymbol{\mathcal{E}}_{\mathsf{u}}, \boldsymbol{\varpi}_\mathsf{u}\right)$ is a concave differentiable function over $\boldsymbol{\varpi}_\mathsf{u}$ when $\boldsymbol{\mathcal{E}}_{\mathsf{u}}$ is fixed, say $\boldsymbol{\mathcal{E}}_{\mathsf{u}}^{\tt iter}$. Thus, 
 $\boldsymbol{\varpi}^{{\tt iter}+1}_{\mathsf{u}}=\arg\max\limits_{ \boldsymbol{\varpi}_{\mathsf{u}}} f\left(\boldsymbol{\mathcal{E}}_{\mathsf{u}}^{\tt iter}, \boldsymbol{\varpi}_{\mathsf{u}}\right)$ can be optimally determined by setting $\frac{\partial f\left(\boldsymbol{\mathcal{E}}_{\mathsf{u}}, \boldsymbol{\varpi}_\mathsf{u}\right)}{\partial \varpi_{\mathsf{u},k}}=0$ for each $\varpi_{\mathsf{u},k}$. This yields $\varpi_{\mathsf{u},k}^{{\tt iter}+1}=\frac{\mathsf{G}_{\mathsf{u},k}\left(\boldsymbol{\mathcal{E}}_{\mathsf{u}}^{\tt iter}\right)}{\mathsf{I}_{\mathsf{u},k}\left(\boldsymbol{\mathcal{E}}_{\mathsf{u}}^{\tt iter}\right)}, \forall k\in\mathcal{U}_{\mathsf{u}}$. Now, for fixed $\boldsymbol{\varpi}_{\mathsf{u}}$, the first and the second term of~\eqref{eq:equivalent_UL} are constants. The third term, i.e., $\sum\limits_{k\in\mathcal{U}_{\mathsf{u}}}\dfrac{(1+\varpi_{\mathsf{u},k})\mathsf{G}_{\mathsf{u},k}\left(\boldsymbol{\mathcal{E}}_{\mathsf{u}}\right)}{\mathsf{G}_{\mathsf{u},k}\left(\boldsymbol{\mathcal{E}}_{\mathsf{u}}\right)+\mathsf{I}_{\mathsf{u},k}\left(\boldsymbol{\mathcal{E}}_{\mathsf{u}}\right)}$, is in the sum of ratios form, which can be reformulated using a quadratic transform as $\max_{\boldsymbol{\mathcal{E}}_{\mathsf{u}}, \boldsymbol{\tilde{\varpi}}} f(\boldsymbol{\mathcal{E}}_{\mathsf{u}}, \boldsymbol{\tilde{\varpi}}_{\mathsf{u}})$, where 
\begin{align}\label{eq:Lagrage_dual_UL}
	f(\boldsymbol{\mathcal{E}}_{\mathsf{u}}, \boldsymbol{\tilde{\varpi}}_{\mathsf{u}})\triangleq& \sum\nolimits_{k\in\mathcal{U}_{\mathsf{u}}}2\tilde{\varpi}_{\mathsf{u},k}\sqrt{(1+\varpi_{\mathsf{u},k})\mathsf{G}_{\mathsf{u},k}\left(\boldsymbol{\mathcal{E}}_{\mathsf{u}}\right)}\notag\\&\hspace*{-17mm}-\sum\nolimits_{k\in\mathcal{U}_{\mathsf{u}}}\tilde{\varpi}_{\mathsf{u}, k}^{2}\left(\mathsf{G}_{\mathsf{u},k}\left(\boldsymbol{\mathcal{E}}_{\mathsf{u}}\right)+\mathsf{I}_{\mathsf{u},k}\left(\boldsymbol{\mathcal{E}}_{\mathsf{u}}\right)\right)+\mathsf{c}_{\boldsymbol{\varpi}_{\mathsf{u}}},
\end{align}
and $\boldsymbol{\tilde{\varpi}}_{\mathsf{u}}\triangleq\left\{ \tilde{\varpi}_{\mathsf{u},1},\ldots,\tilde{\varpi}_{\mathsf{u},\lvert\mathcal{U}_{\mathsf{u}}\rvert} \right\}$ being a new set of auxiliary variables and $\mathsf{c}_{\boldsymbol{\varpi}_{\mathsf{u}}}$ is a constant dependent only on $\boldsymbol{\varpi}_{\mathsf{u}}$.
We can solve for $\tilde{\varpi}_{\mathsf{u},k}$ and $\boldsymbol{\mathcal{E}}_{\mathsf{u}}$ via  partial differentiation of ~\eqref{eq:Lagrage_dual_UL}, for fixed $\varpi_{\mathsf{u},k}$. We summarize the overall recipe in Algorithm~\ref{algo:UL_power_control}. The derivation of updates for $\tilde{\varpi}_{\mathsf{u},{k}}^{{\tt iter}+1}$ and $\mathcal{E}_{\mathsf{u},k}^{{\tt iter} in +1}$~\eqref{eq:update_E} use simple algebraic manipulations, and are omitted for brevity.
\begin{prop}\label{prop:conv_UL}
The UL power control Algorithm~\ref{algo:UL_power_control} is convergent in the objective since $f\left(\boldsymbol{\mathcal{E}}_{\mathsf{u}}, \boldsymbol{\varpi}_{\mathsf{u}}\right)$ in~\eqref{eq:equivalent_UL} is bounded above and monotonically non-decreasing after each iteration.
\end{prop}
\begin{proof}
First, optimizing $f\left(\boldsymbol{\mathcal{E}}_{\mathsf{u}}, \boldsymbol{\varpi}_{\mathsf{u}}\right)$ is equivalent to optimizing  
$f\left(\boldsymbol{\mathcal{E}}_{\mathsf{u}}, \boldsymbol{\varpi}_{\mathsf{u}},\boldsymbol{\tilde{\varpi}}_{\mathsf{u}}\right)= \sum\nolimits_{k\in\mathcal{U}_{\mathsf{u}}}\ln\left(1+\varpi_{ \mathsf{u},k}\right)-\sum\nolimits_{k\in\mathcal{U}_{\mathsf{u}}}\varpi_{\mathsf{u},k}+f(\boldsymbol{\mathcal{E}}_{\mathsf{u}}, \boldsymbol{\tilde{\varpi}}_{\mathsf{u}})$ due to the equivalence of $\sum\nolimits_{k\in\mathcal{U}_{\mathsf{u}}}\frac{(1+\varpi_{\mathsf{u},k})\mathsf{G}_{\mathsf{u},k}\left(\boldsymbol{\mathcal{E}}_{\mathsf{u}}\right)}{\mathsf{G}_{\mathsf{u},k}\left(\boldsymbol{\mathcal{E}}_{\mathsf{u}}\right)+\mathsf{I}_{\mathsf{u},k}\left(\boldsymbol{\mathcal{E}}_{\mathsf{u}}\right)}$ and $f(\boldsymbol{\mathcal{E}}_{\mathsf{u}}, \boldsymbol{\tilde{\varpi}}_{\mathsf{u}})$. Next, we observe that $f\left(\boldsymbol{\mathcal{E}}_{\mathsf{u}}^{{\tt iter}+1}, \boldsymbol{\varpi}_{\mathsf{u}}^{{\tt iter}+1},\boldsymbol{\tilde{\varpi}}_{\mathsf{u}}^{{\tt iter}+1}\right)\geq f\left(\boldsymbol{\mathcal{E}}_{\mathsf{u}}^{\tt iter}, \boldsymbol{\varpi}_{\mathsf{u}}^{\tt iter},\boldsymbol{\tilde{\varpi}}_{\mathsf{u}}^{\tt iter}\right)$,
because at iteration indexed by ${\tt iter}+1$, each of these variables optimally solves the equivalent \emph{convex surrogate} $f(\cdot)$ while keeping the other two variables  fixed. 
Since $f\left(\boldsymbol{\mathcal{E}}_{\mathsf{u}}, \boldsymbol{\varpi}_{\mathsf{u}},\boldsymbol{\tilde{\varpi}}_{\mathsf{u}}\right)$ has a finite upper bound, the algorithm is globally convergent.
\end{proof}
\begin{rem}
Once the UL transmit powers of all the UEs are decided, the CPU can use this information to refine the SINR optimal weights. Specifically, the CPU  initially finds $\boldsymbol{\omega}_{k}^{\mathsf{opt.}}$ considering equal power allocation. Once the UL transmit powers are optimized, CPU can reoptimize $\boldsymbol{\omega}_{k}^{\mathsf{opt.}}$  based on the received SINRs with $\mathcal{E}_{\mathsf{u},k}$ given by Algorithm~\ref{algo:UL_power_control}. This process can be repeated until convergence of the sum UL SE.
\end{rem}
\begin{rem}
We observe that lines~\ref{algo:UL_varpi},~\ref{algo:UL_varpi_tilde}, and~\ref{algo:UL_E} in Algorithm~\ref{algo:UL_power_control} are independent of the combining scheme, and are applicable even under maximal ratio combining, unlike~\cite{SPAWC_Virtual_Duplex} where the power control algorithm is tied to the specific combining scheme used. More importantly, the FP approach guarantees that the maximum objective of the surrogate and original objective functions are the same, unlike the lower-bound surrogate-based optimization in~\cite{SPAWC_Virtual_Duplex}. 
\end{rem}
\begin{algorithm}[!t]
       \SetNlSty{textbf}{}{:}
	\DontPrintSemicolon
	\KwIn{$\kappa_{jn}, \forall j\in\mathcal{A}_{\mathsf{d}}, n\in\mathcal{U}_{\mathsf{d}}$}
	\KwInt{$\boldsymbol{\mathcal{E}}_{\mathsf{u}}^{0}$, ${\tt iter}=0$}
	\While{$\lvert f\left(\boldsymbol{\mathcal{E}}_{\mathsf{u}}, \boldsymbol{\varpi}_{\mathsf{u}}\right)^{{\tt iter}+1}-f\left(\boldsymbol{\mathcal{E}}_{\mathsf{u}}, \boldsymbol{\varpi}_{\mathsf{u}}\right)^{\tt iter}\rvert\geq \delta_{\mathsf{u}}$}{
		Evaluate: $\varpi_{\mathsf{u},k}^{{\tt iter}+1}=\frac{\mathsf{G}_{\mathsf{u},k}\left(\boldsymbol{\mathcal{E}}_{\mathsf{u}}\right)}{\mathsf{I}_{\mathsf{u},k}\left(\boldsymbol{\mathcal{E}}_{\mathsf{u}}\right)}\lvert_{\boldsymbol{\mathcal{E}}_{\mathsf{u}}=\boldsymbol{\mathcal{E}}_{\mathsf{u}}^{\tt iter}}$\;\label{algo:UL_varpi}
		Evaluate: $\frac{\partial f(\boldsymbol{\mathcal{E}}_{\mathsf{u}}, \boldsymbol{\tilde{\varpi}}_{\mathsf{u}})}{\partial \tilde{\varpi}_{\mathsf{u}, k}}\Bigg\lvert_{\substack{\boldsymbol{\mathcal{E}}_{\mathsf{u}}=\boldsymbol{\mathcal{E}}_{\mathsf{u}}^{\tt iter},\\
				\varpi_{\mathsf{u},k}=\varpi_{\mathsf{u},k}^{\tt iter+1}, \forall k\in\mathcal{U}_{\mathsf{u}}}}=0 \Rightarrow$
		$\tilde{\varpi}_{\mathsf{u},{k}}^{{\tt iter}+1}=\frac{\sqrt{\left(1+\varpi_{\mathsf{u},k}^{{\tt iter}+1}\right)\mathsf{G}_{\mathsf{u},k}\left(\boldsymbol{\mathcal{E}}_{\mathsf{u}}^{\tt iter}\right)}}{\mathsf{G}_{\mathsf{u},k}\left(\boldsymbol{\mathcal{E}}_{\mathsf{u}}^{\tt iter}\right)+\mathsf{I}_{\mathsf{u},k}\left(\boldsymbol{\mathcal{E}}_{\mathsf{u}}^{\tt iter}\right)}$\;\label{algo:UL_varpi_tilde}
		Set:  $\frac{\partial f(\boldsymbol{\mathcal{E}}_{\mathsf{u}}, \boldsymbol{\tilde{\varpi}}_{\mathsf{u}})}{\partial \mathcal{E}_{\mathsf{u},k}}\Bigg\lvert_{\substack{\varpi_{\mathsf{u},k}=\varpi_{\mathsf{u},k}^{\tt iter+1},\\ \tilde{\varpi}_{\mathsf{u},{k}}=\tilde{\varpi}_{\mathsf{u},{k}}^{{\tt iter}+1}}}=0$ to obtain~\eqref{eq:update_E}\; \label{algo:UL_E}

Update: $\mathcal{E}_{\mathsf{u},k}^{{\tt iter}+1}=\min\left\{\eqref{eq:update_E}, \mathcal{E}_{\mathsf{u}}\right\}, \forall k\in\mathcal{U}_{\mathsf{u}}$\;
				Update: ${\tt iter}={\tt iter}+1$\;
	}
	\caption{Uplink Power Control}\label{algo:UL_power_control}
\end{algorithm}
\begin{figure*}[!t]
	\begin{align}
		\mathcal{E}_{\mathsf{u},k}^{{\tt iter}+1}=\frac{\left(\tilde{\varpi}_{\mathsf{u},k}^{{\tt iter}+1}\sqrt{\left(1+\varpi_{\mathsf{u},k}^{{\tt iter}+1}\right)}\right)^{2}(N-\tau_{p})\left(\sum\nolimits_{m\in\mathcal{A}_{\mathsf{u}}}\omega^{*}_{mk}\alpha_{mk}^2\right)^2}{\left[\sum\nolimits_{n\in\mathcal{U}_{\mathsf{u}}}(\tilde{\varpi}_{\mathsf{u},n}^{{\tt iter}+1})^{2}\left\{(N-\tau_{p})\left(\sum\nolimits_{m\in\mathcal{A}_{\mathsf{u}}}\omega^{*}_{mn}\alpha_{mk}^2\right)^2+\sum\nolimits_{m\in\mathcal{A}_{\mathsf{u}}}\lvert\omega^{*}_{mn}\rvert^{2}\alpha_{mn}^2(\beta_{mk}-\alpha_{mk}^2)\right\}\right]^{2}}, \forall k\in\mathcal{U}_{\mathsf{u}}.\label{eq:update_E}
	\end{align}
\end{figure*}

\subsubsection{DL power control}\label{sec:DL_power_control}
Now, in DL also, we can apply FP to optimize the DL power control coefficients. To this end, we first introduce some useful mathematical notation. Let $\mathbf{g}_{\mathsf{d},n}\triangleq\left[\mathbb{E}\left[\mathbf{f}_{1n}^{T}\mathbf{p}_{1n}\right],\mathbb{E}\left[\mathbf{f}_{2n}^{T}\mathbf{p}_{2n}\right],\ldots,\mathbb{E}\left[\mathbf{f}_{|\mathcal{A}_{\mathsf{d}}|n}^{T}\mathbf{p}_{|\mathcal{A}_{\mathsf{d}}|n}\right]\right]^T$ (which equals $(N-\tau_p)\left[\alpha_{1n},\alpha_{2n},\ldots,\alpha_{|\mathcal{A}_{\mathsf{d}}|n}\right]^T$ under ZF precoding.) Also, let $\left[ \mathbf{I}_{nq}\right]_{jj'}\triangleq\mathbb{E}\left[\mathbf{f}_{jn}^{T}\mathbf{p}_{jq}\mathbf{p}_{j'q}^{H}\mathbf{f}_{j'n}^{*}\right], \forall j,j'\in\mathcal{A}_{\mathsf{d}}$. With ZF precoding, $\mathbf{I}_{nq}$ has the following structure:
\begin{align}\label{eq:I_nq}
\hspace{-.3cm}[\mathbf{I}_{nq}]_{jj'}=\begin{cases}
(N-\tau_{p})\alpha^2_{jn}\lvert\langle\boldsymbol{\varphi}_{l(n)},\boldsymbol{\varphi}_{l(q)}\rangle\rvert^2\\\hspace{1cm}+(\beta_{jn}-\alpha^2_{jn}), \text{if}~j=j'\\
(N-\tau_{p})\alpha_{j'n}\alpha_{jn}\lvert\langle\boldsymbol{\varphi}_{l(n)},\boldsymbol{\varphi}_{l(q)}\rangle\rvert^2, \text{otherwise}.
\end{cases}\hspace*{-.5cm}
\end{align}	
Then,  $\mathbb{E}[|\sum\nolimits_{j\in\mathcal{A}_{\mathsf{d}}}\kappa_{jq}\mathbf{f}_{jn}^{T}\mathbf{p}_{jq}|^{2}]$ can be expressed as $\sum\limits_{j\in\mathcal{A}_{\mathsf{d}}}\sum\limits_{j'\in\mathcal{A}_{\mathsf{d}}}\kappa_{jq}\mathbb{E}\left[\mathbf{f}_{jn}^{T}\mathbf{p}_{jq}\mathbf{p}_{j'q}^{H}\mathbf{f}_{j'n}^{*}\right]\kappa_{j'q}=\boldsymbol{\kappa}_{q}^T \mathbf{I}_{nq}\boldsymbol{\kappa}_{q}$.
The effective DL noise, independent of DL power control coefficients, is defined as  $\sigma_{\mathsf{eff.},\mathsf{d}, n}^2={\tt IUE}_{n}+N_{0}$. 
Thus DL SINR becomes $\eta_{\mathsf{d},n}(\mathcal{A}_{\mathsf{s}})=\frac{\mathsf{G}_{\mathsf{d},n}\left(\boldsymbol{\kappa}\right)}{\mathsf{I}_{\mathsf{d},n}\left(\boldsymbol{\kappa}\right)}$,
with $\mathsf{G}_{\mathsf{d}, n}\left(\boldsymbol{\kappa}\right)=\left(\boldsymbol{\kappa}_{n}^T\mathbf{g}_{\mathsf{d},n}\right)^2$ and $\mathsf{I}_{\mathsf{d}, n}\left(\boldsymbol{\kappa}\right)\triangleq \sum\nolimits_{q\in\mathcal{U}_{\mathsf{d}}}\boldsymbol{\kappa}_{q}^T \mathbf{I}_{nq}\boldsymbol{\kappa}_{q}-\left(\boldsymbol{\kappa}_{q}^T\mathbf{g}_{\mathsf{d},n}\right)^2+\frac{1}{\mathcal{E}_{\mathsf{d}}}\sigma_{\mathsf{eff.}, \mathsf{d}, n}^2$. Observe that, substituting $\mathbf{g}_{\mathsf{d},n}$ and~\eqref{eq:I_nq}, we obtain Lemma~\ref{lem:DL_SE_DTDD}. Now, the sum DL SE maximization problem becomes:
\begin{subequations}
\begin{align}\label{eq:DL_SE_opt}
\max_{\kappa_{jn}}\quad&\sum\nolimits_{n\in\mathcal{U}_{\mathsf{d}}}\log\left(1+\frac{\mathsf{G}_{\mathsf{d}, n}\left(\boldsymbol{\kappa}\right)}{\mathsf{I}_{\mathsf{d}, n}\left(\boldsymbol{\kappa}\right)}\right)
\\\mathrm{subject~to}\quad & 
\sum\nolimits_{\substack{ q\in\mathcal{U}_{\mathsf{d}}}}\kappa_{jq}^2\leq1, \forall j\in\mathcal{A}_{\mathsf{d}}.
\end{align}
\end{subequations}
We apply the Lagrange-dual transform~\cite{FP_I} with auxiliary variables $\boldsymbol{\varpi}_{\mathsf{d}}=\left[\varpi_{\mathsf{d}, 1},\varpi_{\mathsf{d}, 2},\ldots,\varpi_{\mathsf{d},|\mathcal{U}_{\mathsf{d}}|}\right]$ to obtain:
\begin{subequations}
\begin{align}\label{eq:DL_SE_opt_equivalent}
\max_{\kappa_{jn}, \boldsymbol{\varpi}_{\mathsf{d}}} \, f(\boldsymbol{\kappa},\boldsymbol{\varpi}_{\mathsf{d}}) \triangleq &\sum\nolimits_{n\in\mathcal{U}_{\mathsf{d}}}\!\!\!\ln\left(1+\varpi_{\mathsf{d}, n}\right)-\sum\nolimits_{n\in\mathcal{U}_{\mathsf{d}}}\varpi_{\mathsf{d}, n}\notag\\&+\sum\nolimits_{n\in\mathcal{U}_{\mathsf{d}}}\frac{(1+\varpi_{\mathsf{d}, n})\mathsf{G}_{\mathsf{d},n}\left(\boldsymbol{\kappa}\right)}{\mathsf{G}_{\mathsf{d},n}\left(\boldsymbol{\kappa}\right)+\mathsf{I}_{\mathsf{d}, n}\left(\boldsymbol{\kappa}\right)}\\\mathrm{subject~to}\quad&
\sum\nolimits_{\substack{ q\in\mathcal{U}_{\mathsf{d}}}}\kappa_{jq}^2\leq1; \forall j\in\mathcal{A}_{\mathsf{d}}.
\end{align}
\end{subequations}
The problems in~\eqref{eq:DL_SE_opt} and in~\eqref{eq:DL_SE_opt_equivalent} are equivalent in the sense that $\boldsymbol{\kappa}$ is a solution of~\eqref{eq:DL_SE_opt} if and only if it is the solution of~\eqref{eq:DL_SE_opt_equivalent}~\cite{FP_I}. We now alternately optimize $\boldsymbol{\kappa}$ and $\boldsymbol{\varpi}_{\mathsf{d}}$.  For fixed $\boldsymbol{\kappa}$, $f(\boldsymbol{\kappa},\boldsymbol{\varpi}_{\mathsf{d}})$ is a concave differentiable function over $\boldsymbol{\varpi}_{\mathsf{d}}$. Thus, $\frac{\partial f(\boldsymbol{\kappa},\boldsymbol{\varpi}_{\mathsf{d}})}{\partial \varpi_{\mathsf{d}, n}}=0$ yields $\varpi_{\mathsf{d}, n}^{\mathsf{opt.}}=\frac{\mathsf{G}_{\mathsf{d},n}\left(\boldsymbol{\kappa}\right)}{\mathsf{I}_{\mathsf{d}, n}\left(\boldsymbol{\kappa}\right)}$. 
For fixed $\boldsymbol{\varpi}_{\mathsf{d}}$, the first and the second terms of $f(\boldsymbol{\kappa},\boldsymbol{\varpi}_{\mathsf{d}})$ are constants. Hence, to optimize $\boldsymbol{\kappa}$, we need to solve $\max_{\boldsymbol{\kappa}}\sum\nolimits_{n\in\mathcal{U}_{\mathsf{d}}}\frac{(1+\varpi_{\mathsf{d},n})\mathsf{G}_{\mathsf{d},n}\left(\boldsymbol{\kappa}\right)}{\mathsf{G}_{\mathsf{d},n}\left(\boldsymbol{\kappa}\right)+\mathsf{I}_{\mathsf{d}, n}\left(\boldsymbol{\kappa}\right)}$,
for which we use FP. To do so, we define  $\mathsf{\bar{G}}_{\mathsf{d},n}\left(\boldsymbol{\kappa}\right)\triangleq(1+\varpi_{\mathsf{d},n})\mathsf{G}_{\mathsf{d},n}\left(\boldsymbol{\kappa}\right),$ and $\mathsf{\bar{I}}_{\mathsf{d},n}\left(\boldsymbol{\kappa}\right)\triangleq \mathsf{G}_{\mathsf{d},n}\left(\boldsymbol{\kappa}\right)+\mathsf{I}_{\mathsf{d}, n}\left(\boldsymbol{\kappa}\right).$ Then, the equivalent problem is to maximize $\sum\nolimits_{n\in\mathcal{U}_{\mathsf{d}}}\frac{\mathsf{\bar{G}}_{\mathsf{d},n}\left(\boldsymbol{\kappa}\right)}{\mathsf{\bar{I}}_{\mathsf{d},n}\left(\boldsymbol{\kappa}\right)}$ subject to $\sum\nolimits_{\substack{ q\in\mathcal{U}_{\mathsf{d}}}}\kappa_{jq}^2\leq1.$
Now, the dual of $\sum\nolimits_{n\in\mathcal{U}_{\mathsf{d}}}\frac{\mathsf{\bar{G}}_{\mathsf{d},n}\left(\boldsymbol{\kappa}\right)}{\mathsf{\bar{I}}_{\mathsf{d},n}(\boldsymbol{\kappa})}$ is $\sum\nolimits_{n\in\mathcal{U}_{\mathsf{d}}}(2\tilde{\varpi}_{\mathsf{d}, n}\sqrt{\mathsf{\bar{G}}_{\mathsf{d},n}(\boldsymbol{\kappa})}-\tilde{\varpi}_{\mathsf{d}, n}^2\mathsf{\bar{I}}_{\mathsf{d},n}(\boldsymbol{\kappa}))$ with new auxiliary variables $\boldsymbol{\tilde{\varpi}}_{\mathsf{d}}=\left[\tilde{\varpi}_{\mathsf{d}, 1},\tilde{\varpi}_{\mathsf{d}, 2},\ldots,\tilde{\varpi}_{\mathsf{d}, |\mathcal{U}_{\mathsf{d}}|}\right]^{T}$~\cite[Corollary 1]{FP_I}.
Substituting, we obtain
\begin{subequations}
\begin{align}\label{eq:DL_auxilliary}
\max_{\boldsymbol{\kappa},\boldsymbol{\tilde{\varpi}}_{\mathsf{d}}}\, f(\boldsymbol{\kappa},\boldsymbol{\tilde{\varpi}}_{\mathsf{d}}) \triangleq &\sum\nolimits_{n\in\mathcal{U}_{\mathsf{d}}}\Big(2\tilde{\varpi}_{\mathsf{d}, n}\sqrt{(1+\varpi_{\mathsf{d},n})\left(\boldsymbol{\kappa}_{q}^T\mathbf{g}_{\mathsf{d},n}\right)^2}\notag\\&\hspace*{-12mm}-\tilde{\varpi}_{\mathsf{d}, n}^2\left(\sum\nolimits_{q\in\mathcal{U}_{\mathsf{d}}}\boldsymbol{\kappa}_{q}^T \mathbf{I}_{nq}\boldsymbol{\kappa}_{q}+\sigma_{\mathsf{eff.}, \mathsf{d}, n}^2\right)\Big),\\\mathrm{subject~to}\quad&
\sum\nolimits_{\substack{ q\in\mathcal{U}_{\mathsf{d}}}}\kappa_{jq}^2\leq1, \forall j \in\mathcal{A}_{\mathsf{d}}.
\end{align}
\end{subequations}
 Observe that $\boldsymbol{\varpi}_{\mathsf{d}}$ is already fixed. Then, for fixed $\boldsymbol{\kappa}$, $f(\boldsymbol{\kappa},\boldsymbol{\tilde{\varpi}}_{\mathsf{d}})$ is strongly concave with respect to $\boldsymbol{\tilde{\varpi}}_{\mathsf{d}}$ and thus, we can set $\frac{\partial f(\boldsymbol{\kappa},\boldsymbol{\tilde{\varpi}}_{\mathsf{d}})}{\partial \tilde{\varpi}_{\mathsf{d}, n}}=0$, leading to $\tilde{\varpi}_{\mathsf{d}, n}^{\text{opt.}}=\frac{\sqrt{1+\varpi_{\mathsf{d},n}}\mathbf{g}_{\mathsf{d},n}^{T}\boldsymbol{\kappa}_{n}}{\sum\nolimits_{q\in\mathcal{U}_{\mathsf{d}}}\boldsymbol{\kappa}_{q}^T \mathbf{I}_{nq}\boldsymbol{\kappa}_{q}+\sigma_{\mathsf{eff.}, \mathsf{d}, n}^2}$; $\forall n\in\mathcal{U}_{\mathsf{d}}$. 
 Finally, we obtain $\boldsymbol{\kappa}^{\text{opt.}}$  given $\boldsymbol{\varpi}_{\mathsf{d}}$ and $\boldsymbol{\tilde{\varpi}}_{\mathsf{d}}$ by solving
\begin{subequations}
\begin{align}
\max_{\boldsymbol{\kappa},\boldsymbol{\tilde{\varpi}}_{\mathsf{d}}}\quad &\sum\nolimits_{n\in\mathcal{U}_{\mathsf{d}}}\left(2\tilde{\varpi}_{\mathsf{d}, n}\sqrt{(1+\varpi_{\mathsf{d},n})\left(\boldsymbol{\kappa}_{n}^T\mathbf{g}_{\mathsf{d},n}\right)^2}\right.\notag\\&\left.\hspace*{-7mm}-\tilde{\varpi}_{\mathsf{d}, n}^2\left(\sum\nolimits_{q\in\mathcal{U}_{\mathsf{d}}}\boldsymbol{\kappa}_{q}^T \mathbf{I}_{nq}\boldsymbol{\kappa}_{q}+\sigma_{\mathsf{eff.}, \mathsf{d}, n}^2\right)\right)\label{eq:QCQP},\\\mathrm{subject~to}\quad& 
\sum\nolimits_{\substack{ q\in\mathcal{U}_{\mathsf{d}}}}\kappa_{jq}^2\leq1, \forall j\in\mathcal{A}_{\mathsf{d}}.
\end{align}
\end{subequations}
We observe that matrix $ \mathbf{I}_{nq}$ is positive semi-definite. Hence, $f(\boldsymbol{\kappa}) \triangleq \sum\nolimits_{n\in\mathcal{U}_{\mathsf{d}}}(2\tilde{\varpi}_{\mathsf{d}, n}\sqrt{(1+\varpi_{\mathsf{d},n})\left(\boldsymbol{\kappa}_{n}^T\mathbf{g}_{\mathsf{d},n}\right)^2}-\tilde{\varpi}_{\mathsf{d}, n}^2(\sum\nolimits_{q\in\mathcal{U}_{\mathsf{d}}}\boldsymbol{\kappa}_{q}^T \mathbf{I}_{nq}\boldsymbol{\kappa}_{q}+\sigma_{\text{eff.}}^2))$ is concave with respect to $\boldsymbol{\kappa}$. The above problem is a QCQP, which can be \emph{optimally} solved via ADMM (see~\cite[Chapter~$5$]{Boyd_ADMM}), yielding a closed form update for $\kappa_{jn}$,  which is also known to be efficient in terms of convergence for large dimensional problems compared to using off-the-shelf convex solvers. This is presented next. 

We reformulate~\eqref{eq:QCQP} with the help of new auxiliary variables $\boldsymbol{\pi}=\left[\pi_{jn}\right]_{j\in\mathcal{A}_d,n\in\mathcal{U}_d}$ as 
\begin{subequations}
\begin{align}
	\min_{\boldsymbol{\kappa},\boldsymbol{\pi}}\quad& \mathds{1}_{\boldsymbol{\kappa}}(\boldsymbol{\pi})-f(\boldsymbol{\kappa}),\label{eq:ADMM_formulation}\\
	\mathrm{subject~to}\quad& \kappa_{jn}=\pi_{jn}, \forall j\in\mathcal{A}_{\mathsf{d}}, n\in\mathcal{U}_{\mathsf{d}},
\end{align}
\end{subequations}
where $\mathds{1}_{\boldsymbol{\kappa}}(\boldsymbol{\pi})$ is defined as 
\begin{align}
	\mathds{1}_{\boldsymbol{\kappa}}(\boldsymbol{\pi})=\begin{cases}
		0,\quad \mathrm{if}\quad \|\boldsymbol{\pi}_{j,:}\|^2\leq 1, \forall j\in\mathcal{A}_{\mathsf{d}}\\
		\infty,\quad \mathrm{otherwise},
	\end{cases}
\end{align}
where $\boldsymbol{\pi}_{j,:}\triangleq\left[\pi_{j1},\pi_{j2},\ldots,\pi_{j\lvert\mathcal{U}_{\mathsf{d}}\rvert}\right]^{T}$.
Here, $	\mathds{1}_{\boldsymbol{\kappa}}(\boldsymbol{\pi})$ is an indicator whether the auxiliary variables satisfy the feasibility constraint $\|\boldsymbol{\pi}_{j,:}\|^2\leq 1$ corresponding to $\boldsymbol{\kappa}$. Essentially, $\boldsymbol{\pi}$ is a  copy of the main
optimization variable $\boldsymbol{\kappa}$ and should satisfy the same constraint. Let $\bar{\pi}_{jn}$ denote the scaled dual variables\footnote{In the ADMM terminology, the variables $\boldsymbol{\kappa}$ and $\boldsymbol{\pi}$  can be considered as the first and second blocks of primal variables, respectively. In ADMM, a set of dual variables are introduced for the equality constraint in~\eqref{eq:ADMM_formulation}, and they are to be updated in each iteration with the primal variables.} corresponding to the equality constraints $\pi_{jn} = \kappa_{jn}, \forall j\in\mathcal{A}_{\mathsf{d}}, n\in\mathcal{U}_{\mathsf{d}}$. Then, the augmented Lagrangian can be written as $\mathcal{L}\left(\boldsymbol{\kappa},\boldsymbol{\pi},\bar{\boldsymbol{\pi}}\right) = \mathds{1}_{\boldsymbol{\kappa}}(\boldsymbol{\pi})-f(\boldsymbol{\kappa}) +\frac{\delta_{\mathsf{p}}}{2}\sum\nolimits_{n\in\mathcal{U}_{\mathsf{d}}}\sum\nolimits_{j\in\mathcal{A}_{\mathsf{d}}}\left(\pi_{jn}-\kappa_{jn}+\bar{\pi}_{jn}\right)^2$, where $\delta_{\mathsf{p}}$ is a penalty parameter. We now update $\boldsymbol{\kappa}_{n}, \forall n\in\mathcal{U}_{\mathsf{d}},$ as the solution of 
\begin{align}\label{eq:sol_kappa}
\arg\min_{\boldsymbol{\kappa}}-f(\boldsymbol{\kappa}) +\frac{\delta_{\mathsf{p}}}{2}\sum\limits_{n\in\mathcal{U}_{\mathsf{d}}}\sum\limits_{j\in\mathcal{A}_{\mathsf{d}}}\left(\pi_{jn}-\kappa_{jn}+\bar{\pi}_{jn}\right)^2.
\end{align}
Upon substituting for $f(\boldsymbol{\kappa})$ in~\eqref{eq:sol_kappa}, we get~\eqref{eq:update_kappa_update}, 
from which we obtain  $\boldsymbol{\kappa}_{n}^{\mathsf{opt.}}$ as given in~\eqref{eq:update_kappa}.
\begin{figure*}
\begin{align}\label{eq:update_kappa_update}
	&\arg\min_{\boldsymbol{\kappa}}\Big(\sum_{n\in\mathcal{U}_{\mathsf{d}}}\boldsymbol{\kappa}_{n}^{T}\Big(\sum_{q\in\mathcal{U}_{\mathsf{d}}}\tilde{\varpi}_{\mathsf{d}, n}^2\mathbf{I}_{nq}+\frac{\delta_{\mathsf{p}}}{2}\boldsymbol{I}_{\lvert\mathcal{A}_{\mathsf{d}}\rvert}\Big)\Big)\boldsymbol{\kappa}_{n}-2\sum_{n\in\mathcal{U}_{\mathsf{d}}}(\tilde{\varpi}_{\mathsf{d}, n}\sqrt{(1+\varpi_{\mathsf{d},n})}\mathbf{g}_{\mathsf{d},n}+\frac{\delta_{\mathsf{p}}}{2}(\boldsymbol{\pi}_{:,n}+\boldsymbol{\bar{\pi}}_{:, n}))^{T}\boldsymbol{\kappa_{n}}.
\end{align}
\end{figure*}
\begin{figure*}
\begin{align}\label{eq:update_kappa}	
\boldsymbol{\kappa}_{n}^{\mathsf{opt.}}=(\sum\nolimits_{q\in\mathcal{U}_{\mathsf{d}}}\tilde{\varpi}_{\mathsf{d}, n}^2\mathbf{I}_{nq}+\frac{\delta_{\mathsf{p}}}{2}\boldsymbol{I}_{\lvert\mathcal{A}_{\mathsf{d}}\rvert})^{-1}\left(\tilde{\varpi}_{\mathsf{d}, n}\sqrt{(1+\varpi_{\mathsf{d},n})}\mathbf{g}_{\mathsf{d},n}+\frac{\delta_{\mathsf{p}}}{2}\left(\boldsymbol{\pi}_{:,n}+\boldsymbol{\bar{\pi}}_{:, n}\right)\right)\in\mathbb{C}^{\lvert\mathcal{A}_{\mathsf{d}}\rvert\times 1}, \forall n\in\mathcal{U}_{\mathsf{d}}.
 \end{align}
\end{figure*}
Here, $\boldsymbol{\pi}_{:,n}\triangleq\left[\pi_{1,n},\pi_{2,n},\ldots,\pi_{\lvert\mathcal{A}_{\mathsf{d}\rvert,n}}\right]^{T}$ and $\boldsymbol{\bar{\pi}}_{:,n}\triangleq\left[\bar{\pi}_{1,n},\bar{\pi}_{2,n},\ldots,\bar{\pi}_{\lvert\mathcal{A}_{\mathsf{d}\rvert,n}}\right]^{T}$. 
Next, we can find the optimal update of the second block of primal variables as $$\arg\min_{\boldsymbol{\pi}}\mathds{1}_{\boldsymbol{\kappa}}(\boldsymbol{\pi}) +\frac{\delta_{\mathsf{p}}}{2}\sum\nolimits_{n\in\mathcal{U}_{\mathsf{d}}}\sum\nolimits_{j\in\mathcal{A}_{\mathsf{d}}}\left(\pi_{jn}-\kappa_{jn}+\bar{\pi}_{jn}\right)^2.$$ 
Equivalently, $\pi_{jn}^{\mathsf{opt.}}$ is the solution of
\begin{subequations}
\begin{align}
	\arg\min_{\pi_{jn}}\quad\quad\quad&\hspace*{-7mm} \frac{\delta_{\mathsf{p}}}{2}\sum\nolimits_{n\in\mathcal{U}_{\mathsf{d}}}\sum\nolimits_{j\in\mathcal{A}_{\mathsf{d}}}\left(\pi_{jn}-\kappa_{jn}+\bar{\pi}_{jn}\right)^2,\label{eq:ADMM_second_primal}\\\mathrm{subject~to}\quad\quad&\sum\nolimits_{n\in\mathcal{U}_{\mathsf{d}}}\pi_{jn}^{2}\leq 1, \forall j\in\mathcal{A}_{\mathsf{d}}.
\end{align}
\end{subequations}
Now,~\eqref{eq:ADMM_second_primal} can be solved independently for each AP index by evaluating $\arg\min_{\pi_{jn}}\quad \frac{\delta_{\mathsf{p}}}{2}\sum\nolimits_{n\in\mathcal{U}_{\mathsf{d}}}\left(\pi_{jn}-\kappa_{jn}+\bar{\pi}_{jn}\right)^2$ for each $j\in\mathcal{A}_{\mathsf{d}}$ subject to the per-AP power constraint $\sum_{n\in\mathcal{U}_{\mathsf{d}}}\pi_{jn}^{2}\leq 1$.
Using the KKT condition, we can show that the optimal solution corresponding to the $j$th DL AP, $\boldsymbol{\pi}_{j}^{\mathsf{opt.}}$, is 
\begin{align}\label{eq:update_pi}
	\boldsymbol{\pi}_{j}^{\mathsf{opt.}}=\min\left\{1,\sqrt{\frac{1}{\|\boldsymbol{\dot{\pi}}_{j}\|^2}}\right\}\boldsymbol{\dot{\pi}}_{j}, \forall j\in\mathcal{A}_{\mathsf{d}},
\end{align}
with $\boldsymbol{\dot{\pi}}_{j}\triangleq\left[(\kappa_{j1}-\bar{\pi}_{j1}),(\kappa_{j2}-\bar{\pi}_{j2}),\ldots,(\kappa_{j\lvert\mathcal{U}_{\mathsf{d}}\rvert}-\bar{\pi}_{j\lvert\mathcal{U}_{\mathsf{d}}\rvert})\right]^{T}$. We summarize the iterative recipe for solving our original problem~\eqref{eq:QCQP} via ADMM approach in Algorithm~\ref{algo:ADMMM}. We note that the $\{\kappa_{jn}\}$  yielded by Algorithm~\ref{algo:ADMMM} are \emph{globally optimal}. We present the overall DL power control recipe in Algorithm~\ref{algo:DL_power_control}. The stopping criterion of the algorithm is decided by the threshold $\delta_{\mathsf{d}}$.  
Also, with ZF precoding, the closed-form update equation for $\tilde{\varpi}^{\tt iter+1}_{n}$ is presented in~\eqref{eq:update_mu}. To update $\boldsymbol{\kappa}$, $\mathbf{I}_{nq}$ needs to be substituted in~\eqref{eq:update_kappa} from~\eqref{eq:I_nq}.
\begin{figure*}
	\begin{align}\label{eq:update_mu}
	\tilde{\varpi}^{\tt iter+1}_{\mathsf{d}, n}=\frac{\sqrt{1+\varpi_{\mathsf{d},n}^{{\tt iter}+1}}\sqrt{(N-\tau_{p})}\sqrt{\mathcal{E}_{\mathsf{d}}}\left(\sum\nolimits_{j\in\mathcal{A}_{\mathsf{d}}}\alpha_{jn}\kappa_{jn}^{\tt iter}\right)}{\sum\limits_{q\in\mathcal{U}_{\mathsf{d}}}\sum\limits_{j\in\mathcal{A}_{\mathsf{d}}}\mathcal{E}_{\mathsf{d}}(\kappa_{jq}^{\tt iter})^2(\beta_{jn}-\alpha_{jn}^2)+(N-\tau_{p})\sum\limits_{ q\in\mathcal{U}_{\mathsf{d}}}\mathcal{E}_{\mathsf{d}}\left(\sum\nolimits_{j\in\mathcal{A}_{\mathsf{d}}}\kappa_{jq}^{\tt iter}\alpha_{jn}\right)^2\lvert\langle\boldsymbol{\varphi}_{l(n)},\boldsymbol{\varphi}_{l(q)}\rangle\rvert^2+\sigma_{\mathsf{eff.}, \mathsf{d}, n}^2}.
	\end{align}
\end{figure*}
\begin{algorithm}[!t]
	\caption{Solving ~\eqref{eq:QCQP} via ADMM}\label{algo:ADMMM}
	\SetNlSty{textbf}{}{:}
		\DontPrintSemicolon
	\KwIn{$\pi_{jn}^{0}, \forall j\in\mathcal{A}_{\mathsf{d}}, n\in\mathcal{U}_{\mathsf{d}}$, $\delta_{\mathsf{ADMM}}>0$}
	\KwInt{$\bar{\pi}_{jn}^{0}=0, \forall j\in\mathcal{A}_{\mathsf{d}}, n\in\mathcal{U}_{\mathsf{d}}$, ${\tt iter}=0$}
	\While{$\|\boldsymbol{\kappa_{n}}-\boldsymbol{\pi}_{:,n}\|\geq \delta_{\mathsf{ADMM}}$}{
	Evaluate: $\boldsymbol{\kappa}_{n}^{{\tt iter}+1}$ using~\eqref{eq:update_kappa} with $\boldsymbol{\pi}_{:,n}^{\tt iter}$ and $\boldsymbol{\bar{\pi}}_{:, n}^{\tt iter}$\;\label{algo:DL_kappa}
	Evaluate: $\boldsymbol{\pi}_{:,n}^{{\tt iter}+1}$ using~\eqref{eq:update_pi} with $\boldsymbol{\dot{\pi}}_{j}\triangleq\left[(\kappa_{j1}^{{\tt iter}+1}-\bar{\pi}_{j1}^{{\tt iter}}),\ldots,(\kappa_{j\lvert\mathcal{U}_{\mathsf{d}}\rvert}^{{\tt iter}+1}-\bar{\pi}_{j\lvert\mathcal{U}_{\mathsf{d}}\rvert}^{{\tt iter}})\right]^{T}$\;\label{algo:DL_pi}
	Update: $\bar{\pi}_{jn}^{{\tt iter}+1}=\pi_{jn}^{\tt iter}-\kappa_{jn}^{{\tt iter}+1}+\bar{\pi}_{jn}^{\tt iter}$\;\label{algo:DL_bar_pi}
	Update: ${\tt iter}={\tt iter}+1$
}
\end{algorithm}

\begin{algorithm}[!t]
	\caption{Downlink Power Control}\label{algo:DL_power_control}
	\SetNlSty{textbf}{}{:}
		\DontPrintSemicolon
	\KwIn{$\mathcal{E}_{\mathsf{u},k}, \forall k\in\mathcal{U}_{\mathsf{u}}$}
	\KwInt{$\kappa_{jn}^{0}, \forall j\in\mathcal{A}_{\mathsf{d}}, n\in\mathcal{U}_{\mathsf{d}}$, ${\tt iter}=0$}
	\While{$\lvert f(\boldsymbol{\kappa},\boldsymbol{\varpi}_{\mathsf{d}})^{{\tt iter}+1}-f(\boldsymbol{\kappa},\boldsymbol{\varpi}_{\mathsf{d}})^{\tt iter}\rvert\geq \delta_{\mathsf{d}}$}{
Evaluate: $\varpi_{\mathsf{d},n}^{\tt iter+1}=\dfrac{\left((\boldsymbol{\kappa}_{q}^{\tt iter})^T\mathbf{g}_{\mathsf{d},n}\right)^2}{\sum\limits_{q\in\mathcal{U}_{\mathsf{d}}}(\boldsymbol{\kappa}_{q}^{\tt iter})^T \mathbf{I}_{nq}\boldsymbol{\kappa}_{q}^{\tt iter}-\left((\boldsymbol{\kappa}_{q}^{\tt iter})^T\mathbf{g}_{\mathsf{d},n}\right)^2+\sigma_{\text{eff.}}^2}$\;\label{algo:DL_varpi_d}
Evaluate: $\tilde{\varpi}^{\tt iter+1}_{n}=\frac{\sqrt{1+\varpi_{\mathsf{d},n}^{\tt iter+1}}\mathbf{g}_{\mathsf{d},n}^{T}\boldsymbol{\kappa}_{n}^{\tt iter}}{\sum\limits_{q\in\mathcal{U}_{\mathsf{d}}}(\boldsymbol{\kappa}_{q}^{\tt iter})^T \mathbf{I}_{nq}\boldsymbol{\kappa}_{q}^{\tt iter}+\sigma_{\text{eff.}}^2}$\;\label{algo:DL_varpi_d_tilde}
Evaluate: $\boldsymbol{\kappa}_{q}^{{\tt iter}+1}$ via solving~\eqref{eq:QCQP} with $\varpi_{\mathsf{d},n}=\varpi_{\mathsf{d},n}^{\tt iter+1}$ and $\tilde{\varpi}_{\mathsf{d}, n}=\tilde{\varpi}^{\tt iter+1}_{\mathsf{d}, n}$\;
Update: ${\tt iter}={\tt iter}+1$\;
}
\end{algorithm}
\begin{prop}\label{prop:DL_power_control}
	The DL power allocation algorithm is convergent in objective since $f(\boldsymbol{\kappa},\boldsymbol{\varpi}_{\mathsf{d}})$ is bounded above and monotonically non-decreasing after each iteration.
\end{prop}

We now highlight the key benefits of our DL power control scheme. First of all, the DL power allocation algorithm is precoder and combiner agnostic~(lines~\ref{algo:DL_varpi_d}, \ref{algo:DL_varpi_d_tilde} in Algorithm~\ref{algo:DL_power_control} and lines~\ref{algo:DL_kappa},~\ref{algo:DL_pi}, and~\ref{algo:DL_bar_pi} in Algorithm~\ref{algo:ADMMM} apply to any choice of precoder), similar to our UL power allocation algorithm. This makes our algorithms more widely applicable compared to~\cite{SPAWC_Virtual_Duplex}. Also, as mentioned for UL, most of the works in literature lower bound the original cost function (e.g., the sum SE optimization in~\cite{SPAWC_Virtual_Duplex},~\cite{Joint_UE_NAFD}) by a series of surrogate convex functions and use available general-purpose convex solvers. In contrast, we provide closed-form update equations for all the auxiliary variables, and thanks to FP, our algorithm directly optimizes the original cost function. 
\begin{rem}
One can generalize~\eqref{eq:problem_AP_powercontrol}, for example, to include user priority or fairness guarantees, by considering weighted sum UL-DL SE. The technical development, algorithms, and update equations directly extend to weighted sum UL-DL SE maximization, thanks to the FP-based approach~\cite{FP_I}, at the cost of additional notational bookkeeping. Hence, for simplicity, we do not include it in this paper.
\end{rem}

	 \section{FD CF  MIMO System}\label{sec:FD_CF}
	We now present the SE analysis and power control for an FD-enabled CF MIMO system.
	The analysis for the FD system is similar to that of the DTDD system presented in Sec.~\ref{sec:DTDD_SE}, except that all the APs are now capable of concurrent transmission and reception.  Hence, $\mathcal{A}_{\mathsf{u}}=\mathcal{A}_{\mathsf{d}}=\mathcal{A}$ with $\lvert\mathcal{A}\rvert=M$. Also, each AP suffers from IrAI. We can express the $k$th stream of the UL received signal~(after local combining using $\mathbf{v}_{mk}\in\mathbb{C}^{N_{\tt rx}}$) at the $m$th FD AP as
		\begin{align}\label{eq:s_uk_hat_FD}
			&\hat{s}_{\mathsf{u},mk}=\sqrt{\mathcal{E}_{\mathsf{u},k}}\mathbf{v}_{mk}^{H}\mathbf{f}_{mk}s_{\mathsf{u},k}+\sum\nolimits_{n\in\mathcal{U}_{\mathsf{u}}\backslash k}\sqrt{\mathcal{E}_{\mathsf{u},n}}\mathbf{v}_{mk}^{H}\mathbf{f}_{mn}s_{\mathsf{u},n}\notag\\&+\underbrace{\sqrt{\mathcal{E}_{\mathsf{d}}}\sum\nolimits_{\substack{j\neq m,\\ j= 1}}^{M}\sum\nolimits_{n\in\mathcal{U}_{\mathsf{d}}}\kappa_{jn}\mathbf{v}_{mk}^{H} \tilde{\mathbf{G}}_{mj}\mathbf{p}_{jn}s_{\mathsf{d},n}}_{\text{InAI from all APs except the $m$th AP}}\notag\\&+\underbrace{\sqrt{\mathcal{E}_{\mathsf{d}}}\sum\nolimits_{n\in\mathcal{U}_{\mathsf{d}}}\kappa_{mn}\mathbf{v}_{mk}^{H} \mathbf{G}_{m}^{\mathsf{SI}}\mathbf{p}_{mn}s_{\mathsf{d},n}}_{\text{IrAI of $m$th AP}}+\mathbf{v}_{mk}^{H}\mathbf{w}_{\mathsf{u},mk}.
		\end{align}
		In~\eqref{eq:s_uk_hat_FD}, the third and fourth terms correspond to the InAI and IrAI, as indicated. We see that DL signals from all the APs interfere with the UL signal received at any AP, unlike DTDD, where only the APs scheduled in DL interfere with signals received at the APs scheduled in UL. Then, the weighted received signal from the $k$th UL UE at the CPU becomes
		\begin{multline}
			\hat{s}_{\mathsf{u},k}=\sum\nolimits_{m=1}^{M}\!\!\omega_{mk}^{*}\hat{s}_{\mathsf{u},mk}=\sqrt{\mathcal{E}_{\mathsf{u},k}}\sum\nolimits_{m=1}^{M}\!\!\omega_{mk}^{*}\mathbf{v}_{mk}^{H}\mathbf{f}_{mk}s_{\mathsf{u},k}\\+\sum\nolimits_{n\in\mathcal{U}_{\mathsf{u}}\backslash k}\sqrt{\mathcal{E}_{\mathsf{u},n}}\sum\nolimits_{m=1}^{M}\omega_{mk}^{*}\mathbf{v}_{mk}^{H}\mathbf{f}_{mn}s_{\mathsf{u},n}\\+\sqrt{\mathcal{E}_{\mathsf{d}}}\sum\nolimits_{m=1}^{M}\omega_{mk}^{*}\sum\nolimits_{\substack{j\neq m,\\ j= 1}}^{M}\sum\nolimits_{n\in\mathcal{U}_{\mathsf{d}}}\kappa_{jn}\mathbf{v}_{mk}^{H} \tilde{\mathbf{G}}_{mj}\mathbf{p}_{jn}s_{\mathsf{d},n}\\+\sqrt{\mathcal{E}_{\mathsf{d}}}\sum\nolimits_{m=1}^{M}\omega_{mk}^{*}\sum\nolimits_{n\in\mathcal{U}_{\mathsf{d}}}\kappa_{mn}\mathbf{v}_{mk}^{H} \mathbf{G}_{m}^{\mathsf{SI}}\mathbf{p}_{mn}s_{\mathsf{d},n}\\+\sum\nolimits_{m=1}^{M}\omega_{mk}^{*}\mathbf{v}_{mk}^{H}\mathbf{w}_{\mathsf{u},mk}.
		\end{multline}
		Hence, we can write the UL SINR of the $k$th UE for the FD system as follows:
		\begin{align*}
			\frac{\mathcal{E}_{\mathsf{u},k}\lvert\boldsymbol{\omega_{k}}^{H}\mathbb{E}[\mathbf{u}_{kk}]\rvert^{2}}{\boldsymbol{\omega_{k}}^{H}				
					\begin{pmatrix}
					\sum\nolimits_{i\in\mathcal{U}_{\mathsf{u}}}\mathcal{E}_{\mathsf{u},i}\mathbb{E}[\mathbf{u}_{ki}\mathbf{u}_{ki}^{H}]-\mathcal{E}_{\mathsf{u},k}\mathbb{E}\left[\mathbf{u}_{kk}\right]\mathbb{E}[\mathbf{u}_{kk}^{H}]&\\ +\sum\nolimits_{n\in\mathcal{U}_{\mathsf{d}}}\mathbb{E}[\mathbf{a}_{kn}\mathbf{a}_{kn}^{H}]+\mathbf{N}_{\mathsf{eff.}}
				\end{pmatrix}
				\boldsymbol{\omega_{k}}},
		\end{align*}
		where $\boldsymbol{\omega_{k}}\triangleq[\omega_{1k},\omega_{2k},\ldots,\omega_{Mk}]^T\in\mathbb{C}^{M}$, $\mathbf{u}_{ki}\triangleq[\mathbf{v}_{1k}^{H}\mathbf{f}_{1i},\mathbf{v}_{2k}^{H}\mathbf{f}_{2i},\ldots,\mathbf{v}_{Mk}^{H}\mathbf{f}_{Mi}]^{T}\in\mathbb{C}^{M}$, $[\mathbf{a}_{kn}]_{m}=\sum\nolimits_{\substack{j=1,\\ j\neq m}}^{M}\sqrt{\mathcal{E}_{\mathsf{d}}}\kappa_{jn}\mathbf{v}_{mk}^{H} \tilde{\mathbf{G}}_{mj}\mathbf{p}_{jn}+\sqrt{\mathcal{E}_{\mathsf{d}}}\kappa_{mn}\mathbf{v}_{mk}^{H} {\mathbf{G}}_{m}^{\mathsf{SI}}\mathbf{p}_{mn}$, and $\mathbf{N}_{\mathsf{eff.}}=N_{0}{\tt diag}( \mathbb{E}[\|\mathbf{v}_{1k}\|^2], \ldots,\mathbb{E}[\|\mathbf{v}_{Mk}\|^2])\in\mathbb{C}^{M\times M}$. We note that the FD-SINR expression is similar to the DTDD case~(see~(\eqref{eq:UL_SINR}) except that we have set $\mathcal{A}_{\mathsf{u}}=M$ since all the APs are capable of UL reception, and the addition of IrAI in  $[\mathbf{a}_{kn}]_{m}$.
Hence, to avoid repetition, we present the final results related to the optimal CPU combining weights, the UL and DL SE expressions, and the power control algorithms for the FD system without detailed proofs.\footnote{Our SE expressions match with those in the FD literature under special cases such as perfect channel estimation~\cite{FD_CF_ICC} and equal weight-based combining at the CPU~\cite{FD_CF_ICC, SPAWC_Virtual_Duplex}.}
 \begin{cor}\label{corr:FD_weights}
 	The optimal weighting vector  for the $k$th UL UE at the CPU is $\boldsymbol{\alpha}_{k}^{\mathsf{opt.}}=\mathcal{E}_{\mathsf{u},k}\mathbf{R}_{\alpha}^{-1}\bar{\mathbf{u}}_{k}$, with $\bar{\mathbf{u}}_{k}$$=[\alpha_{\mathsf{u}, 1k}^2,\alpha_{\mathsf{u}, 2k}^2,\ldots,\alpha_{\mathsf{u}, Mk}^2]^{T}$ and $\mathbf{R}_{\alpha}=\sum\limits_{k'\in\mathcal{P}_{l(k)\backslash k}}\mathcal{E}_{\mathsf{u},k'} \bar{\mathbf{u}}_{k'} \bar{\mathbf{u}}_{k'}^{H}+\frac{1}{N_{\mathsf{rx}}-\tau_p}\mathbf{\dot{R}}_{\alpha}+\frac{\mathcal{E}_{\mathsf{d}}}{N_{\mathsf{rx}}-\tau_{p}}\mathbf{\ddot{R}}_{\alpha}$, where $\mathbf{\dot{R}}_{\alpha}$ and $\mathbf{\ddot{R}}_{\alpha}$ are  diagonal matrices with $m$th diagonal entry being
 	\begin{subequations}
 		\begin{align*}
 			&[\mathbf{\dot{R}}_{\alpha}]_{m}=\sum\nolimits_{k'\in\mathcal{U}_{\mathsf{u}}}\hspace{-1mm}\mathcal{E}_{\mathsf{u},k'}\alpha_{\mathsf{u}, mk}^2(\beta_{\mathsf{u}, mk'}-\alpha_{\mathsf{u}, mk'}^2)+N_{0}\alpha_{\mathsf{u}, mk'}^2,\hspace{-6mm}\\&[\mathbf{\ddot{R}}_{\alpha}]_{m}=\sum\nolimits_{j=1, j\neq m}^{M}\hspace{-2.5mm} N_{\mathsf{tx}}\kappa_{jn}^2\zeta_{mj}^{\mathsf{InAP}}\alpha_{\mathsf{u}, mk}^2+N_{\mathsf{tx}}\kappa_{mn}^2\alpha_{\mathsf{u}, mk}^2\zeta_{mm}^{\mathsf{SI}},
 		\end{align*}
 	\end{subequations}
 respectively, for $m=1,2,\ldots, M$.
 \end{cor}
 \begin{cor}\label{corr:FD_SE}
 	The UL and DL SINRs of the FD CF MIMO system are, respectively,
 	\begin{align}\label{eq:UL_SINR_CF_FD}
 		\eta_{\mathsf{u},k}(\boldsymbol{\kappa},\boldsymbol{\mathcal{E}}_{\mathsf{u}})=\frac{(N_{\mathsf{rx}}-\tau_{p})\mathcal{E}_{\mathsf{u},k}(\sum\nolimits_{m=1}^{M}\omega_{mk}^{*}\alpha_{\mathsf{u}, mk}^{2})^2}{\begin{pmatrix}
 				{\tt EST}_{\mathsf{u},k}+{\tt MUI}_{\mathsf{u},k}\\+{\tt IAP}_{k}+N_{0}\sum\nolimits_{m=1}^{M}|\omega_{mk}^{*}|^{2}\alpha_{\mathsf{u}, mk}^{2}
 			\end{pmatrix}},
 	\end{align}
 and 
  \begin{align}
 	\eta_{\mathsf{d},n}(\boldsymbol{\kappa},\boldsymbol{\mathcal{E}}_{\mathsf{u}})=\frac{(N_{\mathsf{tx}}-\tau_{p})(\sum\nolimits_{j=1}^{M}\sqrt{\mathcal{E}_{\mathsf{d}}}\alpha_{\mathsf{d}, jn}\kappa_{jn})^2}{{\tt EST}_{\mathsf{d},n}+{\tt MUI}_{\mathsf{d},n}+{\tt IUE}_{n}+N_{0}},
 \end{align}
 	where ${\tt EST}_{\mathsf{u},k}$, ${\tt MUI}_{\mathsf{u},k}$ and ${\tt IAP}_{k}$ correspond to interferences caused by the channel estimation error, multi-UE interference, and inter-AP interference. For DL, ${\tt EST}_{\mathsf{d},n}$,  ${\tt MUI}^{d}_{n}$ and ${\tt IUE}_{n}$ represent error due to channel estimation, the DL multi-UE interference and the UL UE to DL UE CLI, respectively. These terms can be evaluated as
 	\begin{subequations}
 		\begin{align*}
 			&{\tt EST}_{\mathsf{u},k}= \sum\nolimits_{k'\in\mathcal{U}_{\mathsf{u}}}\hspace*{-1mm}\mathcal{E}_{\mathsf{u},k'}\hspace*{-2mm}\sum\nolimits_{m=1}^{M}|\omega_{mk}^*|^2\alpha_{\mathsf{u}, mk}^2(\beta_{\mathsf{u}, mk'}-\alpha_{\mathsf{u}, mk'}^2),\hspace*{-4mm}\\
 			&{\tt MUI}_{\mathsf{u},k}= (N_{\mathsf{rx}}-\tau_{p})\sum\nolimits_{i\in\mathcal{P}_{l(k)}\backslash k}\mathcal{E}_{\mathsf{u},i}\Big(\sum\nolimits_{m=1}^{M}\omega_{mk}^{*}\alpha_{\mathsf{u}, mi}^{2}\Big)^2,
 			\\
 			&{\tt IAP}_{k}=N_{\mathsf{tx}}\mathcal{E}_{\mathsf{d}}\sum\limits_{n\in\mathcal{U}_{\mathsf{d}}}\sum\limits_{m=1}^{M}\Big(\sum\limits_{\substack{j=1, j\neq m}}^{M}\kappa_{jn}^2\zeta_{mj}^{\mathsf{InAP}}|\omega_{mk}^{*}|^{2}\alpha_{\mathsf{u}, mk}^2\notag\\&\hspace*{1.7cm}+\kappa_{mn}^2|\omega_{mk}^{*}|^{2}\alpha_{\mathsf{u}, mk}^2\zeta_{mm}^{\mathsf{SI}}\Big),\\
 			&{\tt EST}_{\mathsf{d},n}= \sum\nolimits_{q\in\mathcal{U}_{\mathsf{d}}}\sum\nolimits_{j=1}^{M}\mathcal{E}_{\mathsf{d}}\kappa_{jq}^2(\beta_{\mathsf{d}, jn}-\alpha_{\mathsf{d}, jn}^2),\\
 				&{\tt MUI}^{d}_{n}=(N_{\mathsf{tx}}-\tau_{p})\mathcal{E}_{\mathsf{d}}\sum\nolimits_{ q\in\mathcal{P}_{l(n)}\backslash n}\Big(\sum\nolimits_{j=1}^{M}\kappa_{jq}\alpha_{\mathsf{d}, jn}\Big)^2,\label{eq:dl_ncoh_FD}
 			\\&{\tt IUE}_{n}=\sum\nolimits_{{k\in\mathcal{U}_{\mathsf{u}}}}\mathcal{E}_{\mathsf{u},k}\epsilon_{nk}.
 			\end{align*}
 	\end{subequations}
 Thus, the sum UL-DL SE of the FD enabled CF-system can be expressed as $	\mathcal{R}_{\mathsf{s}}(\boldsymbol{\kappa},\boldsymbol{\mathcal{E}}_{\mathsf{u}})=\frac{\tau-\tau_{p}}{\tau}[\sum\nolimits_{k\in\mathcal{U}_{\mathsf{u}}}\log(1+\eta_{\mathsf{u},k}(\boldsymbol{\kappa},\boldsymbol{\mathcal{E}}_{\mathsf{u}}))+\sum\nolimits_{n\in\mathcal{U}_{\mathsf{d}}}\log(1+\eta_{\mathsf{d},n}(\boldsymbol{\kappa},\boldsymbol{\mathcal{E}}_{\mathsf{u}}))]$.
 \end{cor}
 \begin{proof}
 The proof follows using similar techniques as in Lemmas~\ref{lemm:UL_SINR_DTDD}
and~\ref{lem:DL_SE_DTDD}. We only explain the proof for IrAI. The power of the IrAI with ZF is
		$\mathbb{E}\left[\left|\sqrt{\mathcal{E}_{\mathsf{d}}}\kappa_{mn}\mathbf{v}_{mk}^{H} {\mathbf{G}}_{m}^{\mathsf{SI}}\mathbf{p}_{mn}\right|^{2}\right]=\mathcal{E}_{\mathsf{d}}\kappa_{mn}^2\mathbb{E}\left[{\tt tr}\left\{\mathbf{v}_{mk}^{H} \mathbf{G}_{m}^{\mathsf{SI}}\mathbf{p}_{mn}\mathbf{p}_{mn}^{H}{\mathbf{G}_{m}^{\mathsf{SI}}}^{H}\mathbf{v}_{mk}\right\}\right]=\mathcal{E}_{\mathsf{d}}\kappa_{mn}^2{\tt tr}\left\{\mathbb{E}\left[\mathbf{v}_{mk}^{H}\mathbb{E}\left[ \mathbf{G}_{m}^{\mathsf{SI}}\mathbf{p}_{mn}\mathbf{p}_{mn}^{H}{\mathbf{G}_{m}^{\mathsf{SI}}}^{H}\right]\mathbf{v}_{mk}\right]\right\}=N_{\tt tx}\mathcal{E}_{\mathsf{d}}\kappa_{mn}^2\zeta_{mm}^{\mathsf{SI}}{\tt tr}\left\{\mathbb{E}\left[\|\mathbf{v}_{mk}\|^{2}\right]\right\}=\frac{N_{\tt tx}\mathcal{E}_{\mathsf{d}}\kappa_{mn}^2\zeta_{mm}^{\mathsf{SI}}\alpha^2_{mk}}{N_{\tt rx}-\tau_{p}}$,
		where for the inner expectation, we note that the $i$th diagonal entry, $\mathbb{E}\left[ \mathbf{G}_{m}^{\mathsf{SI}}\mathbf{p}_{mn}\mathbf{p}_{mn}^{H}{\mathbf{G}_{m}^{\mathsf{SI}}}^{H}\right]_{ii}$, can be evaluated as 
		\begin{multline}
		\mathbb{E}\left[ \mathbf{G}_{m}^{\mathsf{SI}}\left[i,:\right]\mathbf{p}_{mn}\mathbf{p}_{mn}^{H}{\mathbf{G}_{m}^{\mathsf{SI}}}^{H}\left[i,:\right]\right]\\={\tt tr}\left\{\mathbb{E}\left[\mathbf{p}_{mn}\mathbf{p}_{mn}^{H}\right]\mathbb{E}\left[{\mathbf{G}_{m}^{\mathsf{SI}}}^{H}\left[i,:\right]\mathbf{G}_{m}^{\mathsf{SI}}\left[i,:\right]\right]\right\}=N_{\tt tx}\zeta_{mm}^{\mathsf{SI}}.\notag
		\end{multline}
		Here, $\mathbf{G}_{m}^{\mathsf{SI}}\left[i,:\right]$ denotes the $i$th row of $\mathbf{G}_{m}^{\mathsf{SI}}$.
		 It is easy to show that the off-diagonal terms evaluate to zero. Thus, $\mathbb{E}\left[ \mathbf{G}_{m}^{\mathsf{SI}}\mathbf{p}_{mn}\mathbf{p}_{mn}^{H}{\mathbf{G}_{m}^{\mathsf{SI}}}^{H}\right] = N_{\tt tx}\zeta_{mm}^{\mathsf{SI}}\mathbf{I}_{N_{\tt rx}}$. Finally, with ZF combining, $\mathbb{E}\left[\|\mathbf{v}_{mk}\|^{2}\right]=\frac{\alpha^2_{mk}}{N_{\tt rx}-\tau_{p}}$.
 \end{proof}
From the expression for ${\tt IAP}_{k}$, we see that the DL signals from \emph{all} the APs interfere with the UL signals of any AP, unlike DTDD. On the other hand,  there is no need for scheduling APs in UL/DL in the FD system.
Thus, we only need to consider power allocation for the FD system. However, it is easy to apply the UL and DL power allocation protocols developed for the DTDD-enabled system in the FD case. As mentioned earlier, instead of the scheduled UL and DL AP subsets, we now have $\mathcal{A}_{\mathsf{u}}=\mathcal{A}_{\mathsf{d}}=\{1,2,\ldots, M\}$. This only changes the limits of the summations in the beamforming gain and the interference terms. Secondly, due to IrAI, for the sub-problem of UL power allocation, the effective noise additionally includes the power of IrAI. To elaborate, in the FD case, we have $\mathsf{G}_{\mathsf{u},k}(\boldsymbol{\mathcal{E}}_{\mathsf{u}})=(N-\tau_{p})\mathcal{E}_{\mathsf{u},k}\left(\sum\nolimits_{m=1}^{M}\omega_{mk}^{*}\alpha_{mk}^{2}\right)^2,$ 
$		\mathsf{I}_{\mathsf{u},k}(\boldsymbol{\mathcal{E}}_{\mathsf{u}})={\tt EST}_{\mathsf{u},k}+{\tt MUI}_{\mathsf{u},k}+\sigma_{\mathsf{eff.},\mathsf{u},k}^2$, and $
		\sigma_{\mathsf{eff.},\mathsf{u},k}^2\triangleq\left({\tt IAP}_{k}+N_{0}\sum\nolimits_{m=1}^{M}|\omega_{mk}^{*}|^{2}\alpha_{mk}^{2}\right).$
Recall that $\sigma_{\mathsf{eff.},\mathsf{u},k}^2$ is the power of the effective noise, which does not depend on the UL transmit powers. Here, ${\tt IAP}_{k}$ also includes the IrAI power, unlike DTDD. Thus, the original structure of the problem, as described for DTDD in Section~\ref{sec:UL_power_control}, does not change. Similar arguments can also be made for the case of DL. Thus, the algorithms~(namely Algorithm~\ref{algo:UL_power_control}, Algorithm~\ref{algo:ADMMM}, and Algorithm~\ref{algo:DL_power_control}) derived for DTDD directly apply to the FD system. We omit the details for brevity.

\section{Numerical Results}\label{sec:numerical}
We consider that the UEs are dropped uniformly at random locations in a $1$~km$^2$ square area. We take $5,000$ random channel instantiations for Monte Carlo averaging. We consider $50\%$ of the UEs to have UL data demands. The APs are deployed on a uniform rectangular grid for better coverage. The large scale fading between the $m$th AP and the $k$th UE is modeled as $\beta_{mk}=10^{\frac{\mathsf{PL}_{mk}+\sigma _{\mathsf{sh.}}z_{mk}}{10}}$, where the path-loss $\mathsf{PL}_{mk}$ follows the three-slope model in~\cite{FD_CF_ICC}, $\sigma _{\mathsf{sh.}}=6$ dB, and $z_{mk}\sim\mathcal{N}(0,1)$. The system bandwidth and noise figure are taken as $20$ MHz and $9$ dB, respectively, which gives a noise variance of $-92$ dBm. The coherence interval consists of  $200$ channel uses~\cite{making_cellfree}.
The pilot SNR is taken as $20$ dB. We set the algorithm parameters $\delta_{\mathsf{u}}$, $\delta_{\mathsf{d}}$, $\delta_{\mathsf{ADMM}}$ and $\delta_{\mathsf{p}}$ to $0.001$.  Other parameters, such as the number of APs, UEs, UL/DL data transmit powers, etc., are mentioned in the plots. 

In Fig.~\ref{fig:UL_power_control}, we compare the cumulative distribution function~(CDF) of the achievable UL SE with different power allocation schemes. We observe that FP-based power control added with weighted combing at the CPU~(see $\mathsf{WC+FP}$) uniformly outperforms only FP-based power control~(see $\mathsf{FP}$ without $\mathsf{WC}$) and only weighted combing at the CPU~(see $\mathsf{WC}$). This underlines the need for weighted combing along with UL power control rather than applying each individually. Further, we compare the proposed $\mathsf{WC+FP}$ scheme with the estimated channel variance and large-scale fading-dependent power control scheme proposed by Nikbakht et al. in~\cite{Nikbakht, Nikbakht_2}, and we observe almost $4$-fold improvement in $90\%$-likely UL SE rendered by our algorithm. Also, we verify the correctness of our derived closed-form expression for weighted combing given in Lemma~\ref{lem:optimal_weights}~(see $\mathsf{WC:}$ Lemma $5$) with that of Lemma~\ref{lem:weighted_UL_SINR}~(see $\mathsf{WC:}$ Lemma $5$), and see that SE achieved by derived weights matches the theoretical SE.

\begin{figure}[!]
\centering
	\centering
	\includegraphics[width=0.8\linewidth]{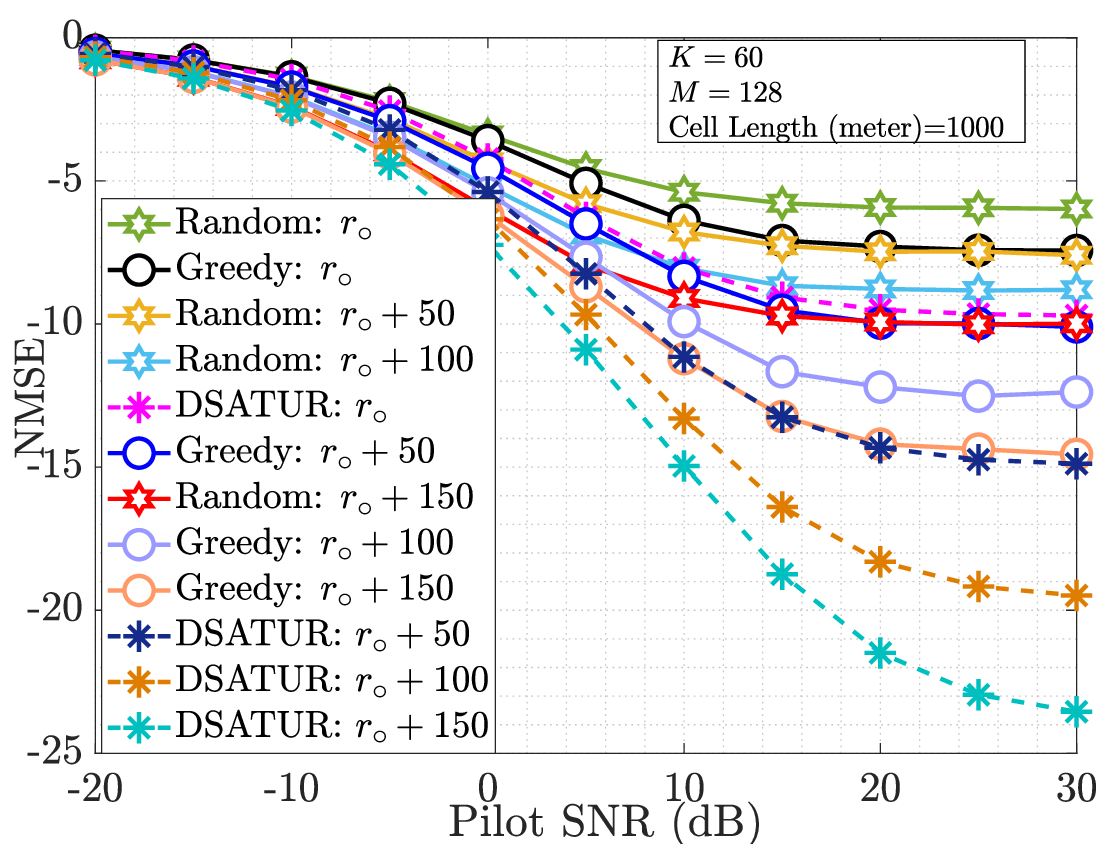}
	\caption{Comparison of Algorithm~\ref{algo:pilot_allocation} with existing literature in terms of NMSE.}\label{fig:SNR_p_vs_NMSE}
\end{figure}
\begin{figure}[!]
	\centering
	\includegraphics[width=0.8\linewidth]{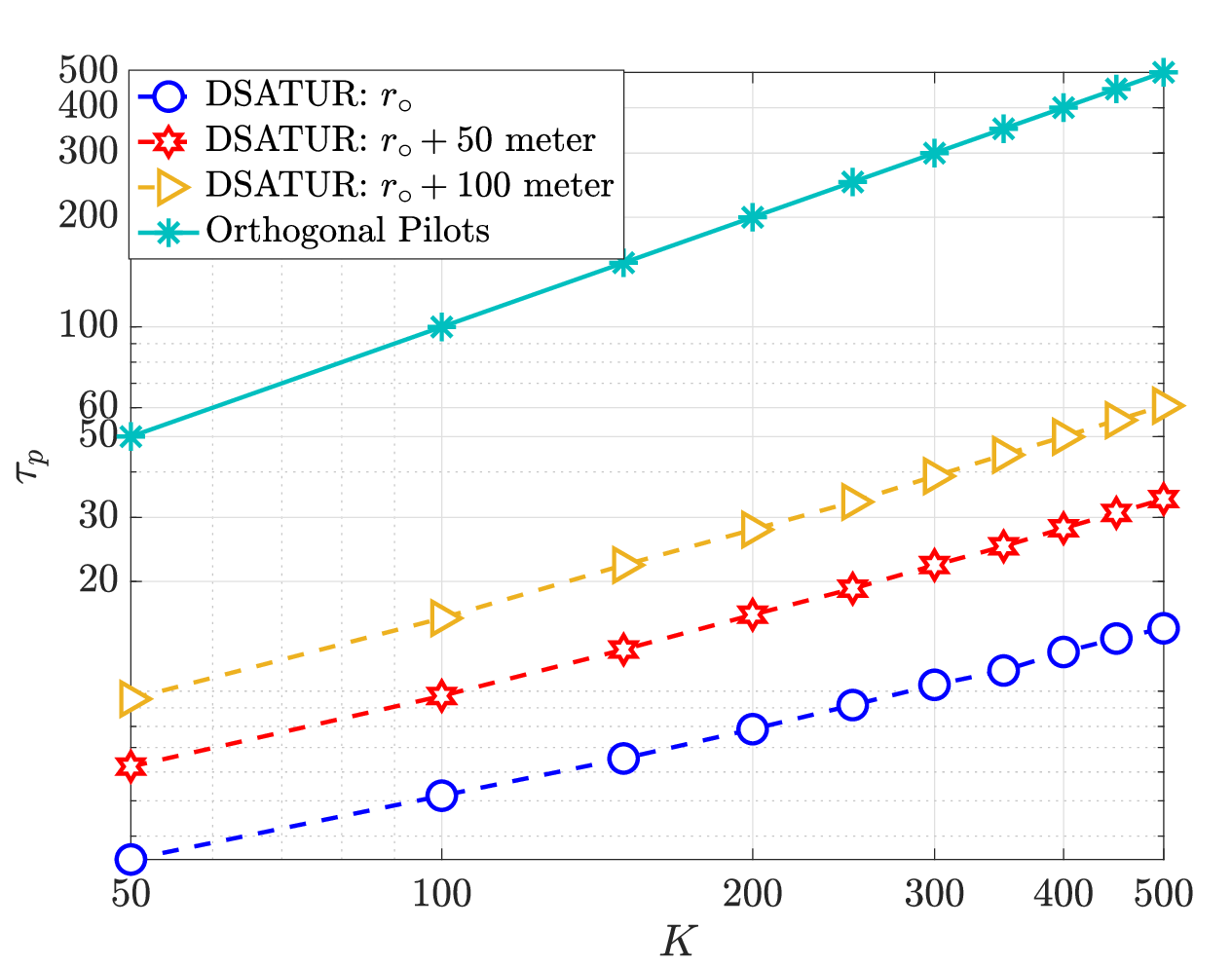}
	\caption{Scaling of pilot length with number of UEs~($K$).}\label{fig:fig_tau_p_vs_K}
\end{figure}

First, we numerically evaluate the efficacy of the proposed pilot allocation algorithm. The system model closest to our work is~\cite{DTDD_TCoM}, and thus, we compare the proposed algorithm with the greedy pilot allocation presented in ~\cite{DTDD_TCoM}. We also compare against a random pilot allocation scheme.\footnote{For both the greedy and random allocation schemes, we take the pilot length returned by our algorithm and assign that many pilots across the UEs.} In Fig.~\ref{fig:SNR_p_vs_NMSE}, we compare the performance of the pilot allocation scheme in terms of the NMSE in channel estimation at the APs as a function of pilot SNR. The pilot allocation returned by Algorithm~\ref{algo:pilot_allocation} leads to considerably lower NMSE than the greedy method. Also, as we increase the value of $r_{\mathsf{o}}$, the cardinality of $\mathcal{U}_{k}$ in~\eqref{eq:U_k} increases. This leads to a higher value of pilot length, and hence, the NMSE decreases even further. 
Fig.~\ref{fig:fig_tau_p_vs_K} illustrates the minimum pilot length required for given UE density and $r_{\mathsf{o}}$. Even with a large number of UEs, say $500$, the pilot length required is only \emph{one-tenth} that required for allocating fully orthonormal pilots. 

Next, to demonstrate the benefit of pilot length optimization, in Fig.~\ref{fig:fig_SE_vs_K}, we plot the sum UL-DL SE versus the number of UL and DL UEs. We see that, with orthogonal pilots, the pre-log factor $\frac{\tau-\tau_p}{\tau}$ can substantially degrade the SE as $\tau_p$ becomes comparable to $K$.  This trade-off is evident in Fig.~\ref{fig:fig_SE_vs_K}; when the UE load is much smaller than the coherence duration, the proposed scheme performs similarly to orthogonal pilot allocation. However, there is a dramatic improvement in the sum UL-DL SE by our proposed algorithm compared to fully orthogonal pilot allocation at a higher number of UEs. This underlines the utility of pilot reuse in a CF-MIMO system and the necessity of an algorithm that optimizes the pilot length and simultaneously reduces pilot contamination via systematic AP and UE clustering. Further, in Fig.~\ref{fig:fig_CDF_pilot_allocations}, we plot the CDFs achieved via different pilot allocations and observe that our proposed algorithm uniformly outperforms existing iterative and greedy methods~\cite{Heng_Liu_TVT, Lozano_Pilot, Location_Pilot}.

\begin{figure}
		\centering
		\includegraphics[width=0.85\linewidth]{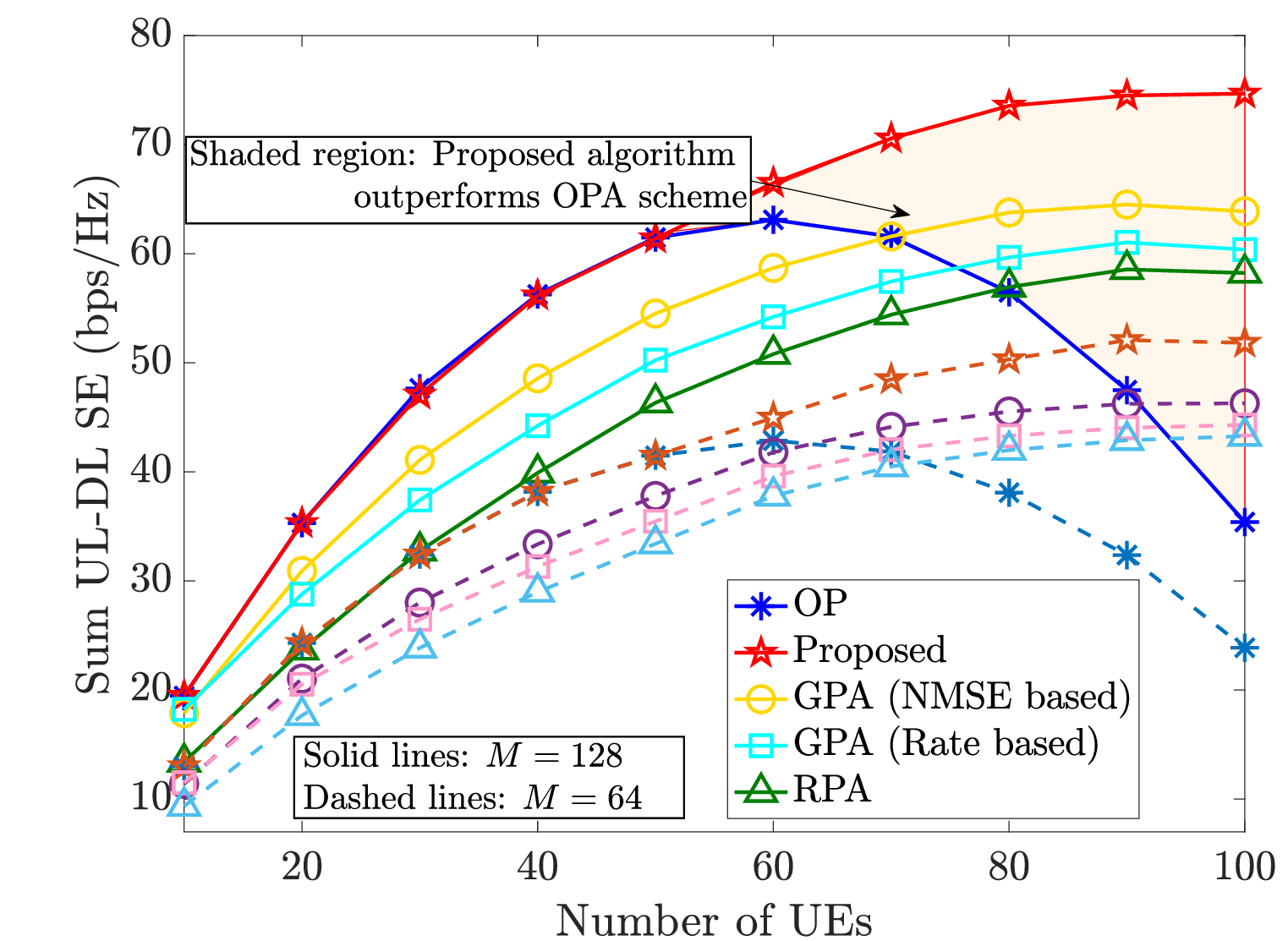}
		\caption{Sum UL-DL SE vs. the number of UEs. (OP: orthogonal pilots across UEs, Proposed: Algorithm~\ref{algo:pilot_allocation}, RPA: Random pilot assignment. NMSE-based and rate-based greedy assignment correspond to~\cite{DTDD_TCoM} and~\cite{cell_free_small_cells}, respectively.)}\label{fig:fig_SE_vs_K}
	\end{figure}
	\begin{figure}
	\centering
		\includegraphics[width=0.85\linewidth]{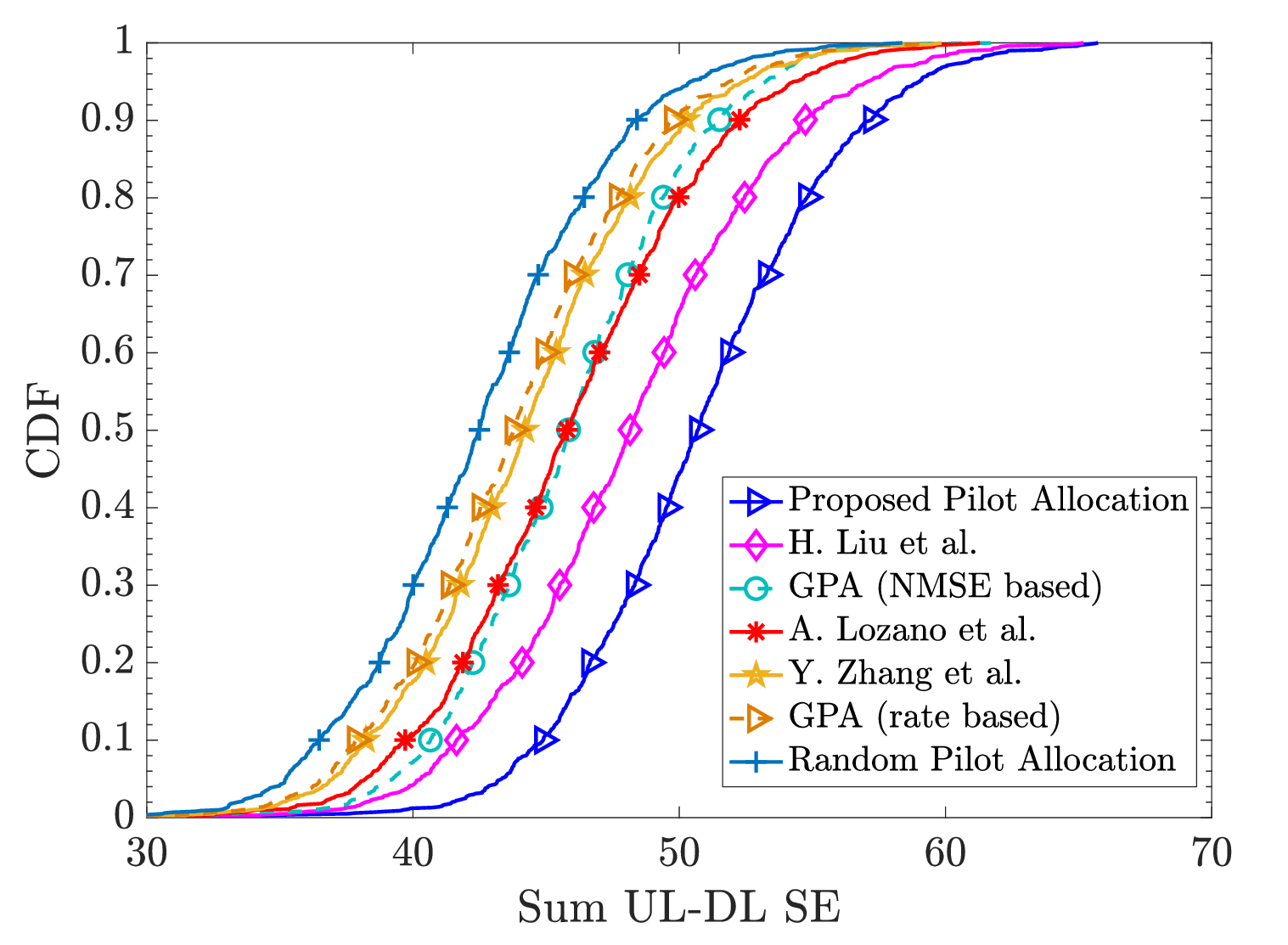}
		\caption{CDF of the sum UL-DL SE with $K=80$ and $M=64$. The legends H. Liu et al., A. Lozano et al., and Y. Zhang et al. correspond to the methods proposed by the authors in~\cite{Heng_Liu_TVT}~\cite{Lozano_Pilot}, and~\cite{Location_Pilot}, respectively.}\label{fig:fig_CDF_pilot_allocations}
	\end{figure}

\begin{figure}[!t]
	\centering
	\includegraphics[width=0.85\linewidth]{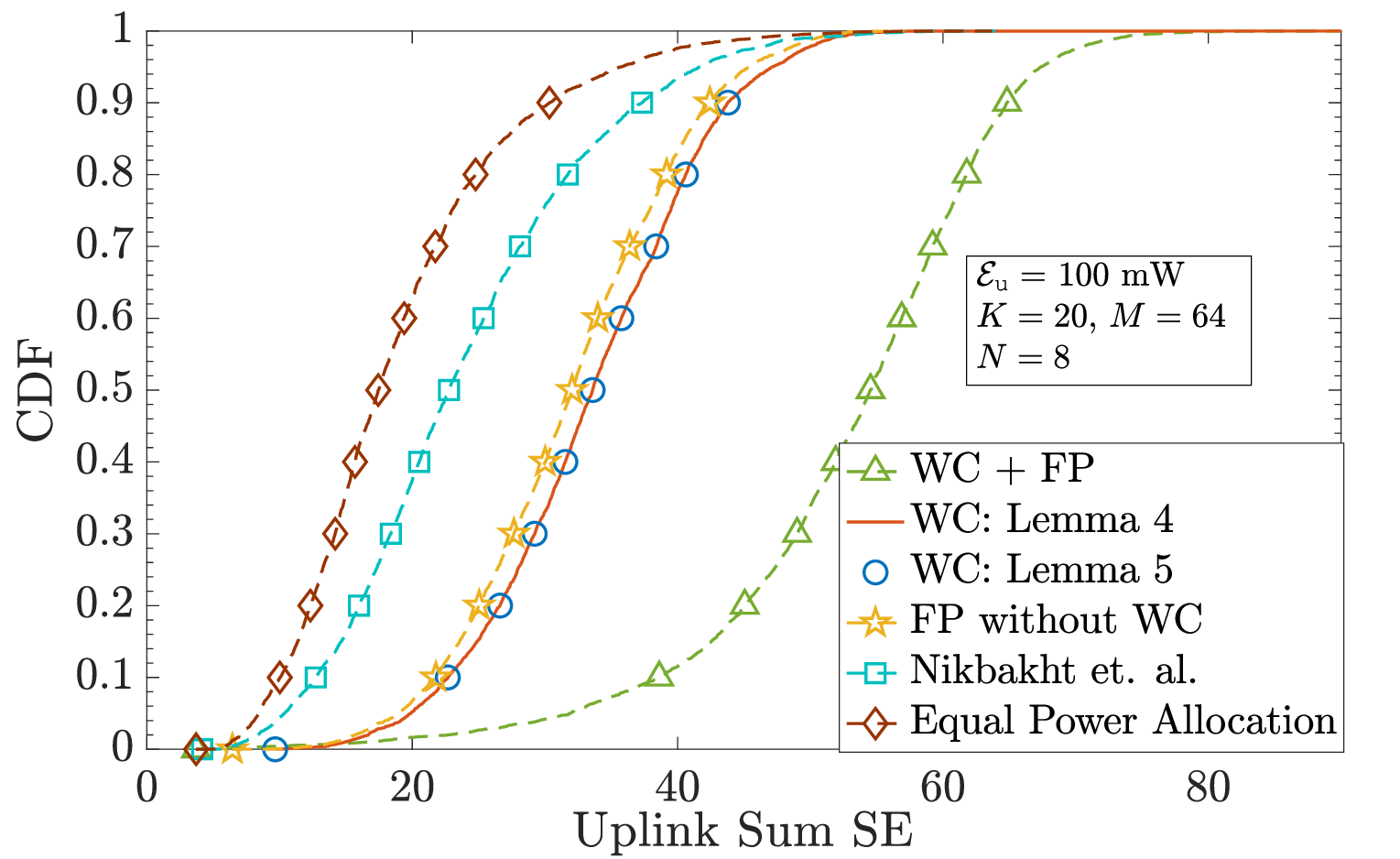}
	\caption{UL sum SE under the proposed power control algorithm and comparison with existing approaches~\cite{Nikbakht, Nikbakht_2}. Optimal weighting at the CPU along with FP-based power control yields the best performance.}\label{fig:UL_power_control}
\end{figure}

\begin{figure}[!t]
	\centering
	\includegraphics[width=\linewidth]{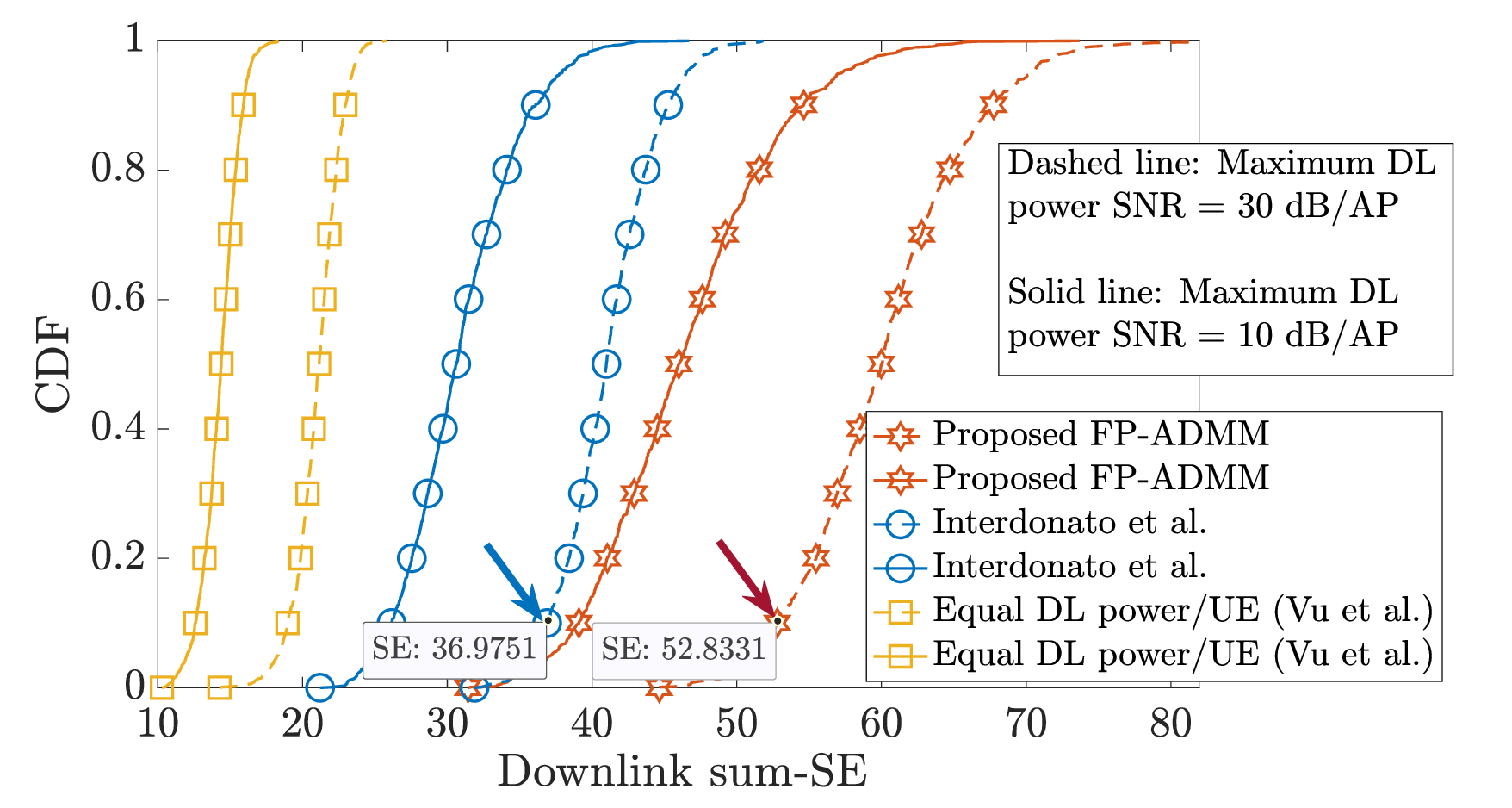}
	\caption{DL sum SE under the proposed power control algorithm and comparison with an existing approach~\cite{Interdonato}. This figure illustrates the improvement in DL sum SE that can be attained via our algorithm.}\label{fig:DL_power_control}
\end{figure}

In Fig.~\ref{fig:DL_power_control}, we compare the CDFs of the achievable DL SE with our proposed FP and ADMM-based algorithm with equal power allocation and the scalable DL power allocation algorithm proposed by Interdonato et al. in~\cite{Interdonato}. We observe almost $43\%$ improvement in the $90\%$ sum DL SE attained by our proposed algorithm compared to the method in~\cite{Interdonato}. Further, compared to the equal power allocation scheme for DL outlined by Vu et al. in~\cite{FD_CF_ICC}, Algorithm~\ref{algo:DL_power_control} procures almost $5$-fold improvement in the SE. Also, as we increase the maximum DL transmit power budget per AP, the DL SE uniformly improves for all schemes.

We now compare the performances of DTDD and FD CF MIMO. In particular, we assume the InAI and IrAI are well suppressed, at $-40$ dB.  Figure~\ref{fig:CDF_DTDD_FD} compares the sum SEs under the two duplexing schemes considering different AP and antenna densities. We observe that an FD system with $(N_{\mathsf{tx}}=N_{\mathsf{rx}}=N, M=64)$ offers only $6\%$ improvement in the $90\%$-likely sum SE compared to a DTDD enabled CF-system with $N$ antennas per $64$ APs. However, the former system has double the antenna density compared to the DTDD CF system. If we consider the same antenna density in the two systems, then the $90\%$-likely sum SE of the DTDD CF system is $21\%$ more than that of the FD system (see $(N_{\mathsf{tx}}=N_{\mathsf{rx}}=N/2, M=64)$). This is because, in DTDD, the APs are scheduled based on the local UL/DL load in its vicinity, and hence, if there is more UL load near to one or a set of APs, those APs are scheduled in UL, which in turn leads to a beamforming gain in UL that scales with $N$. On the other hand, in the FD system, the beamforming gain scales with $N/2$. Recall that although in the FD system, all the APs are FD enabled, APs far away from the UEs contribute minimally to the overall sum SE. Thus, scheduling APs based on the localized traffic load is more beneficial.

In Fig.~\ref{fig:CDF_IrAI}, we illustrate the effect of IrAI on the performance of the FD CF system and contrast it with the DTDD CF system. When the IrAI is $-20$ dB, FD uniformly outperforms the DTDD system even when the FD system has double the antenna density. However, as the IrAI strength increases, the sum SE of FD starts to deteriorate, achieving a $40\%$ lower sum SE than DTDD when IrAI is $10$ dB. Thus, the performance of FD is highly dependent on the level of IrAI suppression, while DTDD completely obviates the need for IrAI suppression and offers similar performance to that of an FD system having double the antenna density and low IrAI. 
	
In Fig.~\ref{fig:fig_SE_vs_SI}, we further inspect the variation of the sum SE over a wide range of IrAI for the FD system and compare the performance to the DTDD system. Even with double antenna density, the FD system can perform very poorly when IrAI strength becomes more pronounced~(see the shaded region).
	\begin{figure*}
		\centering
		\begin{subfigure}{0.47\textwidth}
	\centering
\includegraphics[width=0.9\linewidth]{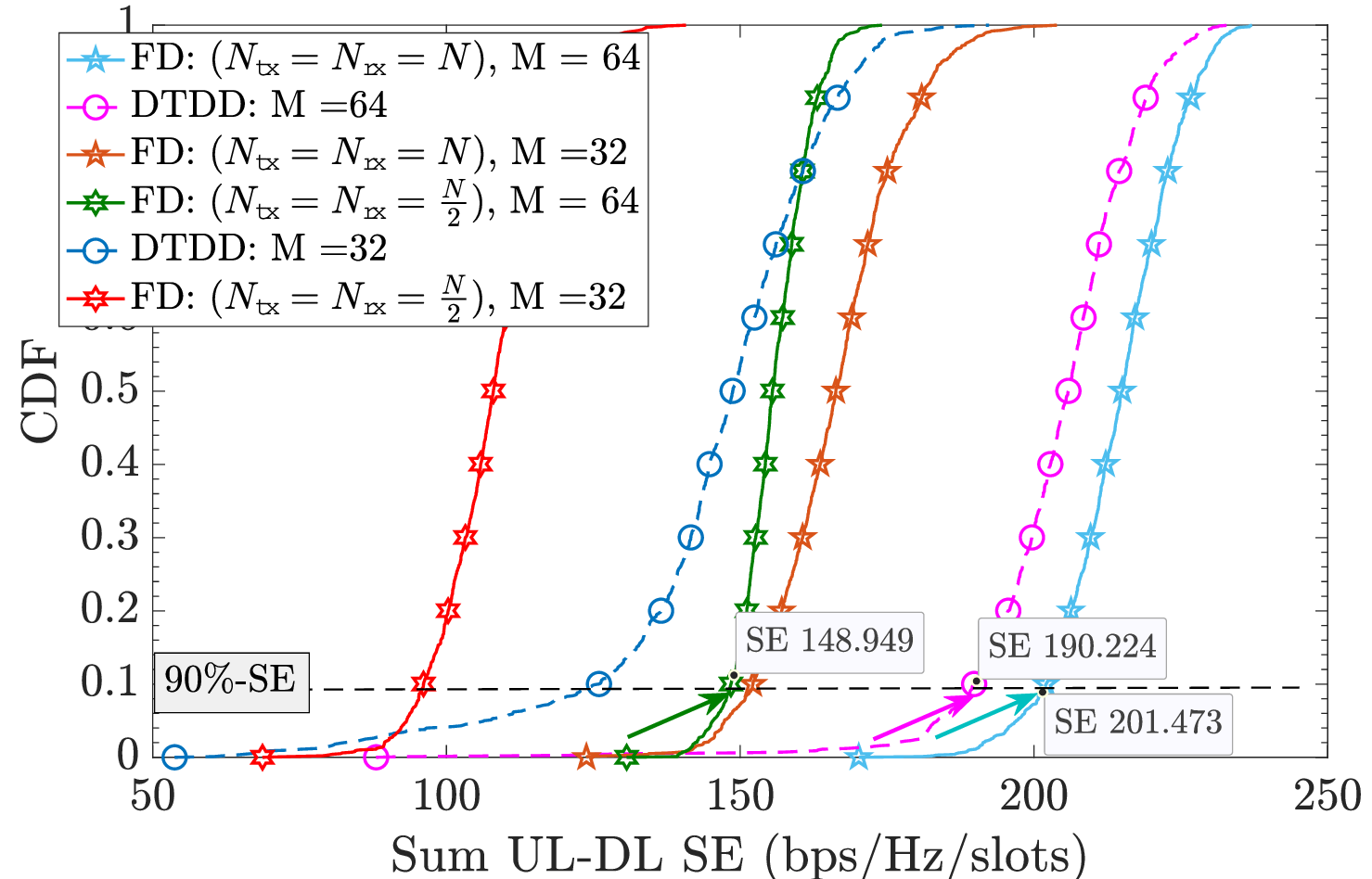}
\caption{Performance comparison of DTDD and FD systems with various antenna and AP densities. We consider $K=40$. Each HD AP is equipped with $N= 8$ antennas. InAI and IrAI strengths are taken as $-40$ dB.}\label{fig:CDF_DTDD_FD}
		\end{subfigure}\hfill
			\begin{subfigure}{0.47\textwidth}
				\centering
			\includegraphics[width=0.9\linewidth]{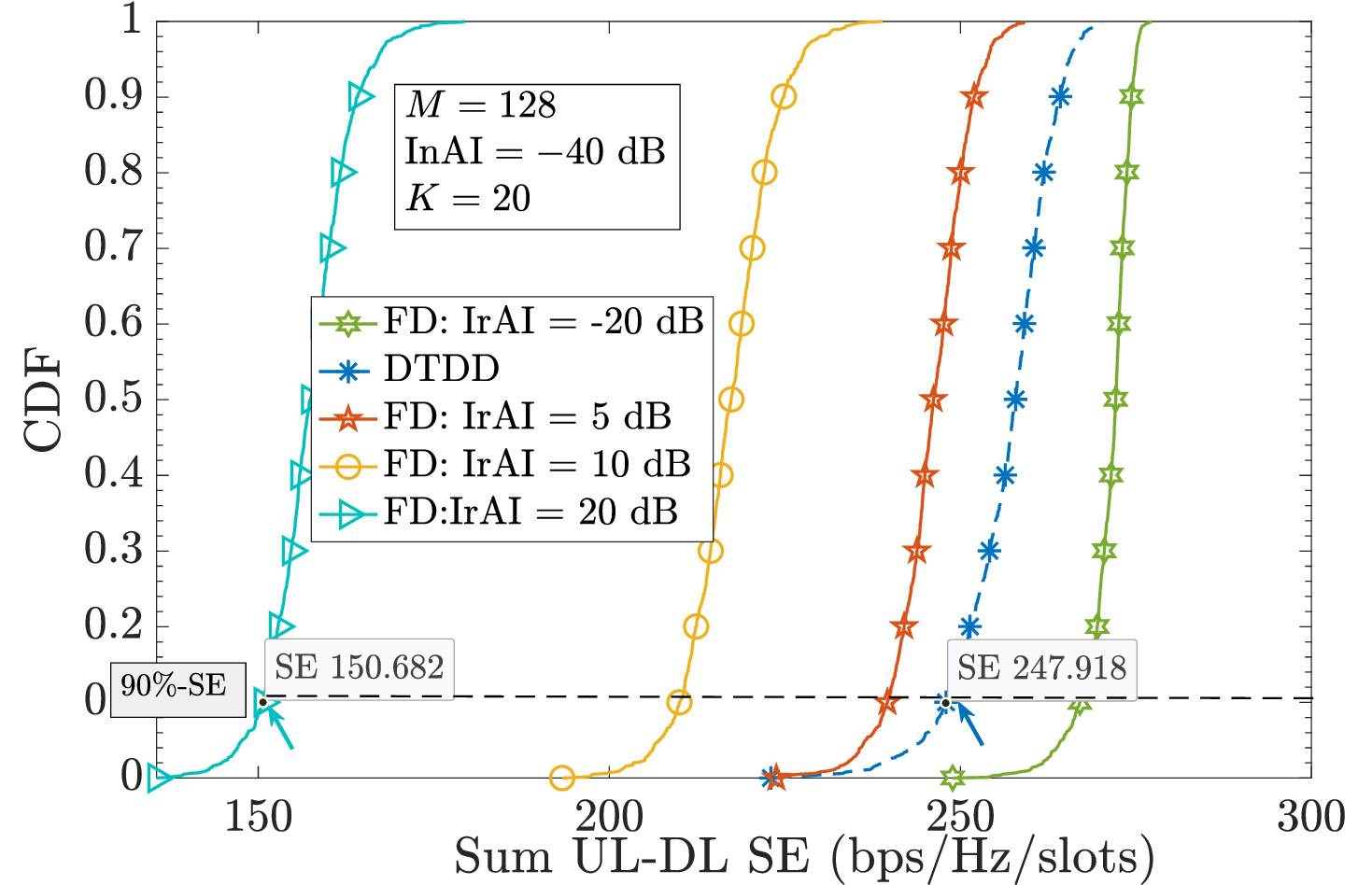}
			\caption{Effect of intra-AP interference~(IrAI) on the performance of the FD system while InAI is maintained the same for both DTDD and FD systems.}\label{fig:CDF_IrAI}
		\end{subfigure}
		\caption{Performance of DTDD compared to FD system for different antenna and AP densities with ZF combiner and precoder.}\label{fig:figs_DTDD_FD}
	\end{figure*}
	
\begin{figure}[!]
			\centering
			\includegraphics[width=0.9\linewidth]{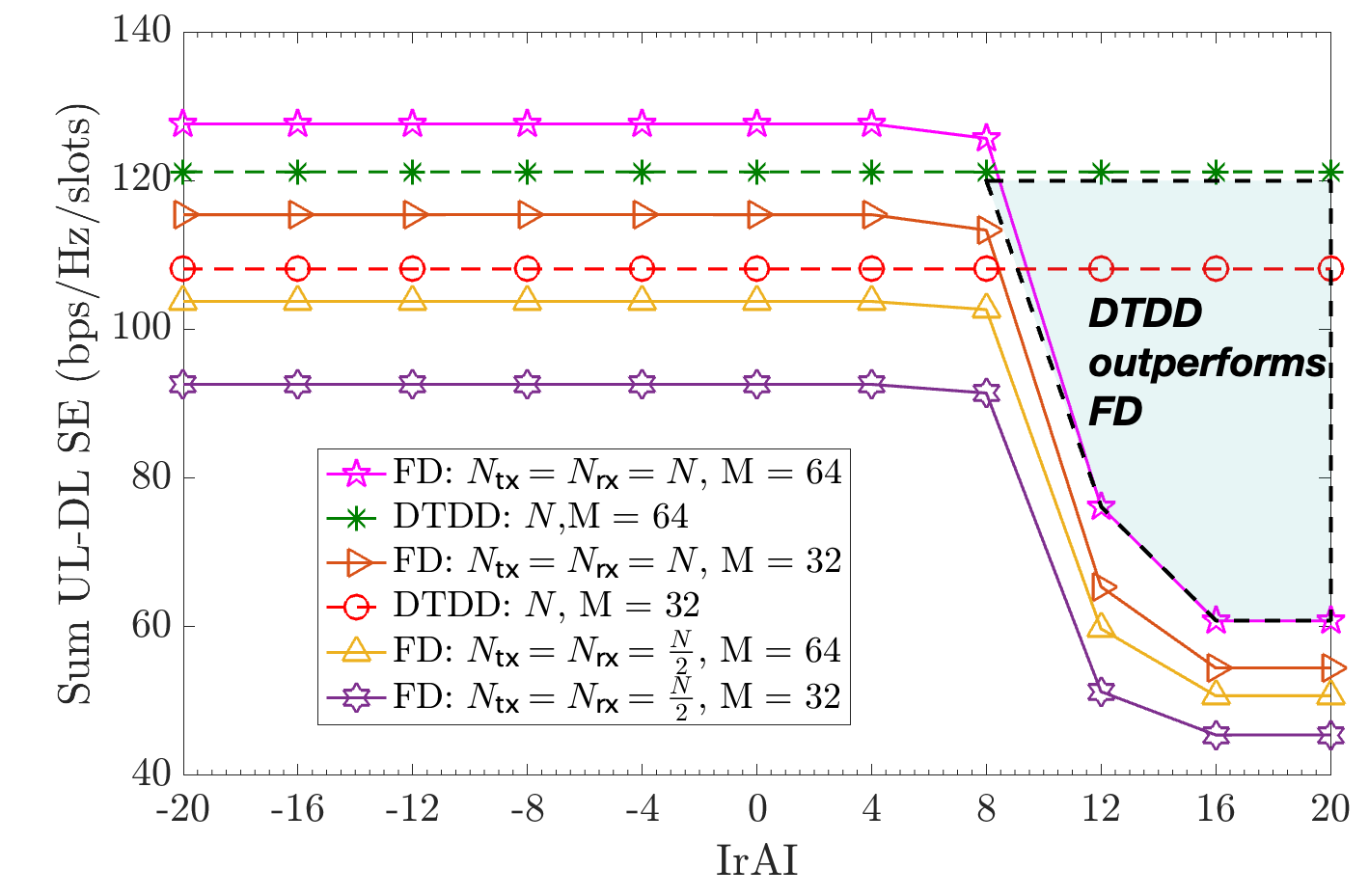}
			\caption{Sum UL-DL SE as a function of IrAI. DTDD can outperform FD even though the latter has double the antenna density}\label{fig:fig_SE_vs_SI}
		\end{figure}

	In Fig.~\ref{fig:InAI}, we illustrate the effect of InAI on the sum SE. An FD system with $(M = 64, N_{\mathsf{tx}} = N_{\mathsf{rx}} = N=8)$ with  IrAI $-20$ dB outperforms DTDD with $(M = 64, N = 8)$ (i.e., half the antenna density compared to the FD system) when the InAI is no more than $\approx 11$ dB above the noise floor. However, beyond an InAI of $11$~dB, the sum SE of the FD system degrades compared to the DTDD system. This is because, in the FD system, all APs cause InAI, while in DTDD, only the DL-scheduled APs cause inter-AP interference.
	
	\begin{figure}[!]
		\centering
		\includegraphics[width=0.9\linewidth]{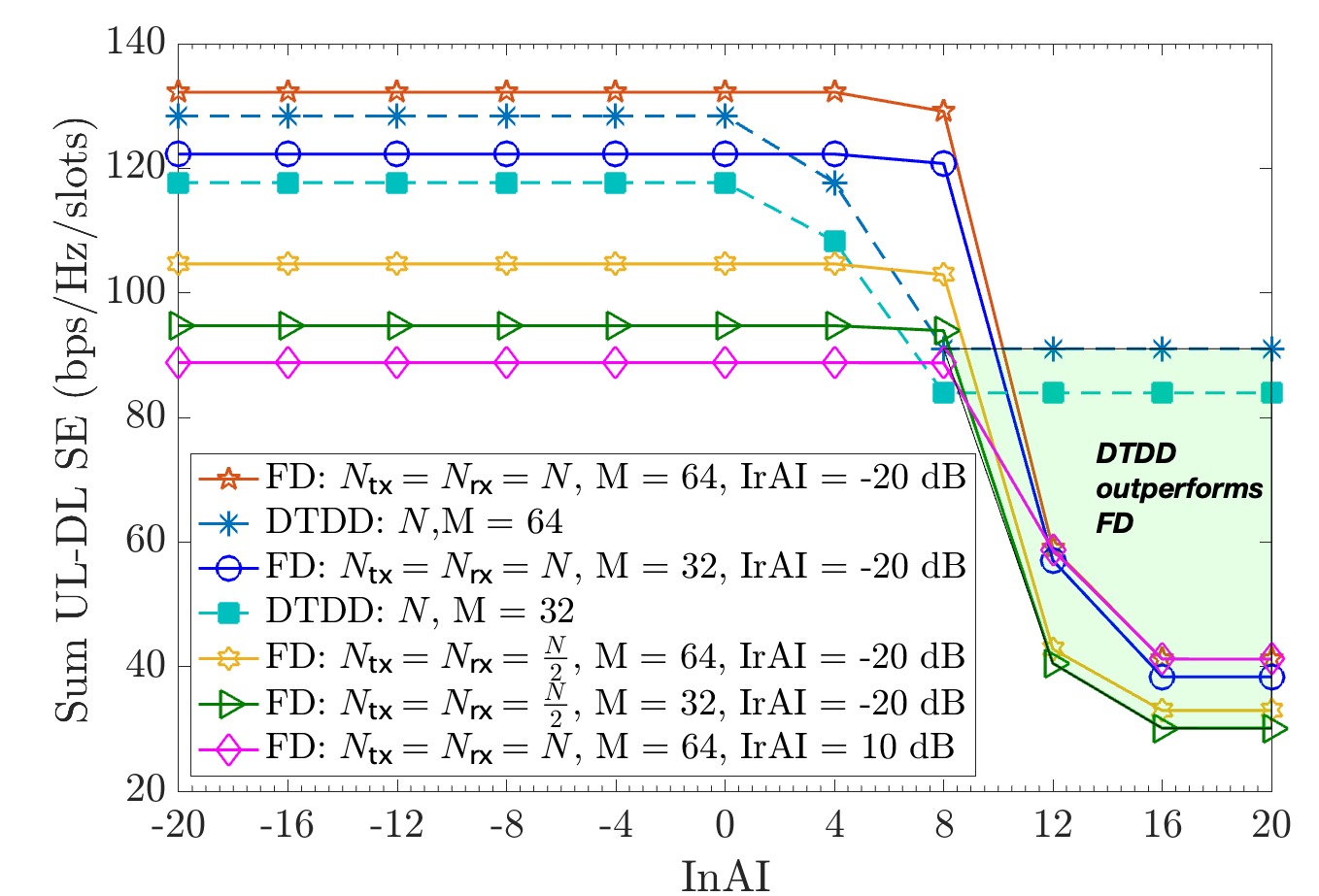}
		\caption{Effects of inter-AP interference~(InAI) on the sum UL-DL SE. We observe that DTDD is more resilient to InAI.}\label{fig:InAI}
	\end{figure}
	
	Finally, in Fig.~\ref{fig:fig_CDF_MMSE_SE}, we plot the CDF of sum UL-DL SE of the DTDD and FD systems with MMSE combining in the UL and RZF precoding in the DL, and illustrate the effects of both power control and IrAI on the sum UL-DL SE. We obtain substantial benefits by applying the proposed power control algorithms compared to equal power allocation. 
	This illustrates the applicability of the algorithms developed here to different precoder and combining schemes. 
	Also, with similar antenna density, DTDD uniformly outperforms FD, even with MMSE and RZF. This is because of the additional degrees of freedom DTDD offers in terms of UL and DL AP scheduling and, consequently, mitigating the effects of InAI better than the FD system.
	
	\begin{figure}
		\centering
		\includegraphics[width=0.9\linewidth]{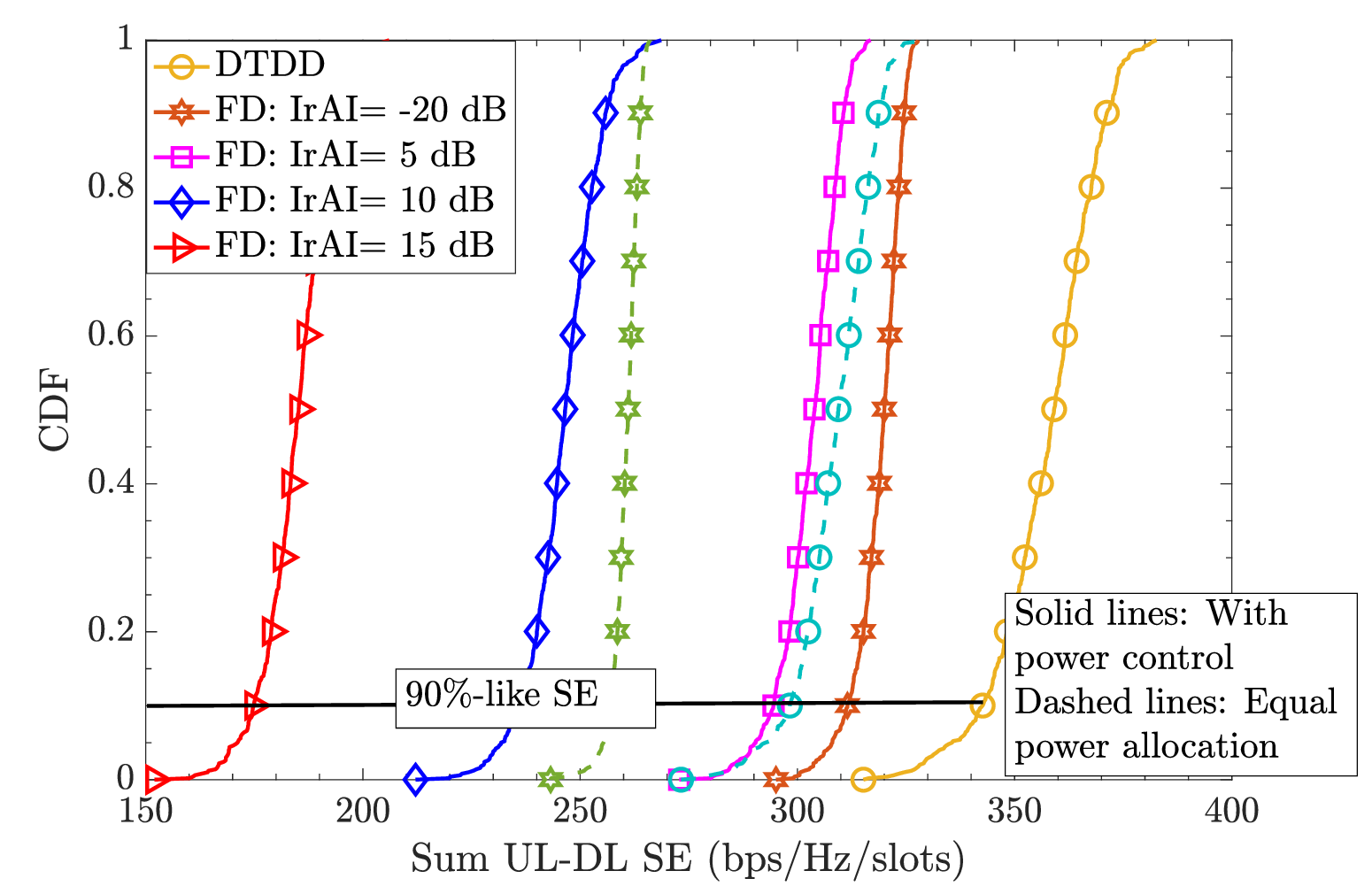}
		\caption{Effect of IrAI on the sum UL-DL SE with MMSE and RZF. InAI strength is taken as $-40$ dB, and we consider $K=20, M= 64, N =8, N_{\mathsf{tx}}=N_{rx}=N/2$~(i.e., similar antenna density.)}\label{fig:fig_CDF_MMSE_SE}
	\end{figure}

\section{Conclusion}
In this paper, we presented a comparative study of DTDD and FD in CF systems. 
First, we developed a novel graph coloring-based pilot allocation algorithm that ensures no contamination in the received signals at the APs in the vicinity of every UE while minimizing the required pilot length. Then, we optimized the UL and DL power allocation and AP scheduling for DTDD, to maximize the sum UL-DL SE. We solved this NP-hard and non-convex problem by decoupling it into AP-scheduling, UL, and DL power allocation sub-problems. We developed FP-based UL/DL power allocation algorithms and proved the convergence of the sub-problems to local optima. Further, we provided closed-form update equations using the Lagrange dual transform and ADMM for the sub-problems, making them easy to implement. We numerically illustrated the superiority of the algorithms over existing methods. Finally, we saw that DTDD outperforms FD when the two systems have a similar antenna density. This happens because DTDD can schedule the APs in UL or DL based on the localized traffic load and achieve better array gain for a given antenna density and InAI suppression. Thus, we conclude that although both DTDD and FD enable the CF system to serve UL and DL UEs concurrently, DTDD is preferable because it can meet and even outperform FD without requiring the use of IrAI cancellation hardware. Fairness guarantees under the two duplexing schemes are a good direction for future work.
%% Bibliography%%
\ifCLASSOPTIONcaptionsoff
\newpage
\fi
\bibliographystyle{IEEEtran.bst}
\bibliography{CF.bib}

\end{document}